\documentclass[twocolumn,floatfix,aps,pra,superscriptaddress]{revtex4-2}
\usepackage[dvipsnames]{xcolor}
\usepackage{graphicx}
\usepackage{float}
\usepackage{dcolumn}
\usepackage{bm}
\usepackage{amsfonts,amssymb,amscd,amsmath,amsthm}
\usepackage{mathrsfs}
\usepackage[ruled,linesnumbered]{algorithm2e}
\usepackage{algpseudocode}
\usepackage{enumerate}
\usepackage{epsfig}
\usepackage{physics}
\usepackage{makecell}
\usepackage{bbm}
\usepackage{svg}
\usepackage{tikz}
\usepackage[colorlinks = true]{hyperref}
\hypersetup{colorlinks=true, linkcolor=blue, citecolor=blue, urlcolor=black}
\usepackage{amsmath, amsthm, amsfonts, amssymb, amsbsy, mathtools, xcolor, bm, bbm, graphicx, changepage, cleveref}

\usepackage{physics}
\usepackage{epstopdf}
\usepackage{framed}
\usepackage{multirow}
\usepackage{color}
\usepackage{longtable}
\usepackage{comment}
\usepackage{qcircuit}
\usepackage{faktor}
\usepackage[normalem]{ulem}
\usepackage[caption=false, justification=centerlast]{subfig}

\makeatletter
\newtheorem*{rep@theorem}{\rep@title}
\newcommand{\newreptheorem}[2]{%
\newenvironment{rep#1}[1]{%
 \def\rep@title{#2 \ref{##1}}%
 \begin{rep@theorem}}%
 {\end{rep@theorem}}}
\makeatother

\newtheorem{remark}{Remark}
\newtheorem{theorem}{Theorem}
\newtheorem{proposition}{Proposition}
\newtheorem{observation}{Observation}
\newtheorem{lemma}{Lemma}

\newtheorem*{assumption*}{Assumption}
\newtheorem{corollary}{Corollary}
\newtheorem{definition}{Definition}

\usepackage{lineno}

\renewcommand{\d}{\mathrm{d}}

\newcommand{\ii}{\textup{i}}

\begin{document}

\title{Observable-Driven Speed-ups in Quantum Simulations}
\author{Wenjun Yu}
\author{Jue Xu}
\author{Qi Zhao}
\email{zhaoqi@cs.hku.hk}
\affiliation{QICI Quantum Information and Computation Initiative, Department of Computer Science, The University of Hong Kong, Pokfulam Road, Hong Kong}
\begin{abstract}
    As quantum technology advances, quantum simulation becomes increasingly promising, with significant implications for quantum many-body physics and quantum chemistry. 
    Despite being one of the most accessible simulation methods, the product formula encounters challenges due to the pessimistic gate count estimation.
    In this work, we elucidate how observable knowledge can accelerate quantum simulations.
    By focusing on specific families of observables, we reduce product-formula simulation errors and gate counts in both short-time and arbitrary-time scenarios. 
    For short-time simulations, we deliberately design and tailor product formulas to achieve size-independent errors for local and certain global observables. 
    In arbitrary-time simulations, 
    we reveal that Pauli-summation structured observables generally
    reduce average errors.
    Specifically, we obtain quadratic error reductions proportional to the number of summands for observables with evenly distributed Pauli coefficients. 
    Our advanced error analyses, supported by numerical studies, indicate improved gate count estimation. 
    We anticipate that the explored speed-ups can pave the way for efficiently realizing quantum simulations and demonstrating advantages on near-term quantum devices.
\end{abstract}

\maketitle

Quantum simulation, aimed at mimicking the temporal evolution of quantum systems, stands as one of the most promising applications of quantum computers~\cite{feynman1982simulating}. 
It serves as a reliable tool for investigating quantum many-body physics~\cite{noh2016quantum,schreiber2015observation,randall2021many,su2023observation,tran2019locality,heylQuantumLocalizationBounds2019} and quantum field theory~\cite{jordanQuantumAlgorithmsQuantum2012, jordanBQPcompletenessScatteringScalar2018}.
Moreover, it holds profound implications in various fields that seek to uncover the dynamics of systems, including materials science and quantum chemistry~\cite{babbushLowDepthQuantumSimulation2018, mcardleQuantumComputationalChemistry2020, leeEvaluatingEvidenceExponential2023}.

Since Lloyd's pioneering proposal for digital quantum simulation~\cite{lloyd1996universal}, significant efforts have been devoted to developing more efficient simulation algorithms, including families of product formulas~\cite{aharonov2003adiabatic,childs2004quantum,berry2007efficient,somma2016trotter}, linear combination of unitaries (LCU)~\cite{childsHamiltonianSimulationUsing2012,berry2014exponential,Berry2015,berry2015hamiltonian,haah2021quantum}, and quantum signal processing (QSP)~\cite{Low2017optimal,low2019hamiltonian}. 
Both LCU and QSP approaches achieve linear gate counts in terms of simulation time, scaling optimally as inferred in~\cite{berry2007efficient}.
However, these methods demand numerous ancillary qubits for controlling or block encoding, requiring high-quality multi-qubit gates, which are challenging in near-term quantum computing.
In contrast, product-formula simulation involves only short-range quantum gates and empirically requires fewer resources~\cite{childs2018toward}, making it an accessible choice for realizing simulations in experiments~\cite{Brown2006,lanyon2011universal,lv2018quantum,jafferisTraversableWormholeDynamics2022,kimEvidenceUtilityQuantum2023} and achieving quantum advantages in the future.

While existing product-formula methods are nearly optimal for certain Hamiltonians~\cite{childs2019nearly,childs2021theory}, the associated error bounds typically emphasize the worst-case scenarios for states and measurements.
This pessimism renders the resource requirements far surpassing current experimental capabilities~\cite{childs2018toward}.
However, practical settings, especially measurements, 
are often chosen deliberately, enabling further improvements in error analyses and resource estimations as shown in Fig.~\ref{fig:idea}(a).
Specific measurement structures are prevalently adopted in simulation due to their simplicity or direct relevance to the underlying physics, such as local observables in studying system scramblings~\cite{li2017measuring,garttner2017measuring,landsmanVerifiedQuantumInformation2019,green2022experimental}, and summation observables in magnetization and correlation functions~\cite{simon2011quantum,martinezRealtimeDynamicsLattice2016, monroeProgrammableQuantumSimulations2021, jafferisTraversableWormholeDynamics2022, kimEvidenceUtilityQuantum2023}.
Recent work~\cite{childs2021theory,heylQuantumLocalizationBounds2019} has depicted these observable-driven advantages in product-formula simulation for some special cases.
Nevertheless, there is a missing of a thorough exploration of quantum simulation speed-ups by utilizing specific knowledge of observables.

\begin{figure*}[tb]
    \centering
    \includegraphics[width=1.95\columnwidth]{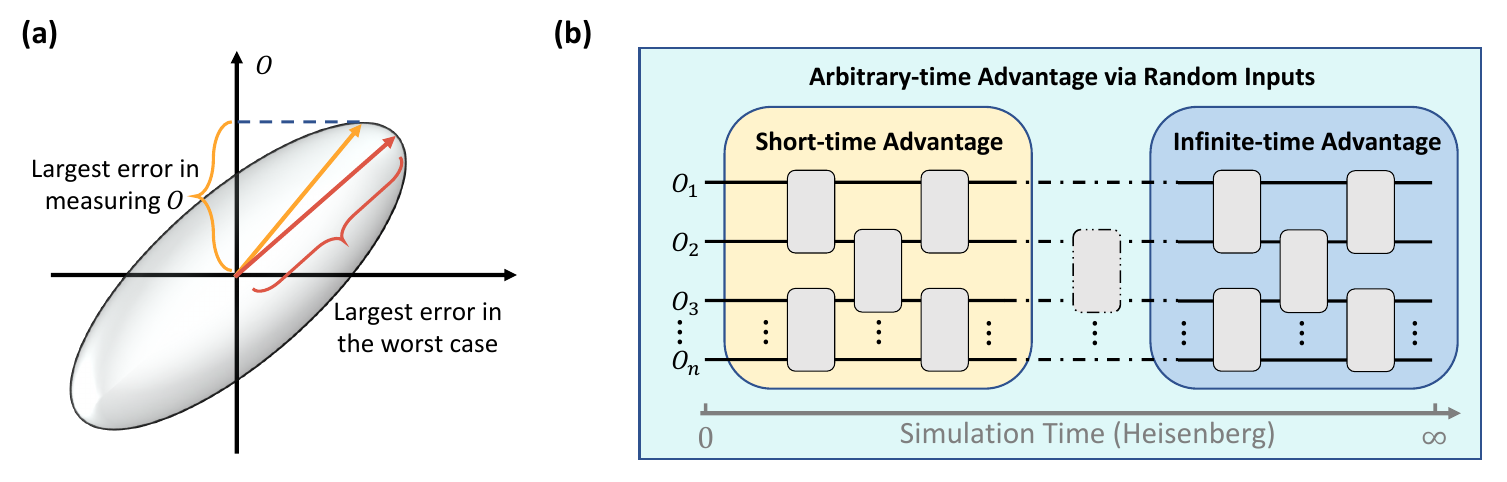}
    \caption{Illustrative diagrams of observable-driven simulation speed-ups.
    \textbf{(a)} Schematic illustrating the benefits of utilizing observable knowledge to analyze simulation errors.
   The operator-norm distance between unitaries captures the largest error from the worst-case scenario among states and observables (red bracket). Using observable knowledge can significantly reduce simulation errors (orange bracket).
    \textbf{(b)} Schematic of advantages at different simulation times within Heisenberg's picture.
    In this work, short-time advantages, previously analyzed in~\cite{childs2021theory} for local observables, are extended to more general observables.
    Infinite-time advantages of local observables were observed in~\cite{heylQuantumLocalizationBounds2019}.
    Beyond these two special cases, we also explore the observable-driven advantages in the general arbitrary-time simulation by examining average performances with random input states.
    }
    \label{fig:idea}
\end{figure*}

In this work, we comprehensively analyze product-formula simulations by leveraging observable knowledge, uncovering its advantages for both short- and arbitrary-time simulation.
For short-time simulations, we design specific product formulas for different observables.
Notably, our methods and analyses extend the advantage beyond local observables in~\cite{childs2021theory} to certain global-observable simulations. 
For specific Hamiltonians, our analyses suggest size-independent simulation errors for both local and global observables.
For arbitrary simulation times, we evaluate the state-independent observable-driven speed-ups by averaging over random input states.
We show that observables consisting of a sum of Pauli operators generally reduce average errors compared to previous analyses without observable knowledge~\cite{zhao2022hamiltonian,chen2024average}, resulting in the typical quadratic reductions with the number of summands.
In both scenarios, improved error analyses imply better gate counts than previously expected, which we anticipate will facilitate the practical advantages of quantum simulations.
We validate these theoretical results and demonstrate speed-ups in practical application through our numerical studies.

Our short-time simulation analyses reveal that the locality of observables generally provides a speed-up.
For an arbitrary local observable, we delineate its support expansion in a generic product-formula circuit within Heisenberg's picture.
Based on this, we propose an optimally decomposed and permuted product-formula circuit to limit the expanding support at the slowest rate.
Removing all irrelevant gates outside the support results in size-independent simulation errors and gate counts for geometrically local and power-law Hamiltonians.
This method achieves lower errors than the worst-case analysis in~\cite{childs2021theory} for sublinear simulation times related to the system size. 
Locality from global observables consisting of a sum of local observables also leads to short-time advantages.
To view this, we design a product formula that simultaneously suppresses the expanding supports of all local summand observables, ensuring size-independent errors through a similar support analysis.
These improved error analyses suggest more efficient simulation, benefiting practical applications such as observing dynamical quantum phase transitions~\cite{heylDynamicalQuantumPhase2013, heylDynamicalQuantumPhase2018, denicolaEntanglementViewDynamical2021, halimehLocalMeasuresDynamical2021} and preparing entangled states \cite{PhysRevX.8.021012,zhou2022scheme}.

Considering our short-time analyses and the previous infinite-time results~\cite{heylQuantumLocalizationBounds2019}, a natural question arises as in Fig.~\ref{fig:idea}(b): do observables offer advantages for arbitrary-time simulations? 
Although a general analysis remains challenging, we shift our focus to random-input simulations to isolate state effects and identify speed-ups from observable knowledge.
Using observable knowledge, our new random-input error bound is linear with the observable's normalized Schatten 2-norm, significantly improved upon previous analyses that relied on the operator norm~\cite{zhao2022hamiltonian}.
Our error analysis generally guarantees error reductions for observables consisting of a sum of multiple Pauli operators (Pauli-summation observables). 
Specifically, we prove quadratic reductions in average errors proportional to the number of summands for observables with evenly distributed Pauli coefficients. 
Moreover, the efficiency of computing 2-norms makes our analysis accessible even for complicated observables and Hamiltonians.
Therefore, for more complicated dynamics and measurements, such as tasks in molecular systems, our analysis generally leads to more significant speed-ups.

\section*{Results}\label{sec:re}

\subsection{Observable Advantages for Short-Time Simulations}
In this section, we investigate the advantages of short-time product-formula simulation on specific observables.
We denote the support of an arbitrary operator $A$ as $S(A)$.
Throughout the remainder of the paper, ``local operator" refers to the operators acting nontrivially on only a constant number of qubits, \emph{i.e.}, $|S(A)|=\order{1}$.
For both a local observable and a global observable composed of a sum of local operators, we introduce optimal product formulas that suppress the expansion of observable supports, known as the \emph{light cone}, resulting in size-independent simulation errors. 
In this sense, our findings indicate that the locality of observables generally endows better errors for short-time simulation.

To illustrate, we present the general form of a product formula for simulating a generic $H$,
\begin{gather}\label{eq:product}
    \mathscr{S}(t)=\prod_{\upsilon=1}^\Upsilon\prod_{\gamma=1}^{\Gamma}\mathrm{e}^{\ii ta_{(\upsilon,\gamma)}H_{\pi_\upsilon(\gamma)}},
\end{gather}
with decomposition $H=\sum_{\gamma=1}^\Gamma H_\gamma$, and permutation $\pi_\upsilon$, collectively denoted by \emph{configuration}.
Here, $\Upsilon$ denotes the stage number, and $\{a_{(\upsilon,\gamma)}\}$ are some real coefficients.
For a long time $t$, the conventional approach for simulations is to divide the time into $r$ small steps and use $\mathscr{S}^r(t/r)$ to approximate the ideal evolution. 

\paragraph*{Local Observables.---}To explore the advantage of a local observable $O$, we need a proper configuration to decide the product formula.
We introduce the corresponding \emph{interactive decomposition} and \emph{edge sets} regarding its support $S$ (shorthand for $S(O)$ due to its frequent use).
Given the Pauli decomposition $H=\sum\nolimits_{\alpha\in{\sf P}^n}s_\alpha P_\alpha$, we define
\begin{align}\label{eq:edge_set}
\begin{split}
    H_0^S\coloneqq\sum_{\alpha:S(P_\alpha)\subseteq S}s_\alpha P_\alpha,&\ \ \ E_0^S\coloneqq S,\\
    H^{S}_k\coloneqq \sum_{\substack{\alpha:P_\alpha\notin H^{S}_{k-1}\\S(P_\alpha)\cap E^{S}_{k-1}\neq\emptyset}}s_\alpha P_\alpha,&\ 
    E^{S}_k\coloneqq S(H^{S}_k)-E^{S}_{k-1},
\end{split}
\end{align}
where $P_\alpha\notin H^{S}_{k-1}$ indicates that $P_\alpha$ is not included in $H^{S}_{k-1}$.
Figure~\ref{fig:edge}(a) exhibits a schematic of $\{H_k^S\}$ and $\{E_k^S\}$.
\begin{figure*}[tb]
    \centering
    \includegraphics[width=1.96\columnwidth]{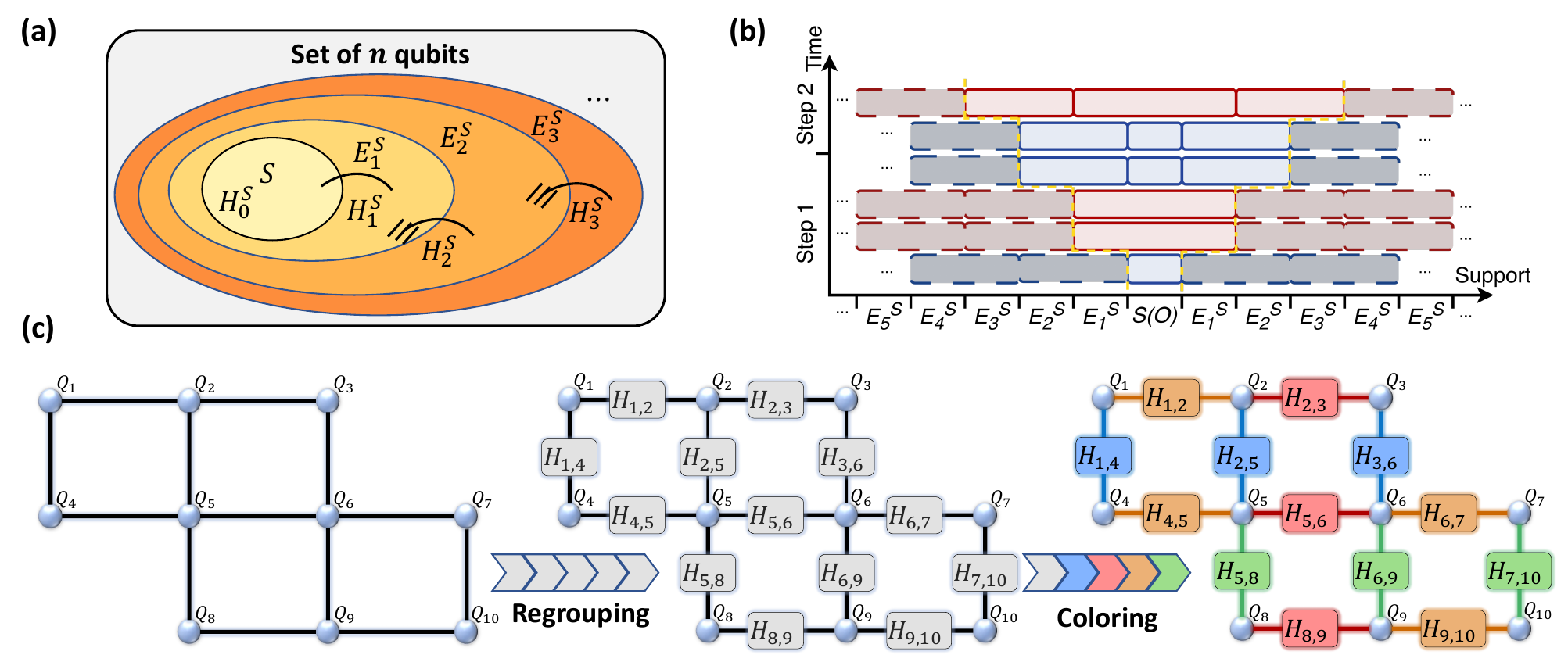}
    \caption{
    Illustrative diagrams of designing product formulas in different scenarios for short-time simulations.
        \textbf{(a)} A general depiction of the edge-set partition and the corresponding interactive decomposition regarding $S$.
        Layers signify edge sets $\{E^{S}_{k}\}_k$, with the line between $E^{S}_{k-1}$ and $E^{S}_{k}$ representing sub-Hamiltonian interactions $H^{S}_{k}$ for integer $k\geq1$.
        \textbf{(b)} Visualizing the support expansion of the operator $O(t)$ in the second-order Trotter formula circuit.
        Here, we adopt the interactive decomposition and even-odd permutation. 
    Each colored block represents the matrix exponential, with 
    blue denoting sub-Hamiltonians with even subscripts and red denoting odd ones.
    Shaded blocks indicate unitaries outside the expanding support of $O(t)$.
    We only need to implement the effective (bright) blocks, as outlined in Alg.~\ref{alg:rpf}.
    \textbf{(c)} Illustration of the regrouping and coloring procedure for a two-dimensional nearest-neighbor lattice Hamiltonian.
    The regrouping preserves nearest-neighbor interactions, with graph edges corresponding to lattice edges.
    Edges are colored based on their parities along the vertical and horizontal axes. 
    }
    \label{fig:edge}
\end{figure*}

We find that edge sets determine lower bounds for the expanding support of $O(t)$ evolved under product formulas. 
This lower bound can be attained by employing the formula with the interactive decomposition $\{H_k^S\}_k$ and an even-odd permutation defined as:
\begin{align}\label{eq:eo}
    \pi^{\text{eo}}_\upsilon(0,1,2,3,\cdots)=
    \begin{cases}
        0,2,\cdots, 1,3,\cdots & \text{$\upsilon$ is odd}\\
         1,3,\cdots,0,2,\cdots & \text{$\upsilon$ is even}
    \end{cases}.
\end{align}
To elucidate this claim, we define an operation $\uplus$ between a unitary $U$ and the support $S$ of $O$ as follows:
\begin{gather}\label{eq:oplus}
    U\uplus S\coloneqq \begin{cases}
    S\cup S(U)& S\cap S(U)\neq\emptyset\\
    S& \text{otherwise}
    \end{cases}.
\end{gather}
This operation estimates the largest possible support of $UOU^\dag$.
We then introduce a lemma regarding the expanding support of an evolved operator, with the proof sketched in Methods and formally presented in Appendix~\ref{sec:append-short}.
\begin{lemma}[Optimal Configuration]\label{lm:optimal}
Consider an operator $O$ with support $S$, an $n$-qubit Hamiltonian operator $H$, and a $\Upsilon$-stage product formula as in Eq.~\eqref{eq:product} with an arbitrary decomposition and permutation.
The expansion of support estimation is lower bounded as 
\begin{gather*}\label{eq:worst-case-prod}
\left(\biguplus_{\upsilon=1}^{\Upsilon}\biguplus_{\gamma=1}^{\Gamma}\mathrm{e}^{\ii ta_{(\upsilon,\gamma)}H_{\pi_\upsilon(\gamma)}}\uplus S\right)\supseteq \bigcup_{k=0}^\Upsilon E_k^S,
\end{gather*}
where equality holds with the decomposition in Eq.~\eqref{eq:edge_set} and permutation in Eq.~\eqref{eq:eo}.
\end{lemma}
Coincidentally, this optimal configuration satisfies the back-and-forth permutation constraint of the celebrated Suzuki-Trotter formula.
Adopting $\{a_{(\upsilon,\gamma)}\}$ and $\Upsilon$ from~\cite{suzuki1991general}, we construct a standard $p$th-order Suzuki-Trotter formula from this configuration with $p$ being a positive even integer.
This formula includes some unitaries disjoint with the light cone of $O(t)$, which are irrelevant to the simulation since they commute with the evolved observable.
Removing these irrelevant unitaries preserves the simulation but simplifies it.
Figure~\ref{fig:edge} (b) depicts the light cone and the corresponding reduced product formula, implying a smaller gate count.
We summarize this reduced circuit in Alg.~\ref{alg:rpf}.
In addition to the gate reduction, erroneous terms outside the light cone are irrelevant to the simulation. 
Consequently, the new error bound is linear with the ``width" of the light cone and independent of the system size $n$, as stated in Thm.~\ref{thm:single}.

\begin{theorem}[Local-Observable Error]\label{thm:single}
Consider a local observable $O$ with support $S$.
Suppose the $n$-qubit $H$ is $\ell$-local with a constant $\ell$ and 
has bounded interaction per qubit. 
With light-cone width $\textstyle w_r\coloneqq \sum\nolimits_{k=0}^{r\Upsilon+1}\|H_k^S\|_1$, the simulation error of $O$ by an $r$-step $p$th-order $U$ from Alg.~\ref{alg:rpf} is bounded by $$\|\mathrm{e}^{\ii Ht}O\mathrm{e}^{-\ii Ht}-U OU^\dagger\|=\order{\frac{\|O\|w_rt^{p+1}}{r^{p}}}.$$ 
\end{theorem}
\noindent We sketch the proof in Methods and defer the detailed proof to Appendix~\ref{sec:append-short}.
Interaction per qubit denotes the strength of all interactions overlapping with a fixed qubit.
The 1-norm $\|\cdot\|_1$ sums the absolute values of all Pauli coefficients of an operator.
These constraints are formally defined in Appendix~\ref{sec:append-short}.
The width $w_r$ must be smaller than the norm of the entire Hamiltonian $\|H\|_1$ to ensure the advantage of this error bound over the worst-case analysis $\order{\|O\|\|H\|_1t^{p+1}r^{-p}}$ in~\cite{childs2021theory}. 
Therefore, a short simulation time is necessary.
For example, $t=o(n)$ suffices in one-dimensional lattice models.
Nevertheless, this short time can still be polynomial in $n$, extending well beyond the capabilities of classical simulation~\cite{wildClassicalSimulationShortTime2023,bravyiClassicalSimulationPeaked2023,bravyiClassicalAlgorithmsQuantum2021}.

\begin{algorithm}[t]
\caption{Reduced Product Formula}\label{alg:rpf}
    \SetKwInOut{Input}{input}\SetKwInOut{Output}{output}
    \Input{Observable $O$ with support $S$; Hamiltonian $H$; Evolving time $t$; Step number $r$; Order $p$}
    \Output{Unitary $U$}
    Decompose $H$ into $H_0^S,\cdots,H_{\Gamma_0}^S$ according to Eq.~\eqref{eq:edge_set}\;
    Adopt Suzuki-Trotter coefficients $a_{(\upsilon,\gamma)}$ from~\cite{suzuki1991general}\;
    $\Upsilon\leftarrow 2\cdot5^{p/2-1}$,
    $\tau\leftarrow t/r$,
    $U\leftarrow I$\;
    \For{$j=1,\cdots,r$}{
        \For{$\upsilon=1,\cdots,\Upsilon$}{
            $U\leftarrow\prod_{\gamma=0}^{ \upsilon+(j-1)\Upsilon}\mathrm{e}^{\ii a_{(\upsilon,\gamma)}\tau H_{\pi^{\text{eo}}_{\upsilon}(\gamma)}^S}\cdot U$\;
            
        }
    }
\end{algorithm}
For clarity, we utilized matrix exponentials of sub-Hamiltonians $\{H_k^S\}$ in Alg.~\ref{alg:rpf}. 
However, in practice, we can replace these gates with exponentials of Pauli components of $\{H_k^S\}$ to simplify realization while keeping the error scaling the same as Thm.~\ref{thm:single} since these Pauli decompositions have already been used in our proof.

\paragraph*{Global Observables.---}
In practical applications of Hamiltonian simulation, measuring global observables is essential for uncovering the overall properties of quantum systems~\cite{zhang2017observation, jurcevicDirectObservationDynamical2017, hugginsUnbiasingFermionicQuantum2022, wangRealizationFractionalQuantum2024}. 
However, directly applying the previous support analysis is inapplicable since the global supports can rapidly saturate the whole $n$ qubits. 
Here, we circumvent this challenge by showing that the support analysis can still be useful when the observable is a sum of mutually commutative local observables, $\sum_{m=1}^MO_m$. 

Our approach involves analyzing the light cone of each summand operator separately.
Unlike the previous optimal configuration, which is specifically tailored for a fixed observable, we desire a configuration that consistently suppresses the support expansions of all local observables.
As defined below, we derive this configuration from the \emph{interaction hypergraph} of the Hamiltonian $H$.
\begin{definition}\label{def:graph}
Regarding the Pauli decomposition $H=\sum_{\alpha\in{\sf P}^n}s_\alpha P_\alpha$, the sets $\{S_i\}_{i=1}^L$ are defined such that every nonzero Pauli term belongs to some $S_i$, with no $S_i$ belonging to any other $S_j$.
Additionally, at least one Pauli term has support equal to every $S_i$.
The interaction hypergraph $G$ consists of the set of $n$ qubits and $\{S_i\}_{i=1}^L$ as its vertex and hyperedge sets, respectively. 
\end{definition}
We regroup $H$ into $\sum_{i=1}^LH_{S_i}$ so that every $H_{S_i}$ comprises Pauli terms supported within $S_i$.
This regrouping forms a decomposition of $H$.
The permutation arises from the edge-coloring of $G$ with an assignment $\varphi:\{S_i\}_{i=1}^L\rightarrow[\chi]$ and a total of $\chi$ colors, ensuring disjoint hyperedges of the same color.
Sub-Hamiltonians with supports (hyperedges) of the same color commute with each other.
Consequently, we can permute sub-Hamiltonians of the same color consecutively and adopt a back-and-forth manner among different stages.
This decomposition and permutation completes our chromatic configuration.
Since the permutation satisfies the symmetric constraint from~\cite{suzuki1991general}, we next employ $\{a_{(\upsilon,\gamma)}\}$ and $\Upsilon$ to construct a standard $p$th-order Suzuki-Trotter formula with our chromatic configuration, as summarized in Alg.~\ref{alg:mpf}.
A concrete example of regrouping and coloring of a nearest-neighbor model is illustrated in Fig.~\ref{fig:edge}(c).

\begin{algorithm}[t]
    \caption{Chromatic Product Formula}\label{alg:mpf}
    \SetKwInOut{Input}{input}\SetKwInOut{Output}{output}
    \Input{Hamiltonian $H$; Evolving time $t$; Step number $r$; Order $p$; Coloring $\varphi$ and number of colors $\chi$}
    \Output{Unitary $U$ organized from product formula}
    Regroup $H$ into $H_{S_1},\cdots,H_{S_L}$\;
   Adopt Suzuki-Trotter coefficients $a_{(\upsilon,\gamma)}$ from~\cite{suzuki1991general}\;
    $\Upsilon\leftarrow2\cdot5^{p/2-1}$, $\tau\leftarrow t/r$, $U\leftarrow I$\;
    \For{$j=1,\cdots,r$}{
        \For{$\upsilon=1,\cdots,\Upsilon$}{
            \eIf{$\upsilon$ is odd}{
                $U\leftarrow\prod_{c=1}^{\chi}\prod_{\gamma=1\to L:\varphi(S_\gamma)=c}\mathrm{e}^{\ii a_{(\upsilon,\gamma)}\tau H_{S_\gamma}}\cdot U$
            }{
                $U\leftarrow\prod_{c=\chi}^{1}\prod_{\gamma=L\to 1:\varphi(S_\gamma)=c}\mathrm{e}^{\ii a_{(\upsilon,\gamma)}\tau H_{S_\gamma}}\cdot U$
            }
        }
    }
\end{algorithm}

This chromatic configuration offers advantages by simultaneously suppressing the expansion of each individual summand observable.
Therefore, we achieve a mildly expanding light cone for each summand, as detailed in the following lemma.
We sketched its proof in Methods and deferred the formal one in Appendix~\ref{sec:Hamil}.
\begin{lemma}[Chromatic Optimality]\label{lm:global_support}
For an arbitrary set of observables $\{O_m\}$ and a Hamiltonian $H$, an $r$-step product formula $U$ from Alg.~\ref{alg:mpf} simultaneously 
enlarges the support of every $UO_mU^\dagger$ to at most $\bigcup_{k=0}^{(\chi-1) r\Upsilon+1}E^{S(O_m)}_k$. Given that $\chi$ is a constant, the support expansion of all observables is optimally slow. 
\end{lemma}
Unlike the single-observable case, the light cones cannot facilitate circuit reductions as each exponential gets involved in different light cones.
Despite needing to implement all exponentials as in Alg.~\ref{alg:mpf}, we can 
 reduce the error bound by only considering each summand's light cone separately.
Similarly, the error bound in Thm.~\ref{thm:multiple} is linear with the sum of light-cone widths,  independent of the system size.
This also suggests an advantage compared to the worst-case error analysis in~\cite{childs2021theory}, especially when the simulation time is short (but still polynomial in most cases).
\begin{theorem}[Global-Observable Error]\label{thm:multiple}
Consider the observable $O=\sum_{m=1}^MO_m$ where every $O_m$ is a local observable with support $S(O_m)$.
Suppose the $n$-qubit $H$ is $\ell$-local with a constant $\ell$ and has bounded interaction per qubit.
    With light-cone widths $w_{r,m}\coloneqq \sum\nolimits_{k=0}^{(\chi-1) r\Upsilon+3}\|H_k^{S(O_m)}\|_1$, the simulation error by an $r$-step $p$th-order $U$ from Alg.~\ref{alg:mpf} is bounded by $$\|\mathrm{e}^{\ii Ht}O\mathrm{e}^{-\ii Ht}-U OU^\dagger\|=\order{\frac{\sum_{m=1}^M\|O_m\|w_{r,m}t^{p+1}}{r^{p}}}.$$
\end{theorem}

Similarly, we can decompose the regrouped sub-Hamiltonians used in Alg.~\ref{alg:mpf} into Pauli operators, as we have included the Pauli-level analysis in the proof in Appendix~\ref{sec:append-multi}.

\subsection{Observable Advantages for Arbitrary-Time Simulations}

In our preceding analyses, we primarily focused on short-time simulations with specific knowledge of observables.
The question arises naturally to explore the observable-driven advantages in arbitrary-time simulations.
To address this, we turn our attention to the average error of random-input simulations, isolating the effects of states. 
In this section, we analyze the observable-driven reductions in average errors, revealing that the advantage generally exists for observables composed of a sum of Pauli operators.

To quantify the distance between arbitrary pairs of unitaries $U$ and $U_0$ for a given observable $O$ and an input ensemble $\mu=\{p_i,\psi_i\}$, we employ the average distance between measurements as a metric:
\begin{equation}
    D(U_0,U)_{O,\mu}\coloneqq \mathbb{E}_{\psi\sim\mu} |\bra{\psi}U_0O U_0^\dag\ket{\psi}-\bra{\psi}UO U^\dag\ket{\psi}|. 
\end{equation}
This metric is experimentally motivated since quantum circuits always end with measurements.

In the simulation context, algorithms often consist of a sequence of $r$ short steps of unitaries.
Therefore, our focus turns to the distance between the $r$ steps of the ideal evolution $\mathscr{U}_0=\mathrm{e}^{\ii Ht/r}$ and its approximation $\mathscr{U}=\mathscr{S}(t/r)$.
To investigate the average simulation error across the entire state space, we examine the input ensemble of Haar-random states, which can be intuitively understood as uniform randomness in the state space~\cite{mele2024introduction}. 
While Haar-random sampling is technically demanding, we can use $t$-design ensembles as a feasible alternative.
Meanwhile, these ensembles are indistinguishable from the Haar measure given $t$ copies of states. 
In the following, we show that sampling states from a 2-design ensemble $\mu_2$ suffices to exhibit error reductions with observable knowledge.

\begin{theorem}[Average Distance and Variance]\label{thm:random}
    For a 2-design ensemble $\mu_2$ of quantum states and an observable $O$, we can bound the average distance of simulation:
    \begin{gather*}
        D(\mathscr{U}_0^r,\mathscr{U}^r)_{O,\mu_2}\leq\sqrt{2}r\frac{\|O\|_2\cdot\|\mathscr{M}\|_2}{\sqrt{d(d+1)}},
    \end{gather*}
    where $\mathscr{M}\coloneqq\mathscr{U}_0^\dag\mathscr{U}-I$ represents the multiplicative error for each small step, $d=2^n$ is the dimension of the $n$-qubit Hilbert space, and $\|A\|_p=[\Tr(|A|^p)]^{1/p}$ denotes the Schatten $p$-norm.
    The variance of errors can also be bounded by
\begin{gather*}
     \text{Var}(\mathscr{U}_0^r,\mathscr{U}^r)_{O,\mu_2
     }\le\frac{2r^2\|O\|^2_2\cdot\|\mathscr{M}\|^2_2}{d(d+1)}.
 \end{gather*}
\end{theorem}
\noindent The proof is outlined in Methods, with a formal version provided in Appendix~\ref{Sec:Append-random}.
Based on the variance bound, we can show the concentration of our error analysis through Chebyshev's inequality.
We obtain the following corollary with the nested-commutator analysis by further narrowing the simulation approximation to the standard Suzuki formula methods.
 \begin{corollary}[Product-Formula Average Error]\label{co:product}
    Adopting $\mathscr{U}^r=\mathscr{S}_p(t/r)^r$ a standard $r$-step $p$th-order Suzuki product formula with the decomposition $H=\sum_{\gamma=1}^\Gamma H_\gamma$, the $D(\mathscr{U}_0^r,\mathscr{U}^r)_{O,\mu_2}$ from Thm.~\ref{thm:random} is bounded by $\order{T_2\|O\|_2t^{p+1}d^{-1/2}r^{-p})}$  with
\begin{gather*}
    T_2\coloneqq\sum_{\gamma_1,\dots,\gamma_{p+1}=1}^\Gamma\frac{1}{\sqrt{d}}\left\|[H_{\gamma_{p+1}},[H_{\gamma_p},\dots,[H_{\gamma_2},H_{\gamma_{1}}]]] \right\|_2.
\end{gather*}
Specifically, we have a triangle-bound for $p=2$ case,
\begin{align*}
    \frac{\sqrt{2}\|O\|_2t^3}{12dr^2}&\Bigg(\sum_{\gamma_1=1}^{\Gamma-1}
        \norm{\qty[\sum_{\gamma_3=\gamma_1+1}^\Gamma H_{\gamma_3},\qty[\sum_{\gamma_2=\gamma_1+1}^\Gamma H_{\gamma_2},H_{\gamma_1}]]}_2\notag\\
        &\ +
        \frac{1}{2} \sum_{\gamma_1=1}^{\Gamma-1}
        \norm{\qty[H_{\gamma_1},\qty[H_{\gamma_1}, \sum_{\gamma_2=\gamma_1+1}^\Gamma H_{\gamma_2}]]}_2\Bigg).
\end{align*}
\end{corollary}

In comparison, the random-input analysis without observable knowledge only ensures a simulation error based on the operator norm (Schatten $\infty$-norm) of observables~\cite{zhao2022hamiltonian}:
\begin{equation}
 \begin{aligned}
D(\mathscr{U}_0^r,\mathscr{U}^r)_{\mu_1}\coloneqq\max_ {O}D(\mathscr{U}_0^r,\mathscr{U}^r)_{O,\mu_1}\le \frac{2r\|\mathscr{M}\|_2\|O\|}{\sqrt{d}}. 
 \end{aligned}
 \end{equation}
Given that $\|O\|_2/\sqrt{d}\leq\|O\|$, the error reductions generally exist for Pauli-summation observables.
For a typical Pauli-summation observable $O=\sum_{m=1}^MO_m$ with evenly distributed coefficients $\|O_m\|=\Theta(1)$, we obtain a significant advantage since $\|O\|_2/\sqrt{d}=\order{\sqrt{M}}$, while the operator norm can only be bounded as $\order{M}$.

\begin{table*}[t!]
\centering
\resizebox{2.0\columnwidth}{!}{
\begin{tabular}{|c|c|c|c|c|c|}
\hline
                     & Short-Time Local &  Short-Time Global& Random with Ob.  & Random  &Worst-case \\ \hline
NN Lattice & $\order{\frac{t^{p+1}}{r^{p-D}}}$Thm.\ref{thm:single} & $\order{\frac{t^{p+1}}{r^{p-D}}}$Thm.\ref{thm:multiple} & $\order{\frac{\sqrt{n}t^{p+1}}{\sqrt{M}r^p}}$Thm.\ref{thm:random}& $\order{\frac{\sqrt{n}t^{p+1}}{r^{p}}}$\cite{zhao2022hamiltonian} &$\order{\frac{nt^{p+1}}{r^{p}}}$\cite{childs2019nearly} \\ \hline
\makecell{Power-law \\($\alpha>2D$)}&  $\order{\frac{t^{\frac{p(\alpha-2D)}{\alpha-D}+1}}{r^{\frac{p(\alpha-2D)}{\alpha-D}-D}}}$\makecell{\cite{childs2021theory} \&\\Thm.\ref{thm:single}}  &  $\tilde{\mathcal{O}}\left(\frac{t^{\frac{p(\alpha-2D)+D^2}{\alpha-D}+1}}{r^{\frac{(p-D)(\alpha-2D)}{\alpha-D}}}\right)$Thm.\ref{thm:multiple} &$\order{\frac{nt^{p+1}}{\sqrt{M}r^p}}$Thm.\ref{thm:random}   &$\order{\frac{nt^{p+1}}{r^{p}}}$\cite{zhao2022hamiltonian}&   $\order{\frac{nt^{p+1}}{r^{p}}}$\cite{childs2021theory}          \\ \hline
\end{tabular}}
\caption{\label{table:1}
Summary of $p$th-order product-formula simulation errors for nearest-neighbor (NN) and power-law lattice Hamiltonians with different analyses. 
To simplify the comparison, we present analyses for normalized observables across various cases, namely, $\|O\|=1$.
For the short-time global case, we focus on a summation observable $O=\sum_{m=1}^MO_m$ with local summands that satisfy $\sum_{m=1}^M\|O_m\|\sim\|O\|=1$.
For the random-input analysis with observable knowledge (Random with Ob.), we choose the Pauli-summation observable $O=\sum_{m=1}^MO_m$ with evenly distributed coefficients as $\|O_m\|=\order{1/M}$ for all $m$. }
\end{table*}

Even when sampling states from a 2-design ensemble is challenging, we can still establish an error reduction by leveraging the observable knowledge, albeit with a slightly relaxed error bound based on the Schatten 4-norms of the multiplicative error and the observable.
We elaborate on this result in Appendix~\ref{Sec:Append-random}.

\subsection{Applications}\label{sec:case}

In this section, we investigate simulations in several common models with observable knowledge.
We exhibit the observable-driven advantages in both short-time and arbitrary-time simulations compared to analyses without observable knowledge.
We provide a summary of the improved error scalings in Table.~\ref{table:1}.
The detailed derivations of these results are presented in the Appendix~\ref{sec:append-appli}.

\paragraph*{Nearest-Neighbor Hamiltonians.---}
For an $n$-qubit $D$-dimensional lattice $\Lambda$ with nearest-neighbor (NN) interactions, the Hamiltonian is given by $H=\sum_{(i,j)\in\Lambda}H_{i,j}$.
We normalize $H$ by assuming $\|H_{i,j}\|_1\leq1$.
This model is relevant to various intriguing condensed matter systems, including the Fermi- and Bose-Hubbard models \cite{hubbardElectronCorrelationsNarrow1963, hensgensQuantumSimulationFermi2017, shaoAntiferromagneticPhaseTransition2024}. 

In short-time simulations, we first consider a local observable $O$ with its constant-size support $S$.
According to Eq.~\eqref{eq:edge_set}, the corresponding edge set $E_k^S$ comprises $\order{k^{D-1}}$ qubits in the lattice, implying that the norm of the sub-Hamiltonian is $\|H_k^S\|_1=\order{k^{D-1}}$.
Consequently, the simulation error of the $r$-step $p$th-order product formula constructed in Alg.~\ref{alg:rpf} is bounded by 
\begin{gather}
    \|\mathrm{e}^{\ii Ht}O\mathrm{e}^{-\ii Ht}-U OU^\dagger\|=\order{\frac{\|O\|t^{p+1}}{r^{p-D}}}.
\end{gather}
To exceed $\order{\|O\|nt^{p+1}r^{-p}}$ from~\cite{childs2019nearly}, the width of the light cone must be smaller than the overall $\|H\|_1$, namely, $r^D=o(n)$.
Therefore, the gate count of simulation is $\order{r^{D+1}}$, which is system-size independent.

For the global observable $O=\sum_{m=1}^MO_m$, we first consider the advantage of short-time simulation for each local summand separately.
For an arbitrary summand $O_m$, both the edge set $E_k^{S(O_m)}$ and the 1-norm of sub-Hamiltonian $\|H_k^{S(O_m)}\|_1$ are bounded by $\order{k^{D-1}}$.
The regrouping in Def.~\ref{def:graph} preserves the interactions $H=\sum_{(i,j)\in\Lambda}H_{i,j}$.
By partitioning interactions into different axes, we can color all edges (supports) using $\chi=2D$ colors by labeling the parities of edges along each axis, as depicted in Fig.~\ref{fig:edge}(c). 
Therefore, the simulation error of an $r$-step $p$th-order formula from Alg.~\ref{alg:mpf} is
\begin{gather}
   \|\mathrm{e}^{\ii Ht}O\mathrm{e}^{-\ii Ht}-U OU^\dagger\|=\order{\frac{\sum_{m=1}^M\|O_m\|t^{p+1}}{r^{p-D}}}.
\end{gather}
Similarly, we need the time to be as short as $r^D=o(n)$ to outperform the worst-case error.
The overall gate count is $\order{rn}$, which is much more efficient than trivially implementing Alg.~\ref{alg:rpf} for each summand solely, as shown in Appendix~\ref{sec:NNH}.

For more general cases with arbitrary simulation times, we turn to the average simulation error with random inputs. 
Theorem~\ref{thm:random} and its corollary imply the advantage in this NN model with the preceding chromatic decomposition.
Specifically, we show in Appendix~\ref{sec:NNH} that the nested commutator scales as $T_2=\order{\sqrt{n}}$.
Given a Pauli-summation observable $O=\sum_{m=1}^MO_m$ with all norms of Pauli operators $\|O_m\|$ scale the same, the normalized 2-norm of $O$ satisfies $\order{\max_m\|O_m\|\sqrt{M}}$.
The average simulation error from Cor.~\ref{co:product} by a standard $r$-step $p$th-order Suzuki-Trotter formula is
\begin{gather}
    D(\mathscr{U}_0^r,\mathscr{U}^r)_{O,\mu_2}=\order{ \frac{\max_m\|O_m\|\sqrt{Mn}t^{p+1}}{r^p}},
\end{gather}
while the analysis without observable knowledge only offers the bound $\order{\max_m\|O_m\|M\sqrt{n}t^{p+1}r^{-p}}$.
For example, the magnetization $\sum_{j=1}^{n}Z_j/n$ has a normalized Schatten 2-norm $\order{1/\sqrt{n}}$, generating a size-independent error of $\order{t^{p+1}r^{-p}}$.
Our analysis thus provides an additional $\order{\sqrt{n}}$ speed-up compared to the analysis without observable knowledge, which yields $\order{\sqrt{n} t^{p+1}r^{-p}}$ as in~\cite{zhao2022hamiltonian}.

\paragraph*{Power-Law Hamiltonians.---}\label{pl}
For an $n$-qubit $D$-dimensional lattice $\Lambda$, we consider the rapidly power-law decaying interactions $H_P=\sum_{i,j\in\Lambda}H_{i,j}$ with
\begin{gather}
    \|H_{i,j}\|\leq
    \begin{cases}
        \order{1}& i=j\\
        \order{\text{d}(i,j)^{-\alpha}}& i\neq j
    \end{cases}, 
\end{gather}
where $\text{d}(i,j)$ represents the distance between sites $i$ and $j$ in the lattice and a large decay factor $\alpha>2D$.

We first detect the short-time advantage for a single local observable $O$ with a constant-size support $S$.
To this end, we truncate $H_P$ so the light cone would not instantly saturate the system.
When executing an $r$-step $\Upsilon$-stage Alg.~\ref{alg:rpf}, the light cone of $O$ remains within the first $r\Upsilon+2$ edge sets.
Fixing a $d_0>0$, we divide $\Lambda$ into two parts: $\Lambda_{in}\coloneqq\{j\in\Lambda\,|\,\text{d}(j,S)\leq(r\Upsilon+1)d_0\}$ and $\Lambda_{out}=\Lambda\backslash\Lambda_{in}$.
We truncate $H_P$ to $H_{\text{lc}}$ by removing interactions involving $\Lambda_{in}$ with lengths longer than $d_0$.
The truncation error is bounded as calculated in~\cite{tran2019locality},
\begin{gather}\label{eq:lc}
    \|\mathrm{e}^{\ii H_Pt}-\mathrm{e}^{\ii H_{\text{lc}}t}\|\leq\|H_P-H_{\text{lc}}\|t\leq\order{\frac{r^Dt}{d_0^{\alpha-2D}}}.
\end{gather}
There are $r\Upsilon+3$ edge sets in $H_{\text{lc}}$.
For $k\leq r\Upsilon+2$, $E_k^{S}$ comprises $\order{k^{D-1}d_0^D}$ qubits, and $\|H_k^S\|_1=\order{k^{D-1}d_0^D}$.
Since $D$ and $\Upsilon$ are constants, Thm.~\ref{alg:rpf} bound the error of an $r$-step $p$th-order formula from Alg.~\ref{alg:rpf} 
\begin{gather}
    \|\mathrm{e}^{\ii Ht}O\mathrm{e}^{-\ii Ht}-U OU^\dagger\|=\order{\frac{\|O\|r^Dt}{d_0^{\alpha-2D}}+\frac{\|O\|d_0^Dt^{p+1}}{r^{p-D}}},
\end{gather} 
with minimizer $d_0=\order{(r/t)^{p/(\alpha-D)}}$.
This result aligns with the local observable analysis in~\cite{childs2021theory}.

As for the global observable consisting of a summation, $O=\sum_{m=1}^MO_m$, the support analysis and short-time advantages still hold for local summands.
We adopt a more general truncation in $H_{\text{trc}}=\sum_{i,j\in\Lambda,\text{d}(i,j)\leq d_0}H_{i,j}$.
Taking the observable $O$ into account, we get a nearly size-independent truncation error with $\tilde{\mathcal{O}}$ omitting the logarithmic terms, which is analyzed in Appendix~\ref{sec:powerlaw},
\begin{gather}\label{eq:trunc}
    \tilde{\mathcal O}\left(\frac{t^{D+1}}{d_0^{\alpha-2D}}\sum_{m=1}^M\|O_m\|\right).
\end{gather}
To regroup $H_{\text{trc}}$, we partition the lattice into $D$-dimensional cubes with a side length of $d_0$.
Each regrouping sub-Hamiltonian comprises nearest-neighbor cube-based interactions, allowing for a $3^D-1$ coloring.
The $k$th edge set of $O_m$, $E_k^{S(O_m)}$, consists of qubits in $\order{k^{D-1}}$ cubes.
The 1-norm $\|H_k^{S(O_m)}\|_1$ scales as $\order{k^{D-1}d_0^D}$.
According to Thm.~\ref{thm:multiple}, the simulation error of an $r$-step $p$th-order formula from Alg.~\ref{alg:mpf} is bounded by
\begin{gather}\label{eq:sumPow}
    \tilde{\mathcal{O}}\left(\sum\nolimits_{m=1}^M\|O_m\|\left(\frac{t^{D+1}}{d_0^{\alpha-2D}}+\frac{d_0^Dt^{p+1}}{r^{p-D}}\right)\right).
\end{gather}
The minimizer of Eq.~\eqref{eq:sumPow} is $d_0=\tilde{\mathcal{O}}((r/t)^{(p-D)/(\alpha-D)})$.

We also analyze the advantage in arbitrary-time simulations with random inputs.
Even though proving a better estimation of $T_2$ than $\order{n}$ for power-law models is generally challenging, our numerical results in Appendix~\ref{sec:append-additional} suggest that $T_2$ tends to scale as $\order{\sqrt{n}}$.
Additionally, our analysis still offers advantages by the normalized Schatten 2-norm of $O$.
For a Pauli-summation observable with evenly distributed coefficients, the average error using an $r$-step $p$th-order Suzuki-Trotter product formula is 
\begin{gather}
    D(\mathscr{U}_0^r,\mathscr{U}^r)_{O,\mu_2}=\order{ \frac{\max_m\|O_m\|\sqrt{M}nt^{p+1}}{r^p}}.
\end{gather}
As a comparison, average-error analysis without observable knowledge only keeps the worst-case result $\order{\|O\|nt^{p+1}r^{-p}}$ for this model~\cite{zhao2022hamiltonian}.

\begin{figure*}[t]
    \centering
    
    \includegraphics[width=0.95\linewidth]{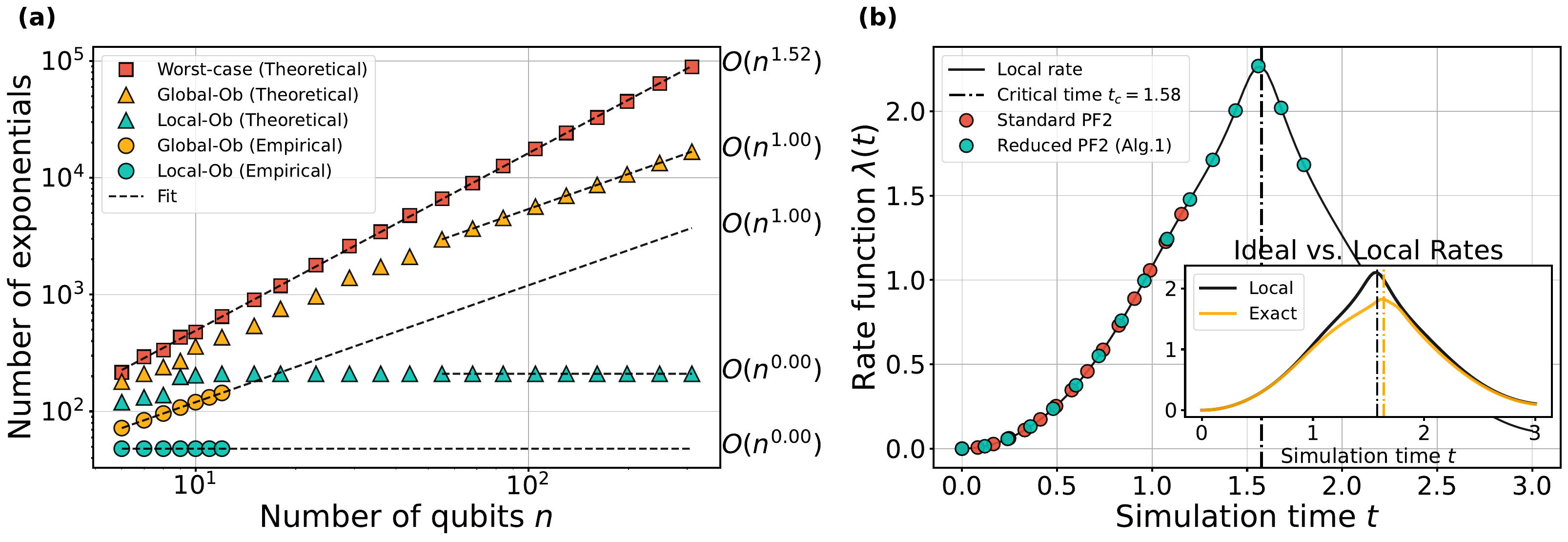}
    \caption{Numerical results for short-time second-order product formula simulations.
    \textbf{(a)} Number of exponentials needed to achieve simulation precision $\epsilon=10^{-3}$ at $t=0.1$ under an MFI Hamiltonian with $J=1$, $h=0.5$, $g=1.2$ as in Eq.~\eqref{eq:H_ising}. 
    The gate counts for local (green) and global (yellow) observables, $Z_1$ and $\frac{1}{n-1}\sum_{j=1}^{n-1}Z_jZ_{j+1}$ are obtained from different theoretical and empirical error estimations. 
    \textbf{(b)} Simulation of dynamical quantum phase transition with $k=3$ local observable as in Eq.~\eqref{eq:echo}.
    We use a 12-qubit TFI Hamiltonian with $J=0.2$ and $h=1$ as in Eq.~\eqref{eq:TFIH}.
    By fixing the precision $\epsilon=0.05$ and a gate budget of 500, the guaranteed simulation times are $t=1.80$ from Thm.~\ref{thm:single} and $t'=1.15$ from worst-case analysis in~\cite{childs2021theory}, with step lengths $\delta t=0.12$ and $\delta t'=0.08$.
    In the inset, we validate the local approximation of the rate function $\lambda_n(t)$ by $\lambda_k(t)$.
    }
    \label{fig:TFIsing}
\end{figure*}

\subsection{Numerical Results} 
In this section, we present numerical results that demonstrate the benefits of incorporating observable knowledge in quantum simulation.
Our results illustrate observable-driven advantages in both short-time and arbitrary-time simulations, utilizing our support and random-input analyses, respectively.
Detailed numerical settings are deferred to Methods and Appendix~\ref{sec:append-num}.

First, we examine the short-time simulation advantages as shown in Fig.~\ref{fig:TFIsing}.
In subplot (a), we focus on the gate counts for simulating an $n$-qubit one-dimensional mixed-field Ising (MFI) Hamiltonian,
\begin{gather}
    H=J\sum_{j=1}^{n-1} X_jX_{j+1} + h\sum_{j=1}^n X_j +g\sum_{j=1}^nY_j.
    \label{eq:H_ising}
\end{gather}
Specifically, we consider the number of exponentials as the gate count required to implement second-order product formulas for simulating the local observable $O=Z_1$ and the global observable $O=\frac{1}{n-1}\sum_{j=1}^{n-1}Z_jZ_{j+1}$ with a fixed precision $\|\mathrm{e}^{\ii Ht}O\mathrm{e}^{-\ii Ht}-UOU^\dag\|\leq\epsilon$.
The theoretical lines are estimated based on error analyses from our Thms.~\ref{thm:single} and~\ref{thm:multiple}, and the worst-case bound in~\cite{childs2021theory}.
These results highlight the observable-driven speed-ups of short-time simulation.
Empirical gate counts from both Algs.~\ref{alg:rpf} and~\ref{alg:mpf} closely align with our theoretical results, suggesting the tightness of our theoretical bounds.

In Fig.~\ref{fig:TFIsing}(b), we demonstrate short-time advantages in the practical task of exploring the local dynamical quantum phase transition (DQPT), which approximates the original DQPT using local observables as introduced in~\cite{halimehLocalMeasuresDynamical2021}.
Specifically, we rehearse the full procedures to estimate the local rate function $\lambda_k(t)\coloneqq-\log(\mathcal{L}_k(t))/k$ using different product formulas with
\begin{gather}\label{eq:echo}
    \mathcal{L}_k(t)\coloneqq\bra{\psi(0)}e^{\ii H t} \prod_{j=1}^k P_j e^{-\ii Ht}\ket{\psi(0)},
\end{gather}
where $k=3$, $\ket{\psi(0)}=\ket{0}^{\otimes n}$, and $P_j:=\op{0}_j$ is the projector on the $j$-th qubit.
We adopt a one-dimensional transverse-field Ising (TFI) Hamiltonian
\begin{gather}\label{eq:TFIH}
    H=J\sum_{j=1}^{n-1} Z_jZ_{j+1}+h\sum_{j=1}^n X_j,
\end{gather}
with $n=12$.
Fixing the precision $\epsilon$ and a gate budget, our bound from Thm.~\ref{thm:single} and the worst-case analysis from~\cite{childs2021theory} report the numbers of steps used and guaranteed simulation times, exceeding which the simulation error cannot be bounded by $\epsilon$.
Using these parameters, we implement the second-order Alg.~\ref{alg:rpf} and standard product formula to estimate the rate functions and get empirical results.
Our bound offers a longer guaranteed simulation time, increasing by 50\% compared to the worst-case bound.
This gap widens with increasing system size.
Clearly, the worst-case analysis fails to capture the critical point with the given constraints.

\begin{figure*}[t]
    \centering
    \includegraphics[width=0.96\linewidth]{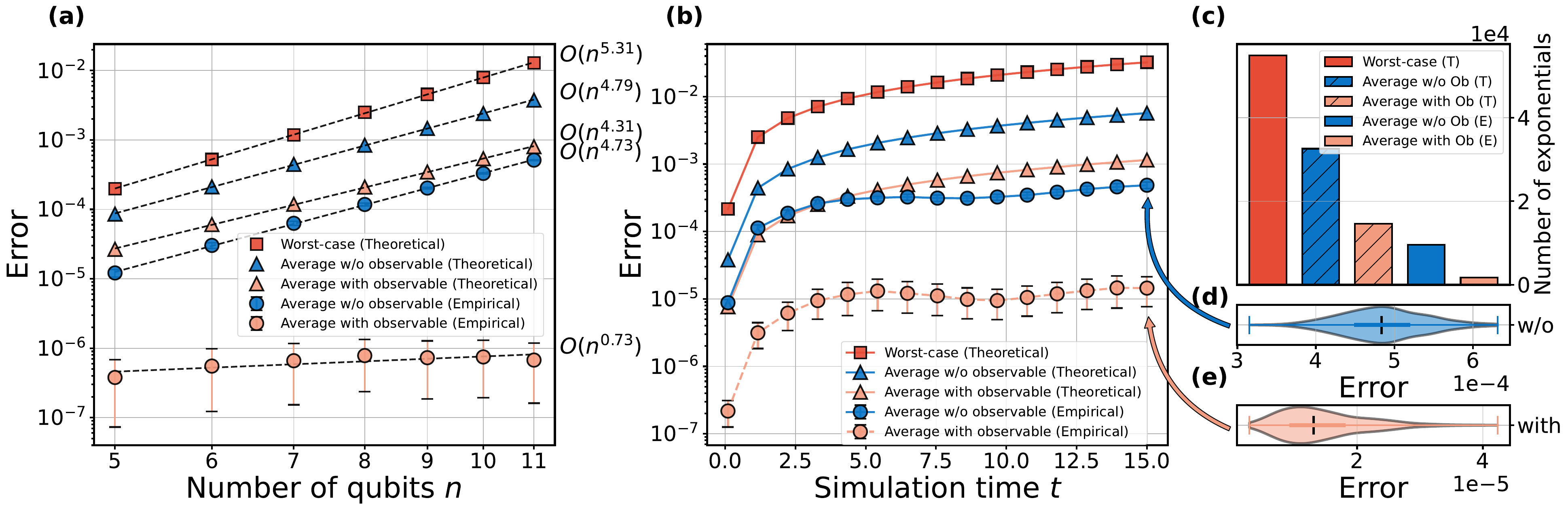}
    \caption{Numerical results for random-input second-order product formula simulations.
    \textbf{(a)} Errors for simulating $\sum_jZ_j$ under a power-law Hamiltonian with $\alpha=4$, $J=1$, and $h=0.5$ as in Eq.~\eqref{eq:power-law}. 
    We simulate dynamics with increasing sizes, setting $t=n$ with a fixed step number $r=10000$. 
    Error bars indicate standard deviations of simulation errors across 500 independent Haar-random states per empirical point.
    \textbf{(b)} Errors for simulating the dynamics of a 6-qubit Hydrogen chain $\text{H}_3$ with bond length $2\mathring{A}$ as in Eq.~\eqref{eq:molecule}, with observable being another $\text{H}_3$ Hamiltonian at bond length $1\mathring{A}$. 
    We fix the step length $t/r=0.1$.
    Error bars represent standard deviations estimated by 500 independent Haar-random states.
    \textbf{(c)} The numbers of exponentials required to ensure average errors smaller than $\epsilon=10^{-3}$ in different theoretical and empirical analyses at $t=15$
    We choose the same dynamics and observable as (b).
    Two empirical bars represent the mean value obtained from 50 rounds of sampling for 50 random states.
    The error bars, denoting standard deviations, are negligibly small.
    \textbf{(d)} and \textbf{(e)} Insets of (b) depicting empirical error distributions at $t=15$ from 500 random states for cases with and without observable knowledge, respectively. }
    \label{fig:random}
\end{figure*}

For the random-input result in Fig.~\ref{fig:random}, we first estimate the (average) simulation distances for different analyses of a one-dimensional power-law decaying Hamiltonian,
\begin{gather}\label{eq:power-law}
    H=\sum_{i=1}^n\sum_{j=i+1}^n\frac{J}{(j-i)^\alpha}(X_iX_j+Y_iY_j+Z_iZ_j)+h\sum_{j=1}^nX_j.
\end{gather}
We focus on the observable $O=\sum_{j=1}^{n}Z_j$.
The theoretical lines are based on the worst-case analysis~\cite{childs2021theory}, the previous random-input analysis without observable knowledge from~\cite{zhao2022hamiltonian}, and our Cor.~\ref{co:product}.
As shown in Fig.~\ref{fig:random}(a), our analysis is significantly tighter than previous random-input analysis, improved by $\order{\sqrt{n}}$.

We exhibit the random-input advantage in a more practical setting in subplots (b)-(e) by simulating the Hamiltonian of a three-atom Hydrogen chain $\text{H}_3$ 
\begin{equation}
    H = H_0 + \sum_{pq}h_{pq}a_p^\dagger a_q + \frac{1}{2} \sum_{pqrs}h_{pqrs} a_p^\dagger a_q^\dagger a_r a_s,
    \label{eq:molecule}
\end{equation}
targeting the observable of another $\text{H}_3$ Hamiltonian with a different bond length.
Although the molecular Hamiltonian is complicated after the Jordan-Wigner transformation, we can efficiently estimate the Schatten 2-norm to bound the error tightly using Cor.~\ref{co:product} even for a large $n$.
As shown in (b), incorporating observable knowledge significantly reduces both theoretical and empirical errors.
Subplot (c) shows the numbers of exponentials required to implement second-order product formulas within a precision of $\epsilon=10^{-3}$ in different cases.
Violins in subplots (d) and (e) illustrate the empirical distributions of simulation errors from Haar-random input states.

\section{Discussion}
Our study comprehensively analyzes the advantages of incorporating observable knowledge in short-time and arbitrary-time simulation scenarios.
In short-time simulations, utilizing
observables knowledge leads to size-independent errors for both local observables and global observables composed of summations over local terms.
In the arbitrary-time case, we discover that observable knowledge generally offers advantages for the average simulation errors from random inputs.

Detecting further speed-ups from observable knowledge for simulating general times presents an interesting direction.
While we focus on average simulation accuracy in this work, a more thorough analysis of observable-driven advantages in other input settings could generate valuable insights.
Additionally, our previous numerical results indicate that the empirical average simulation accuracy can vastly exceed theoretical predictions, suggesting room for future improvements, as shown in Fig.~\ref{fig:random}.

Experimental implementations of the proposed simulation algorithms also hold promise for better performance in practice. 
As shown in our study, theoretical bounds, including short-time and arbitrary-time simulation errors, can be easily calculated in advance. 
These results could lead to more accurate and resource-efficient gate count estimation in some complicated tasks related to quantum chemistry and high-energy physics. 
Such progress could narrow the gap between realizing quantum advantages and near-term experimental capabilities.

\section*{Methods}

\subsection{Sketch proofs for short-time local observables}
In the case of a single local observable within a short-time evolution, we adopt the interactive decomposition and even-odd permutation to implement the product formula.
By using this product formula, the support of the observable slowly expands along with time steps, which is optimal from Lemma~\ref{lm:optimal}.
Therefore, unitary gates outside the support never get involved in the evolution of the observable.
Simply removing them reduces the gate count of this product formula as stated in Alg.~\ref{alg:rpf}.
According to Thm~\ref{thm:single}, this reduced formula offers both a better error bound and a better gate count.

We first sketch the proof of Lemma~\ref{lm:optimal} starting from the first stage of the product formula.
By definition, each qubit in $E_1^S$ must be within the support of some Pauli term $P_\alpha$ in $H$ such that $S(P_\alpha)\cap S\neq \emptyset$.
Thus, the lower bound holds for $\Upsilon=1$,
\begin{gather}
    \left(\biguplus_{\upsilon=1}^{\Upsilon}\biguplus_{\gamma=1}^{\Gamma}\mathrm{e}^{\ii ta_{(\upsilon,\gamma)}H_{\pi_\upsilon(\gamma)}}\uplus S\right)\supseteq \bigcup_{k=0}^\Upsilon E_k^S.
\end{gather}
The proof extends inductively for larger $\Upsilon$ according to the definition of edge sets.
The optimality of the product formula with interactive decomposition and even-odd permutation is also proven inductively.
Starting with $\Upsilon=1$, we show that $S(H_j^S)\cap S(H_k^S)\neq\emptyset$ only for $|j-k|=1$.
Therefore, we have
$\biguplus_{\gamma=1}^{\Gamma}\mathrm{e}^{\ii ta_{(1,\gamma)}H^S_{\pi^{\text{eo}}_1(\gamma)}}\uplus S= E_1^S\cup S$.
This equality can be easily proved inductively for larger $\Upsilon$, which completes the proof.

Alg.~\ref{alg:rpf} is implemented by removing all exponentials disjoint with the support as suggested by Lemma~\ref{lm:optimal}.
Unlike the standard Suzuki-Trotter formula, we cannot directly employ the nested-commutator bound in~\cite{childs2021theory} for this algorithm.
To prove Thm.~\ref{thm:single}, we fill exponentials in Alg.~\ref{alg:rpf} make each step a standard formula of $H$.
We denote this filled method the virtual product formula
\begin{gather}
    \mathscr{S}_{\text{V}}(t)\coloneqq\mathscr{S}_1(\tau)\cdots\mathscr{S}_r(\tau).
\end{gather}
For example, by the end of step $j$, sub-Hamiltonians $\{H_1^S,\cdots,H_{j\Upsilon}^S\}$ would intersect with the light cone.
Hence, we fill step $j$'s formula to $\mathscr{S}_j(\tau)$ with even-odd permutation and decomposition
\begin{gather}
    H=H_1^S+\cdots+H_{j\Upsilon}^S+H_{\text{others},j}.
\end{gather}
Therefore, we can use the triangle inequality to bound the error between $\mathrm{e}^{\ii Ht}$ and $U$ by the following two parts: the error between $U$ and $\mathscr{S}_V(t)$, and the error from $\mathscr{S}_V(t)$ to the ideal $\mathrm{e}^{\ii Ht}$.
It is easy to check that the support of the evolved $O$ under the virtual formula still expands optimally.
In this sense, the exponential of $H_{\text{others},j}$ is outside the support in each step $j$ and does not affect the evolved $O$, which means the first error is zero.
The second error can be bounded by the nested-commutator bound, which is linear with $\sum_{k=0}^{r\Upsilon+1}\|H_k^S\|_1$.
Combining these two parts, we get the error bound linear with the ``width" of the light cone.
\subsection{Sketch proofs for short-time global observables}
In this section, we focus on simulating global observables consisting of a summation of local observables.
For simplicity, we assume all summand observables are commutative.
Otherwise, we can always divide the summation into different stabilizer groups and estimate them independently.
This analysis hinges on a clever design of the configuration such that all observables expand slowly.
To this end, we introduce the interaction hypergraph in Def.~\ref{def:graph} and the edge coloring thereof, which suggest the regrouping decomposition and coloring permutation, respectively.
The corresponding product formula is summarized in Alg.~\ref{alg:mpf}. 

Lemma~\ref{lm:global_support} asserts that a $\Upsilon$-stage Alg.~\ref{alg:mpf} with $\chi$ colors expands an arbitrary support $S(O)$ to its first $(\chi-1)\Upsilon+1$ edge sets.
Note that the sub-Hamiltonians within a single color are mutually disjoint, allowing us to consider the collection of exponentials of the same color as a unified entity in the implementation.
Starting from the first color, exponentials contribute to the expansion of $S(O)$ only if sub-Hamiltonians overlap with $S(O)$, so 
\begin{gather}
   \biguplus_{\gamma=1\to L:\varphi(S_\gamma)=1}\mathrm{e}^{\ii a_{(1,\gamma)}\tau H_{S_\gamma}}\uplus S(O)\subseteq S(O)\cup E_1^{S(O)}. 
\end{gather}
The same illustration applies to subsequent colors and stages, so each color and stage of the product formula enlarges the support by at most one layer.
Alg.~\ref{alg:mpf} employs the back-and-forth permutation, ensuring continuity of colors between stages. 
Therefore, there are altogether $(\chi-1)r\Upsilon+1$ effective colors, which completes the proof.

As for the error bound for simulating $O=\sum_{m=1}^MO_m$ in Thm.~\ref{thm:multiple}, we analyze the error for every single summand observable and combine them by the triangle inequality.
For each summand $O_m$, we introduce the corresponding virtual product formula according to its own edge sets, 
\begin{gather}
    \mathscr{S}_{\text{V},m}(t)\coloneqq\mathscr{S}_{1,m}(\tau)\cdots\mathscr{S}_{r,m}(\tau).
\end{gather}
For example, since by the end of step $j$, the support of $O_m$ at most reaches $\bigcup_{i=0}^{(\chi-1)j\Upsilon+1}E_i^{S(O_m)}$, the virtual formula keeps the coloring decomposition of sub-Hamiltonians intersecting with the light cone and absorb all other terms in the tail.
Considering the commutation relationships regarding the evolved $O_m$, all unitaries outside the light cone are ineffective in both $\mathscr{S}_{\text{V},m}(t)$ and Alg.~\ref{alg:mpf}.
Therefore, simulations for $O_m$ from $\mathscr{S}_{\text{V},m}(t)$ and $U$ are equivalent.
We then analyze simulation error between $\mathscr{S}_{\text{V},m}(t)$ and $\mathrm{e}^{\ii Ht}$.
The nested commutator of the $j$th virtual formula for $O_m$ scales as $\order{\sum_{i=0}^{j\chi\Upsilon+2}h_{m,i}}$.
Combining simulation errors for all $j$ and $O_m$ according to the triangle inequality completes the analysis of overall simulation errors.
\subsection{Sketch proofs for random-input simulation}
We consider the proof for Thm.~\ref{thm:random}. 
First, we divide the distance $D(\mathscr{U}_0^r,\mathscr{U}^r)_{O,\mu_2}$ using the triangle inequality so that each segment represents the single-step error in the following form
\begin{gather}
    \|(\mathscr{U}_0^\dag U_2^\dag\ket{\psi}\bra{\psi}U_2\mathscr{U}_0-\mathscr{U}^\dag U_2^\dag\ket{\psi}\bra{\psi}U_2\mathscr{U})\cdot U_1OU_1^\dag\|_1,
\end{gather}
where $U_1=\mathscr{U}_0^{i-1}$ and $U_2=\mathscr{U}^{r-i}$ for step $i$.
Using the H\"{o}lder inequality, we can decompose this distance into the product of two Schattern 2-norms.
Recalling that the state $\ket{\psi}\bra{\psi}$ is sampled from a 2-design ensemble, we can discard these two unitaries $U_1$ and $U_2$ due to the unitary invariance of the Schattern norms and the 2-design property, making these estimations in all segments the same as
\begin{gather}
   \int_\psi\|(I+\mathscr{M})^\dag\ket{\psi}\bra{\psi}(I+\mathscr{M})-\ket{\psi}\bra{\psi}\|_2d\mu_2.
\end{gather}
We then use  $\int_\psi\ket{\psi}\bra{\psi}^{\otimes2}d\mu_2=\frac{I+S}{d(d+1)}$ to bound the single-step error.
The variance proof follows a similar approach. 
For Cor.~\ref{co:product}, we recruit the similar result from~\cite{zhao2022hamiltonian} to bound $\|\mathscr{M}\|_2$ by the nested commutator and further generate the second-order error bound.
Refer to Appendix~\ref{Sec:Append-random} for full details.

\subsection{Numerical Settings}
In Fig.~\ref{fig:TFIsing}, we evaluate the gate complexities in short-time simulation for different tasks.
In subplot (a), we estimate the gate count to simulate observables under an MFI Hamiltonian within the $10^{-3}$ precision.
We estimate the worst-case theoretical error of the second-order Suzuki-Trotter formula using the explicit version of the nested-commutator bound, Prop.~10, in~\cite{childs2021theory}.
This proposition is also used to derive explicit versions of Thms.~\ref{thm:single} and~\ref{thm:multiple} as substitutes for the asymptotic version of the nested-commutator bounds.
To efficiently compute the error estimation for large $n$, we adapt operator norms in all bounds to the 1-norms (summation of Pauli coefficients) according to the triangle inequality.
Empirical errors are faithfully calculated from $\|\mathrm{e}^{\ii Ht}O\mathrm{e}^{-\ii Ht}-UOU^\dag\|$ with $U$ adopted from Alg.~\ref{alg:rpf} or~\ref{alg:mpf}.
With these error analyses, we perform a binary search for the desired number of steps.
In subplot (b), we explore the DQPT using the local observable.
With the fixed precision and gate budget, we estimated guaranteed simulation times from the explicit versions of Thm.~\ref{thm:single} and the worst-case bound.
Refer to Appendix~\ref{sec:append-num} for detailed bounds.

In Fig.~\ref{fig:random}, we focus on the random-input results of various models.
In (a), we exhibit average-case errors with or without observables and the worst-case errors along system sizes.
The average errors are calculated from Cor.~\ref{co:product} and Thm.~3 from \cite{zhao2022hamiltonian}, respectively.
The empirical average errors are estimated from 500 Haar-random states, with blue circles denoting the random errors without observable knowledge $\|\mathrm{e}^{-\ii Ht}\rho\mathrm{e}^{\ii Ht}-U^\dag\rho U\|_1\|O\|$, and orange circles denoting our observable-induced case $\Tr[(\mathrm{e}^{-\ii Ht}\rho\mathrm{e}^{\ii Ht}-U^\dag\rho U)O]$.
For subplot (b), we simulate Hydrogen chain Hamiltonian $\text{H}_3$ with a bond length $2\mathring{A}$ and measure on another $\text{H}_3$ at $1\mathring{A}$.
Hydrogen chain Hamiltonians from Eq.~\eqref{eq:molecule} with specific bond length and the sto-3g basis are generated by~\cite{mccleanOpenFermionElectronicStructure2020}.
Notably, since the observable consists of anti-commutative Pauli operators, they cannot be simultaneously measured by a projective measurement.
To simplify the product formula and realize measurements, we decompose both the Hamiltonian and observable (two $\text{H}_3$ Hamiltonians) into multiple mutually commutative subsets (stabilizer sets) by a simple greedy algorithm.
The figure depicts overall errors for all subsets with a fixed $t/r$.
In (c), we similarly use the binary search to find the same $r$ for all subsets to calculate the total numbers of exponentials.
Error analyses are the same as (b).

\section*{Acknowledgments}
We are grateful to Pei Zeng, You Zhou, Xiao Yuan, Giulio Chiribella, and Steven T. Flammia for the helpful discussions and comments.
W.Y. is supported by the HKU Presidential Scholarship.
J.X. is supported by the HKU Postgraduate Scholarship.
Q.Z. acknowledges funding from the HKU Seed Fund for Basic Research for New Staff via Project 2201100596, the Guangdong Natural Science Fund via Project 2023A1515012185, the National Natural Science Foundation of China (NSFC) via Project Nos. 12305030 and 12347104, Hong Kong Research Grant Council (RGC) via Project No. 27300823, N\_HKU718/23, and R6010-23, Guangdong Provincial Quantum Science Strategic Initiative GDZX2200001. 


\bibliography{reference}

\begin{thebibliography}{62}%
\makeatletter
\providecommand \@ifxundefined [1]{%
 \@ifx{#1\undefined}
}%
\providecommand \@ifnum [1]{%
 \ifnum #1\expandafter \@firstoftwo
 \else \expandafter \@secondoftwo
 \fi
}%
\providecommand \@ifx [1]{%
 \ifx #1\expandafter \@firstoftwo
 \else \expandafter \@secondoftwo
 \fi
}%
\providecommand \natexlab [1]{#1}%
\providecommand \enquote  [1]{``#1''}%
\providecommand \bibnamefont  [1]{#1}%
\providecommand \bibfnamefont [1]{#1}%
\providecommand \citenamefont [1]{#1}%
\providecommand \href@noop [0]{\@secondoftwo}%
\providecommand \href [0]{\begingroup \@sanitize@url \@href}%
\providecommand \@href[1]{\@@startlink{#1}\@@href}%
\providecommand \@@href[1]{\endgroup#1\@@endlink}%
\providecommand \@sanitize@url [0]{\catcode `\\12\catcode `\$12\catcode `\&12\catcode `\#12\catcode `\^12\catcode `\_12\catcode `\%12\relax}%
\providecommand \@@startlink[1]{}%
\providecommand \@@endlink[0]{}%
\providecommand \url  [0]{\begingroup\@sanitize@url \@url }%
\providecommand \@url [1]{\endgroup\@href {#1}{\urlprefix }}%
\providecommand \urlprefix  [0]{URL }%
\providecommand \Eprint [0]{\href }%
\providecommand \doibase [0]{https://doi.org/}%
\providecommand \selectlanguage [0]{\@gobble}%
\providecommand \bibinfo  [0]{\@secondoftwo}%
\providecommand \bibfield  [0]{\@secondoftwo}%
\providecommand \translation [1]{[#1]}%
\providecommand \BibitemOpen [0]{}%
\providecommand \bibitemStop [0]{}%
\providecommand \bibitemNoStop [0]{.\EOS\space}%
\providecommand \EOS [0]{\spacefactor3000\relax}%
\providecommand \BibitemShut  [1]{\csname bibitem#1\endcsname}%
\let\auto@bib@innerbib\@empty
\bibitem [{\citenamefont {Feynman}(1982)}]{feynman1982simulating}%
  \BibitemOpen
  \bibfield  {author} {\bibinfo {author} {\bibfnamefont {R.~P.}\ \bibnamefont {Feynman}},\ }\bibfield  {title} {\bibinfo {title} {Simulating physics with computers},\ }\href@noop {} {\bibfield  {journal} {\bibinfo  {journal} {International Journal of Theoretical Physics}\ }\textbf {\bibinfo {volume} {21}} (\bibinfo {year} {1982})}\BibitemShut {NoStop}%
\bibitem [{\citenamefont {Noh}\ and\ \citenamefont {Angelakis}(2016)}]{noh2016quantum}%
  \BibitemOpen
  \bibfield  {author} {\bibinfo {author} {\bibfnamefont {C.}~\bibnamefont {Noh}}\ and\ \bibinfo {author} {\bibfnamefont {D.~G.}\ \bibnamefont {Angelakis}},\ }\bibfield  {title} {\bibinfo {title} {Quantum simulations and many-body physics with light},\ }\href@noop {} {\bibfield  {journal} {\bibinfo  {journal} {Reports on Progress in Physics}\ }\textbf {\bibinfo {volume} {80}},\ \bibinfo {pages} {016401} (\bibinfo {year} {2016})}\BibitemShut {NoStop}%
\bibitem [{\citenamefont {Schreiber}\ \emph {et~al.}(2015)\citenamefont {Schreiber}, \citenamefont {Hodgman}, \citenamefont {Bordia}, \citenamefont {L{\"u}schen}, \citenamefont {Fischer}, \citenamefont {Vosk}, \citenamefont {Altman}, \citenamefont {Schneider},\ and\ \citenamefont {Bloch}}]{schreiber2015observation}%
  \BibitemOpen
  \bibfield  {author} {\bibinfo {author} {\bibfnamefont {M.}~\bibnamefont {Schreiber}}, \bibinfo {author} {\bibfnamefont {S.~S.}\ \bibnamefont {Hodgman}}, \bibinfo {author} {\bibfnamefont {P.}~\bibnamefont {Bordia}}, \bibinfo {author} {\bibfnamefont {H.~P.}\ \bibnamefont {L{\"u}schen}}, \bibinfo {author} {\bibfnamefont {M.~H.}\ \bibnamefont {Fischer}}, \bibinfo {author} {\bibfnamefont {R.}~\bibnamefont {Vosk}}, \bibinfo {author} {\bibfnamefont {E.}~\bibnamefont {Altman}}, \bibinfo {author} {\bibfnamefont {U.}~\bibnamefont {Schneider}},\ and\ \bibinfo {author} {\bibfnamefont {I.}~\bibnamefont {Bloch}},\ }\bibfield  {title} {\bibinfo {title} {Observation of many-body localization of interacting fermions in a quasirandom optical lattice},\ }\href@noop {} {\bibfield  {journal} {\bibinfo  {journal} {Science}\ }\textbf {\bibinfo {volume} {349}},\ \bibinfo {pages} {842} (\bibinfo {year} {2015})}\BibitemShut {NoStop}%
\bibitem [{\citenamefont {Randall}\ \emph {et~al.}(2021)\citenamefont {Randall}, \citenamefont {Bradley}, \citenamefont {Van Der~Gronden}, \citenamefont {Galicia}, \citenamefont {Abobeih}, \citenamefont {Markham}, \citenamefont {Twitchen}, \citenamefont {Machado}, \citenamefont {Yao},\ and\ \citenamefont {Taminiau}}]{randall2021many}%
  \BibitemOpen
  \bibfield  {author} {\bibinfo {author} {\bibfnamefont {J.}~\bibnamefont {Randall}}, \bibinfo {author} {\bibfnamefont {C.}~\bibnamefont {Bradley}}, \bibinfo {author} {\bibfnamefont {F.}~\bibnamefont {Van Der~Gronden}}, \bibinfo {author} {\bibfnamefont {A.}~\bibnamefont {Galicia}}, \bibinfo {author} {\bibfnamefont {M.}~\bibnamefont {Abobeih}}, \bibinfo {author} {\bibfnamefont {M.}~\bibnamefont {Markham}}, \bibinfo {author} {\bibfnamefont {D.}~\bibnamefont {Twitchen}}, \bibinfo {author} {\bibfnamefont {F.}~\bibnamefont {Machado}}, \bibinfo {author} {\bibfnamefont {N.}~\bibnamefont {Yao}},\ and\ \bibinfo {author} {\bibfnamefont {T.}~\bibnamefont {Taminiau}},\ }\bibfield  {title} {\bibinfo {title} {Many-body--localized discrete time crystal with a programmable spin-based quantum simulator},\ }\href@noop {} {\bibfield  {journal} {\bibinfo  {journal} {Science}\ }\textbf {\bibinfo {volume} {374}},\ \bibinfo {pages} {1474} (\bibinfo {year} {2021})}\BibitemShut {NoStop}%
\bibitem [{\citenamefont {Su}\ \emph {et~al.}(2023)\citenamefont {Su}, \citenamefont {Sun}, \citenamefont {Hudomal}, \citenamefont {Desaules}, \citenamefont {Zhou}, \citenamefont {Yang}, \citenamefont {Halimeh}, \citenamefont {Yuan}, \citenamefont {Papi{\'c}},\ and\ \citenamefont {Pan}}]{su2023observation}%
  \BibitemOpen
  \bibfield  {author} {\bibinfo {author} {\bibfnamefont {G.-X.}\ \bibnamefont {Su}}, \bibinfo {author} {\bibfnamefont {H.}~\bibnamefont {Sun}}, \bibinfo {author} {\bibfnamefont {A.}~\bibnamefont {Hudomal}}, \bibinfo {author} {\bibfnamefont {J.-Y.}\ \bibnamefont {Desaules}}, \bibinfo {author} {\bibfnamefont {Z.-Y.}\ \bibnamefont {Zhou}}, \bibinfo {author} {\bibfnamefont {B.}~\bibnamefont {Yang}}, \bibinfo {author} {\bibfnamefont {J.~C.}\ \bibnamefont {Halimeh}}, \bibinfo {author} {\bibfnamefont {Z.-S.}\ \bibnamefont {Yuan}}, \bibinfo {author} {\bibfnamefont {Z.}~\bibnamefont {Papi{\'c}}},\ and\ \bibinfo {author} {\bibfnamefont {J.-W.}\ \bibnamefont {Pan}},\ }\bibfield  {title} {\bibinfo {title} {Observation of many-body scarring in a bose-hubbard quantum simulator},\ }\href@noop {} {\bibfield  {journal} {\bibinfo  {journal} {Physical Review Research}\ }\textbf {\bibinfo {volume} {5}},\ \bibinfo {pages} {023010} (\bibinfo {year} {2023})}\BibitemShut {NoStop}%
\bibitem [{\citenamefont {Tran}\ \emph {et~al.}(2019)\citenamefont {Tran}, \citenamefont {Guo}, \citenamefont {Su}, \citenamefont {Garrison}, \citenamefont {Eldredge}, \citenamefont {Foss-Feig}, \citenamefont {Childs},\ and\ \citenamefont {Gorshkov}}]{tran2019locality}%
  \BibitemOpen
  \bibfield  {author} {\bibinfo {author} {\bibfnamefont {M.~C.}\ \bibnamefont {Tran}}, \bibinfo {author} {\bibfnamefont {A.~Y.}\ \bibnamefont {Guo}}, \bibinfo {author} {\bibfnamefont {Y.}~\bibnamefont {Su}}, \bibinfo {author} {\bibfnamefont {J.~R.}\ \bibnamefont {Garrison}}, \bibinfo {author} {\bibfnamefont {Z.}~\bibnamefont {Eldredge}}, \bibinfo {author} {\bibfnamefont {M.}~\bibnamefont {Foss-Feig}}, \bibinfo {author} {\bibfnamefont {A.~M.}\ \bibnamefont {Childs}},\ and\ \bibinfo {author} {\bibfnamefont {A.~V.}\ \bibnamefont {Gorshkov}},\ }\bibfield  {title} {\bibinfo {title} {Locality and digital quantum simulation of power-law interactions},\ }\href@noop {} {\bibfield  {journal} {\bibinfo  {journal} {Physical Review X}\ }\textbf {\bibinfo {volume} {9}},\ \bibinfo {pages} {031006} (\bibinfo {year} {2019})}\BibitemShut {NoStop}%
\bibitem [{\citenamefont {Heyl}\ \emph {et~al.}(2019)\citenamefont {Heyl}, \citenamefont {Hauke},\ and\ \citenamefont {Zoller}}]{heylQuantumLocalizationBounds2019}%
  \BibitemOpen
  \bibfield  {author} {\bibinfo {author} {\bibfnamefont {M.}~\bibnamefont {Heyl}}, \bibinfo {author} {\bibfnamefont {P.}~\bibnamefont {Hauke}},\ and\ \bibinfo {author} {\bibfnamefont {P.}~\bibnamefont {Zoller}},\ }\bibfield  {title} {\bibinfo {title} {Quantum localization bounds {{Trotter}} errors in digital quantum simulation},\ }\href {https://doi.org/10.1126/sciadv.aau8342} {\bibfield  {journal} {\bibinfo  {journal} {Sci. Adv.}\ }\textbf {\bibinfo {volume} {5}},\ \bibinfo {pages} {eaau8342} (\bibinfo {year} {2019})},\ \Eprint {https://arxiv.org/abs/1806.11123} {arXiv:1806.11123} \BibitemShut {NoStop}%
\bibitem [{\citenamefont {Jordan}\ \emph {et~al.}(2012)\citenamefont {Jordan}, \citenamefont {Lee},\ and\ \citenamefont {Preskill}}]{jordanQuantumAlgorithmsQuantum2012}%
  \BibitemOpen
  \bibfield  {author} {\bibinfo {author} {\bibfnamefont {S.~P.}\ \bibnamefont {Jordan}}, \bibinfo {author} {\bibfnamefont {K.~S.~M.}\ \bibnamefont {Lee}},\ and\ \bibinfo {author} {\bibfnamefont {J.}~\bibnamefont {Preskill}},\ }\bibfield  {title} {\bibinfo {title} {Quantum {{Algorithms}} for {{Quantum Field Theories}}},\ }\href {https://doi.org/10.1126/science.1217069} {\bibfield  {journal} {\bibinfo  {journal} {Science}\ }\textbf {\bibinfo {volume} {336}},\ \bibinfo {pages} {1130} (\bibinfo {year} {2012})},\ \Eprint {https://arxiv.org/abs/1111.3633} {arXiv:1111.3633} \BibitemShut {NoStop}%
\bibitem [{\citenamefont {Jordan}\ \emph {et~al.}(2018)\citenamefont {Jordan}, \citenamefont {Krovi}, \citenamefont {Lee},\ and\ \citenamefont {Preskill}}]{jordanBQPcompletenessScatteringScalar2018}%
  \BibitemOpen
  \bibfield  {author} {\bibinfo {author} {\bibfnamefont {S.~P.}\ \bibnamefont {Jordan}}, \bibinfo {author} {\bibfnamefont {H.}~\bibnamefont {Krovi}}, \bibinfo {author} {\bibfnamefont {K.~S.~M.}\ \bibnamefont {Lee}},\ and\ \bibinfo {author} {\bibfnamefont {J.}~\bibnamefont {Preskill}},\ }\bibfield  {title} {\bibinfo {title} {{{BQP-completeness}} of {{Scattering}} in {{Scalar Quantum Field Theory}}},\ }\href {https://doi.org/10.22331/q-2018-01-08-44} {\bibfield  {journal} {\bibinfo  {journal} {Quantum}\ }\textbf {\bibinfo {volume} {2}},\ \bibinfo {pages} {44} (\bibinfo {year} {2018})},\ \Eprint {https://arxiv.org/abs/1703.00454} {arXiv:1703.00454} \BibitemShut {NoStop}%
\bibitem [{\citenamefont {Babbush}\ \emph {et~al.}(2018)\citenamefont {Babbush}, \citenamefont {Wiebe}, \citenamefont {McClean}, \citenamefont {McClain}, \citenamefont {Neven},\ and\ \citenamefont {Chan}}]{babbushLowDepthQuantumSimulation2018}%
  \BibitemOpen
  \bibfield  {author} {\bibinfo {author} {\bibfnamefont {R.}~\bibnamefont {Babbush}}, \bibinfo {author} {\bibfnamefont {N.}~\bibnamefont {Wiebe}}, \bibinfo {author} {\bibfnamefont {J.}~\bibnamefont {McClean}}, \bibinfo {author} {\bibfnamefont {J.}~\bibnamefont {McClain}}, \bibinfo {author} {\bibfnamefont {H.}~\bibnamefont {Neven}},\ and\ \bibinfo {author} {\bibfnamefont {G.~K.-L.}\ \bibnamefont {Chan}},\ }\bibfield  {title} {\bibinfo {title} {Low-{{Depth Quantum Simulation}} of {{Materials}}},\ }\href {https://doi.org/10.1103/PhysRevX.8.011044} {\bibfield  {journal} {\bibinfo  {journal} {Phys. Rev. X}\ }\textbf {\bibinfo {volume} {8}},\ \bibinfo {pages} {011044} (\bibinfo {year} {2018})}\BibitemShut {NoStop}%
\bibitem [{\citenamefont {McArdle}\ \emph {et~al.}(2020)\citenamefont {McArdle}, \citenamefont {Endo}, \citenamefont {{Aspuru-Guzik}}, \citenamefont {Benjamin},\ and\ \citenamefont {Yuan}}]{mcardleQuantumComputationalChemistry2020}%
  \BibitemOpen
  \bibfield  {author} {\bibinfo {author} {\bibfnamefont {S.}~\bibnamefont {McArdle}}, \bibinfo {author} {\bibfnamefont {S.}~\bibnamefont {Endo}}, \bibinfo {author} {\bibfnamefont {A.}~\bibnamefont {{Aspuru-Guzik}}}, \bibinfo {author} {\bibfnamefont {S.~C.}\ \bibnamefont {Benjamin}},\ and\ \bibinfo {author} {\bibfnamefont {X.}~\bibnamefont {Yuan}},\ }\bibfield  {title} {\bibinfo {title} {Quantum computational chemistry},\ }\href {https://doi.org/10.1103/RevModPhys.92.015003} {\bibfield  {journal} {\bibinfo  {journal} {Rev. Mod. Phys.}\ }\textbf {\bibinfo {volume} {92}},\ \bibinfo {pages} {015003} (\bibinfo {year} {2020})},\ \Eprint {https://arxiv.org/abs/1808.10402} {arXiv:1808.10402} \BibitemShut {NoStop}%
\bibitem [{\citenamefont {Lee}\ \emph {et~al.}(2023)\citenamefont {Lee}, \citenamefont {Lee}, \citenamefont {Zhai}, \citenamefont {Tong}, \citenamefont {Dalzell}, \citenamefont {Kumar}, \citenamefont {Helms}, \citenamefont {Gray}, \citenamefont {Cui}, \citenamefont {Liu}, \citenamefont {Kastoryano}, \citenamefont {Babbush}, \citenamefont {Preskill}, \citenamefont {Reichman}, \citenamefont {Campbell}, \citenamefont {Valeev}, \citenamefont {Lin},\ and\ \citenamefont {Chan}}]{leeEvaluatingEvidenceExponential2023}%
  \BibitemOpen
  \bibfield  {author} {\bibinfo {author} {\bibfnamefont {S.}~\bibnamefont {Lee}}, \bibinfo {author} {\bibfnamefont {J.}~\bibnamefont {Lee}}, \bibinfo {author} {\bibfnamefont {H.}~\bibnamefont {Zhai}}, \bibinfo {author} {\bibfnamefont {Y.}~\bibnamefont {Tong}}, \bibinfo {author} {\bibfnamefont {A.~M.}\ \bibnamefont {Dalzell}}, \bibinfo {author} {\bibfnamefont {A.}~\bibnamefont {Kumar}}, \bibinfo {author} {\bibfnamefont {P.}~\bibnamefont {Helms}}, \bibinfo {author} {\bibfnamefont {J.}~\bibnamefont {Gray}}, \bibinfo {author} {\bibfnamefont {Z.-H.}\ \bibnamefont {Cui}}, \bibinfo {author} {\bibfnamefont {W.}~\bibnamefont {Liu}}, \bibinfo {author} {\bibfnamefont {M.}~\bibnamefont {Kastoryano}}, \bibinfo {author} {\bibfnamefont {R.}~\bibnamefont {Babbush}}, \bibinfo {author} {\bibfnamefont {J.}~\bibnamefont {Preskill}}, \bibinfo {author} {\bibfnamefont {D.~R.}\ \bibnamefont {Reichman}}, \bibinfo {author} {\bibfnamefont {E.~T.}\ \bibnamefont {Campbell}}, \bibinfo {author} {\bibfnamefont {E.~F.}\ \bibnamefont
  {Valeev}}, \bibinfo {author} {\bibfnamefont {L.}~\bibnamefont {Lin}},\ and\ \bibinfo {author} {\bibfnamefont {G.~K.-L.}\ \bibnamefont {Chan}},\ }\bibfield  {title} {\bibinfo {title} {Evaluating the evidence for exponential quantum advantage in ground-state quantum chemistry},\ }\href {https://doi.org/10.1038/s41467-023-37587-6} {\bibfield  {journal} {\bibinfo  {journal} {Nat Commun}\ }\textbf {\bibinfo {volume} {14}},\ \bibinfo {pages} {1952} (\bibinfo {year} {2023})},\ \Eprint {https://arxiv.org/abs/2208.02199} {arXiv:2208.02199} \BibitemShut {NoStop}%
\bibitem [{\citenamefont {Lloyd}(1996)}]{lloyd1996universal}%
  \BibitemOpen
  \bibfield  {author} {\bibinfo {author} {\bibfnamefont {S.}~\bibnamefont {Lloyd}},\ }\bibfield  {title} {\bibinfo {title} {Universal quantum simulators},\ }\href@noop {} {\bibfield  {journal} {\bibinfo  {journal} {Science}\ }\textbf {\bibinfo {volume} {273}},\ \bibinfo {pages} {1073} (\bibinfo {year} {1996})}\BibitemShut {NoStop}%
\bibitem [{\citenamefont {Aharonov}\ and\ \citenamefont {Ta-Shma}(2003)}]{aharonov2003adiabatic}%
  \BibitemOpen
  \bibfield  {author} {\bibinfo {author} {\bibfnamefont {D.}~\bibnamefont {Aharonov}}\ and\ \bibinfo {author} {\bibfnamefont {A.}~\bibnamefont {Ta-Shma}},\ }\bibfield  {title} {\bibinfo {title} {Adiabatic quantum state generation and statistical zero knowledge},\ }in\ \href@noop {} {\emph {\bibinfo {booktitle} {Proceedings of the thirty-fifth annual ACM symposium on Theory of computing}}}\ (\bibinfo {year} {2003})\ pp.\ \bibinfo {pages} {20--29}\BibitemShut {NoStop}%
\bibitem [{\citenamefont {Childs}(2004)}]{childs2004quantum}%
  \BibitemOpen
  \bibfield  {author} {\bibinfo {author} {\bibfnamefont {A.~M.}\ \bibnamefont {Childs}},\ }\emph {\bibinfo {title} {Quantum information processing in continuous time}},\ \href@noop {} {Ph.D. thesis},\ \bibinfo  {school} {Massachusetts Institute of Technology} (\bibinfo {year} {2004})\BibitemShut {NoStop}%
\bibitem [{\citenamefont {Berry}\ \emph {et~al.}(2007)\citenamefont {Berry}, \citenamefont {Ahokas}, \citenamefont {Cleve},\ and\ \citenamefont {Sanders}}]{berry2007efficient}%
  \BibitemOpen
  \bibfield  {author} {\bibinfo {author} {\bibfnamefont {D.~W.}\ \bibnamefont {Berry}}, \bibinfo {author} {\bibfnamefont {G.}~\bibnamefont {Ahokas}}, \bibinfo {author} {\bibfnamefont {R.}~\bibnamefont {Cleve}},\ and\ \bibinfo {author} {\bibfnamefont {B.~C.}\ \bibnamefont {Sanders}},\ }\bibfield  {title} {\bibinfo {title} {Efficient quantum algorithms for simulating sparse hamiltonians},\ }\href@noop {} {\bibfield  {journal} {\bibinfo  {journal} {Communications in Mathematical Physics}\ }\textbf {\bibinfo {volume} {270}},\ \bibinfo {pages} {359} (\bibinfo {year} {2007})}\BibitemShut {NoStop}%
\bibitem [{\citenamefont {Somma}(2016)}]{somma2016trotter}%
  \BibitemOpen
  \bibfield  {author} {\bibinfo {author} {\bibfnamefont {R.~D.}\ \bibnamefont {Somma}},\ }\bibfield  {title} {\bibinfo {title} {A trotter-suzuki approximation for lie groups with applications to hamiltonian simulation},\ }\href@noop {} {\bibfield  {journal} {\bibinfo  {journal} {Journal of Mathematical Physics}\ }\textbf {\bibinfo {volume} {57}} (\bibinfo {year} {2016})}\BibitemShut {NoStop}%
\bibitem [{\citenamefont {Childs}\ and\ \citenamefont {Wiebe}(2012)}]{childsHamiltonianSimulationUsing2012}%
  \BibitemOpen
  \bibfield  {author} {\bibinfo {author} {\bibfnamefont {A.~M.}\ \bibnamefont {Childs}}\ and\ \bibinfo {author} {\bibfnamefont {N.}~\bibnamefont {Wiebe}},\ }\bibfield  {title} {\bibinfo {title} {Hamiltonian {{Simulation Using Linear Combinations}} of {{Unitary Operations}}},\ }\bibfield  {journal} {\bibinfo  {journal} {QIC}\ }\textbf {\bibinfo {volume} {12}},\ \href {https://doi.org/10.26421/QIC12.11-12} {10.26421/QIC12.11-12} (\bibinfo {year} {2012}),\ \Eprint {https://arxiv.org/abs/1202.5822} {arXiv:1202.5822} \BibitemShut {NoStop}%
\bibitem [{\citenamefont {Berry}\ \emph {et~al.}(2014)\citenamefont {Berry}, \citenamefont {Childs}, \citenamefont {Cleve}, \citenamefont {Kothari},\ and\ \citenamefont {Somma}}]{berry2014exponential}%
  \BibitemOpen
  \bibfield  {author} {\bibinfo {author} {\bibfnamefont {D.~W.}\ \bibnamefont {Berry}}, \bibinfo {author} {\bibfnamefont {A.~M.}\ \bibnamefont {Childs}}, \bibinfo {author} {\bibfnamefont {R.}~\bibnamefont {Cleve}}, \bibinfo {author} {\bibfnamefont {R.}~\bibnamefont {Kothari}},\ and\ \bibinfo {author} {\bibfnamefont {R.~D.}\ \bibnamefont {Somma}},\ }\bibfield  {title} {\bibinfo {title} {Exponential improvement in precision for simulating sparse hamiltonians},\ }in\ \href@noop {} {\emph {\bibinfo {booktitle} {Proceedings of the forty-sixth annual ACM symposium on Theory of computing}}}\ (\bibinfo {year} {2014})\ pp.\ \bibinfo {pages} {283--292}\BibitemShut {NoStop}%
\bibitem [{\citenamefont {Berry}\ \emph {et~al.}(2015{\natexlab{a}})\citenamefont {Berry}, \citenamefont {Childs}, \citenamefont {Cleve}, \citenamefont {Kothari},\ and\ \citenamefont {Somma}}]{Berry2015}%
  \BibitemOpen
  \bibfield  {author} {\bibinfo {author} {\bibfnamefont {D.~W.}\ \bibnamefont {Berry}}, \bibinfo {author} {\bibfnamefont {A.~M.}\ \bibnamefont {Childs}}, \bibinfo {author} {\bibfnamefont {R.}~\bibnamefont {Cleve}}, \bibinfo {author} {\bibfnamefont {R.}~\bibnamefont {Kothari}},\ and\ \bibinfo {author} {\bibfnamefont {R.~D.}\ \bibnamefont {Somma}},\ }\bibfield  {title} {\bibinfo {title} {Simulating hamiltonian dynamics with a truncated taylor series},\ }\href {https://doi.org/10.1103/PhysRevLett.114.090502} {\bibfield  {journal} {\bibinfo  {journal} {Phys. Rev. Lett.}\ }\textbf {\bibinfo {volume} {114}},\ \bibinfo {pages} {090502} (\bibinfo {year} {2015}{\natexlab{a}})}\BibitemShut {NoStop}%
\bibitem [{\citenamefont {Berry}\ \emph {et~al.}(2015{\natexlab{b}})\citenamefont {Berry}, \citenamefont {Childs},\ and\ \citenamefont {Kothari}}]{berry2015hamiltonian}%
  \BibitemOpen
  \bibfield  {author} {\bibinfo {author} {\bibfnamefont {D.~W.}\ \bibnamefont {Berry}}, \bibinfo {author} {\bibfnamefont {A.~M.}\ \bibnamefont {Childs}},\ and\ \bibinfo {author} {\bibfnamefont {R.}~\bibnamefont {Kothari}},\ }\bibfield  {title} {\bibinfo {title} {Hamiltonian simulation with nearly optimal dependence on all parameters},\ }in\ \href@noop {} {\emph {\bibinfo {booktitle} {2015 IEEE 56th annual symposium on foundations of computer science}}}\ (\bibinfo {organization} {IEEE},\ \bibinfo {year} {2015})\ pp.\ \bibinfo {pages} {792--809}\BibitemShut {NoStop}%
\bibitem [{\citenamefont {Haah}\ \emph {et~al.}(2021)\citenamefont {Haah}, \citenamefont {Hastings}, \citenamefont {Kothari},\ and\ \citenamefont {Low}}]{haah2021quantum}%
  \BibitemOpen
  \bibfield  {author} {\bibinfo {author} {\bibfnamefont {J.}~\bibnamefont {Haah}}, \bibinfo {author} {\bibfnamefont {M.~B.}\ \bibnamefont {Hastings}}, \bibinfo {author} {\bibfnamefont {R.}~\bibnamefont {Kothari}},\ and\ \bibinfo {author} {\bibfnamefont {G.~H.}\ \bibnamefont {Low}},\ }\bibfield  {title} {\bibinfo {title} {Quantum algorithm for simulating real time evolution of lattice hamiltonians},\ }\href@noop {} {\bibfield  {journal} {\bibinfo  {journal} {SIAM Journal on Computing}\ }\textbf {\bibinfo {volume} {52}},\ \bibinfo {pages} {FOCS18} (\bibinfo {year} {2021})}\BibitemShut {NoStop}%
\bibitem [{\citenamefont {Low}\ and\ \citenamefont {Chuang}(2017)}]{Low2017optimal}%
  \BibitemOpen
  \bibfield  {author} {\bibinfo {author} {\bibfnamefont {G.~H.}\ \bibnamefont {Low}}\ and\ \bibinfo {author} {\bibfnamefont {I.~L.}\ \bibnamefont {Chuang}},\ }\bibfield  {title} {\bibinfo {title} {Optimal hamiltonian simulation by quantum signal processing},\ }\href {https://doi.org/10.1103/PhysRevLett.118.010501} {\bibfield  {journal} {\bibinfo  {journal} {Phys. Rev. Lett.}\ }\textbf {\bibinfo {volume} {118}},\ \bibinfo {pages} {010501} (\bibinfo {year} {2017})}\BibitemShut {NoStop}%
\bibitem [{\citenamefont {Low}\ and\ \citenamefont {Chuang}(2019)}]{low2019hamiltonian}%
  \BibitemOpen
  \bibfield  {author} {\bibinfo {author} {\bibfnamefont {G.~H.}\ \bibnamefont {Low}}\ and\ \bibinfo {author} {\bibfnamefont {I.~L.}\ \bibnamefont {Chuang}},\ }\bibfield  {title} {\bibinfo {title} {Hamiltonian simulation by qubitization},\ }\href@noop {} {\bibfield  {journal} {\bibinfo  {journal} {Quantum}\ }\textbf {\bibinfo {volume} {3}},\ \bibinfo {pages} {163} (\bibinfo {year} {2019})}\BibitemShut {NoStop}%
\bibitem [{\citenamefont {Childs}\ \emph {et~al.}(2018)\citenamefont {Childs}, \citenamefont {Maslov}, \citenamefont {Nam}, \citenamefont {Ross},\ and\ \citenamefont {Su}}]{childs2018toward}%
  \BibitemOpen
  \bibfield  {author} {\bibinfo {author} {\bibfnamefont {A.~M.}\ \bibnamefont {Childs}}, \bibinfo {author} {\bibfnamefont {D.}~\bibnamefont {Maslov}}, \bibinfo {author} {\bibfnamefont {Y.}~\bibnamefont {Nam}}, \bibinfo {author} {\bibfnamefont {N.~J.}\ \bibnamefont {Ross}},\ and\ \bibinfo {author} {\bibfnamefont {Y.}~\bibnamefont {Su}},\ }\bibfield  {title} {\bibinfo {title} {Toward the first quantum simulation with quantum speedup},\ }\href@noop {} {\bibfield  {journal} {\bibinfo  {journal} {Proceedings of the National Academy of Sciences}\ }\textbf {\bibinfo {volume} {115}},\ \bibinfo {pages} {9456} (\bibinfo {year} {2018})}\BibitemShut {NoStop}%
\bibitem [{\citenamefont {Brown}\ \emph {et~al.}(2006)\citenamefont {Brown}, \citenamefont {Clark},\ and\ \citenamefont {Chuang}}]{Brown2006}%
  \BibitemOpen
  \bibfield  {author} {\bibinfo {author} {\bibfnamefont {K.~R.}\ \bibnamefont {Brown}}, \bibinfo {author} {\bibfnamefont {R.~J.}\ \bibnamefont {Clark}},\ and\ \bibinfo {author} {\bibfnamefont {I.~L.}\ \bibnamefont {Chuang}},\ }\bibfield  {title} {\bibinfo {title} {Limitations of quantum simulation examined by simulating a pairing hamiltonian using nuclear magnetic resonance},\ }\href {https://doi.org/10.1103/PhysRevLett.97.050504} {\bibfield  {journal} {\bibinfo  {journal} {Phys. Rev. Lett.}\ }\textbf {\bibinfo {volume} {97}},\ \bibinfo {pages} {050504} (\bibinfo {year} {2006})}\BibitemShut {NoStop}%
\bibitem [{\citenamefont {Lanyon}\ \emph {et~al.}(2011)\citenamefont {Lanyon}, \citenamefont {Hempel}, \citenamefont {Nigg}, \citenamefont {M{\"u}ller}, \citenamefont {Gerritsma}, \citenamefont {Z{\"a}hringer}, \citenamefont {Schindler}, \citenamefont {Barreiro}, \citenamefont {Rambach}, \citenamefont {Kirchmair} \emph {et~al.}}]{lanyon2011universal}%
  \BibitemOpen
  \bibfield  {author} {\bibinfo {author} {\bibfnamefont {B.~P.}\ \bibnamefont {Lanyon}}, \bibinfo {author} {\bibfnamefont {C.}~\bibnamefont {Hempel}}, \bibinfo {author} {\bibfnamefont {D.}~\bibnamefont {Nigg}}, \bibinfo {author} {\bibfnamefont {M.}~\bibnamefont {M{\"u}ller}}, \bibinfo {author} {\bibfnamefont {R.}~\bibnamefont {Gerritsma}}, \bibinfo {author} {\bibfnamefont {F.}~\bibnamefont {Z{\"a}hringer}}, \bibinfo {author} {\bibfnamefont {P.}~\bibnamefont {Schindler}}, \bibinfo {author} {\bibfnamefont {J.~T.}\ \bibnamefont {Barreiro}}, \bibinfo {author} {\bibfnamefont {M.}~\bibnamefont {Rambach}}, \bibinfo {author} {\bibfnamefont {G.}~\bibnamefont {Kirchmair}}, \emph {et~al.},\ }\bibfield  {title} {\bibinfo {title} {Universal digital quantum simulation with trapped ions},\ }\href@noop {} {\bibfield  {journal} {\bibinfo  {journal} {Science}\ }\textbf {\bibinfo {volume} {334}},\ \bibinfo {pages} {57} (\bibinfo {year} {2011})}\BibitemShut {NoStop}%
\bibitem [{\citenamefont {Lv}\ \emph {et~al.}(2018)\citenamefont {Lv}, \citenamefont {An}, \citenamefont {Liu}, \citenamefont {Zhang}, \citenamefont {Pedernales}, \citenamefont {Lamata}, \citenamefont {Solano},\ and\ \citenamefont {Kim}}]{lv2018quantum}%
  \BibitemOpen
  \bibfield  {author} {\bibinfo {author} {\bibfnamefont {D.}~\bibnamefont {Lv}}, \bibinfo {author} {\bibfnamefont {S.}~\bibnamefont {An}}, \bibinfo {author} {\bibfnamefont {Z.}~\bibnamefont {Liu}}, \bibinfo {author} {\bibfnamefont {J.-N.}\ \bibnamefont {Zhang}}, \bibinfo {author} {\bibfnamefont {J.~S.}\ \bibnamefont {Pedernales}}, \bibinfo {author} {\bibfnamefont {L.}~\bibnamefont {Lamata}}, \bibinfo {author} {\bibfnamefont {E.}~\bibnamefont {Solano}},\ and\ \bibinfo {author} {\bibfnamefont {K.}~\bibnamefont {Kim}},\ }\bibfield  {title} {\bibinfo {title} {Quantum simulation of the quantum rabi model in a trapped ion},\ }\href@noop {} {\bibfield  {journal} {\bibinfo  {journal} {Physical Review X}\ }\textbf {\bibinfo {volume} {8}},\ \bibinfo {pages} {021027} (\bibinfo {year} {2018})}\BibitemShut {NoStop}%
\bibitem [{\citenamefont {Jafferis}\ \emph {et~al.}(2022)\citenamefont {Jafferis}, \citenamefont {Zlokapa}, \citenamefont {Lykken}, \citenamefont {Kolchmeyer}, \citenamefont {Davis}, \citenamefont {Lauk}, \citenamefont {Neven},\ and\ \citenamefont {Spiropulu}}]{jafferisTraversableWormholeDynamics2022}%
  \BibitemOpen
  \bibfield  {author} {\bibinfo {author} {\bibfnamefont {D.}~\bibnamefont {Jafferis}}, \bibinfo {author} {\bibfnamefont {A.}~\bibnamefont {Zlokapa}}, \bibinfo {author} {\bibfnamefont {J.~D.}\ \bibnamefont {Lykken}}, \bibinfo {author} {\bibfnamefont {D.~K.}\ \bibnamefont {Kolchmeyer}}, \bibinfo {author} {\bibfnamefont {S.~I.}\ \bibnamefont {Davis}}, \bibinfo {author} {\bibfnamefont {N.}~\bibnamefont {Lauk}}, \bibinfo {author} {\bibfnamefont {H.}~\bibnamefont {Neven}},\ and\ \bibinfo {author} {\bibfnamefont {M.}~\bibnamefont {Spiropulu}},\ }\bibfield  {title} {\bibinfo {title} {Traversable wormhole dynamics on a quantum processor},\ }\href {https://doi.org/10.1038/s41586-022-05424-3} {\bibfield  {journal} {\bibinfo  {journal} {Nature}\ }\textbf {\bibinfo {volume} {612}},\ \bibinfo {pages} {51} (\bibinfo {year} {2022})}\BibitemShut {NoStop}%
\bibitem [{\citenamefont {Kim}\ \emph {et~al.}(2023)\citenamefont {Kim}, \citenamefont {Eddins}, \citenamefont {Anand}, \citenamefont {Wei}, \citenamefont {{van den Berg}}, \citenamefont {Rosenblatt}, \citenamefont {Nayfeh}, \citenamefont {Wu}, \citenamefont {Zaletel}, \citenamefont {Temme},\ and\ \citenamefont {Kandala}}]{kimEvidenceUtilityQuantum2023}%
  \BibitemOpen
  \bibfield  {author} {\bibinfo {author} {\bibfnamefont {Y.}~\bibnamefont {Kim}}, \bibinfo {author} {\bibfnamefont {A.}~\bibnamefont {Eddins}}, \bibinfo {author} {\bibfnamefont {S.}~\bibnamefont {Anand}}, \bibinfo {author} {\bibfnamefont {K.~X.}\ \bibnamefont {Wei}}, \bibinfo {author} {\bibfnamefont {E.}~\bibnamefont {{van den Berg}}}, \bibinfo {author} {\bibfnamefont {S.}~\bibnamefont {Rosenblatt}}, \bibinfo {author} {\bibfnamefont {H.}~\bibnamefont {Nayfeh}}, \bibinfo {author} {\bibfnamefont {Y.}~\bibnamefont {Wu}}, \bibinfo {author} {\bibfnamefont {M.}~\bibnamefont {Zaletel}}, \bibinfo {author} {\bibfnamefont {K.}~\bibnamefont {Temme}},\ and\ \bibinfo {author} {\bibfnamefont {A.}~\bibnamefont {Kandala}},\ }\bibfield  {title} {\bibinfo {title} {Evidence for the utility of quantum computing before fault tolerance},\ }\href {https://doi.org/10.1038/s41586-023-06096-3} {\bibfield  {journal} {\bibinfo  {journal} {Nature}\ }\textbf {\bibinfo {volume} {618}},\ \bibinfo {pages} {500} (\bibinfo {year}
  {2023})}\BibitemShut {NoStop}%
\bibitem [{\citenamefont {Childs}\ and\ \citenamefont {Su}(2019)}]{childs2019nearly}%
  \BibitemOpen
  \bibfield  {author} {\bibinfo {author} {\bibfnamefont {A.~M.}\ \bibnamefont {Childs}}\ and\ \bibinfo {author} {\bibfnamefont {Y.}~\bibnamefont {Su}},\ }\bibfield  {title} {\bibinfo {title} {Nearly optimal lattice simulation by product formulas},\ }\href@noop {} {\bibfield  {journal} {\bibinfo  {journal} {Physical review letters}\ }\textbf {\bibinfo {volume} {123}},\ \bibinfo {pages} {050503} (\bibinfo {year} {2019})}\BibitemShut {NoStop}%
\bibitem [{\citenamefont {Childs}\ \emph {et~al.}(2021)\citenamefont {Childs}, \citenamefont {Su}, \citenamefont {Tran}, \citenamefont {Wiebe},\ and\ \citenamefont {Zhu}}]{childs2021theory}%
  \BibitemOpen
  \bibfield  {author} {\bibinfo {author} {\bibfnamefont {A.~M.}\ \bibnamefont {Childs}}, \bibinfo {author} {\bibfnamefont {Y.}~\bibnamefont {Su}}, \bibinfo {author} {\bibfnamefont {M.~C.}\ \bibnamefont {Tran}}, \bibinfo {author} {\bibfnamefont {N.}~\bibnamefont {Wiebe}},\ and\ \bibinfo {author} {\bibfnamefont {S.}~\bibnamefont {Zhu}},\ }\bibfield  {title} {\bibinfo {title} {Theory of trotter error with commutator scaling},\ }\href@noop {} {\bibfield  {journal} {\bibinfo  {journal} {Physical Review X}\ }\textbf {\bibinfo {volume} {11}},\ \bibinfo {pages} {011020} (\bibinfo {year} {2021})}\BibitemShut {NoStop}%
\bibitem [{\citenamefont {Li}\ \emph {et~al.}(2017)\citenamefont {Li}, \citenamefont {Fan}, \citenamefont {Wang}, \citenamefont {Ye}, \citenamefont {Zeng}, \citenamefont {Zhai}, \citenamefont {Peng},\ and\ \citenamefont {Du}}]{li2017measuring}%
  \BibitemOpen
  \bibfield  {author} {\bibinfo {author} {\bibfnamefont {J.}~\bibnamefont {Li}}, \bibinfo {author} {\bibfnamefont {R.}~\bibnamefont {Fan}}, \bibinfo {author} {\bibfnamefont {H.}~\bibnamefont {Wang}}, \bibinfo {author} {\bibfnamefont {B.}~\bibnamefont {Ye}}, \bibinfo {author} {\bibfnamefont {B.}~\bibnamefont {Zeng}}, \bibinfo {author} {\bibfnamefont {H.}~\bibnamefont {Zhai}}, \bibinfo {author} {\bibfnamefont {X.}~\bibnamefont {Peng}},\ and\ \bibinfo {author} {\bibfnamefont {J.}~\bibnamefont {Du}},\ }\bibfield  {title} {\bibinfo {title} {Measuring out-of-time-order correlators on a nuclear magnetic resonance quantum simulator},\ }\href@noop {} {\bibfield  {journal} {\bibinfo  {journal} {Physical Review X}\ }\textbf {\bibinfo {volume} {7}},\ \bibinfo {pages} {031011} (\bibinfo {year} {2017})}\BibitemShut {NoStop}%
\bibitem [{\citenamefont {G{\"a}rttner}\ \emph {et~al.}(2017)\citenamefont {G{\"a}rttner}, \citenamefont {Bohnet}, \citenamefont {Safavi-Naini}, \citenamefont {Wall}, \citenamefont {Bollinger},\ and\ \citenamefont {Rey}}]{garttner2017measuring}%
  \BibitemOpen
  \bibfield  {author} {\bibinfo {author} {\bibfnamefont {M.}~\bibnamefont {G{\"a}rttner}}, \bibinfo {author} {\bibfnamefont {J.~G.}\ \bibnamefont {Bohnet}}, \bibinfo {author} {\bibfnamefont {A.}~\bibnamefont {Safavi-Naini}}, \bibinfo {author} {\bibfnamefont {M.~L.}\ \bibnamefont {Wall}}, \bibinfo {author} {\bibfnamefont {J.~J.}\ \bibnamefont {Bollinger}},\ and\ \bibinfo {author} {\bibfnamefont {A.~M.}\ \bibnamefont {Rey}},\ }\bibfield  {title} {\bibinfo {title} {Measuring out-of-time-order correlations and multiple quantum spectra in a trapped-ion quantum magnet},\ }\href@noop {} {\bibfield  {journal} {\bibinfo  {journal} {Nature Physics}\ }\textbf {\bibinfo {volume} {13}},\ \bibinfo {pages} {781} (\bibinfo {year} {2017})}\BibitemShut {NoStop}%
\bibitem [{\citenamefont {Landsman}\ \emph {et~al.}(2019)\citenamefont {Landsman}, \citenamefont {Figgatt}, \citenamefont {Schuster}, \citenamefont {Linke}, \citenamefont {Yoshida}, \citenamefont {Yao},\ and\ \citenamefont {Monroe}}]{landsmanVerifiedQuantumInformation2019}%
  \BibitemOpen
  \bibfield  {author} {\bibinfo {author} {\bibfnamefont {K.~A.}\ \bibnamefont {Landsman}}, \bibinfo {author} {\bibfnamefont {C.}~\bibnamefont {Figgatt}}, \bibinfo {author} {\bibfnamefont {T.}~\bibnamefont {Schuster}}, \bibinfo {author} {\bibfnamefont {N.~M.}\ \bibnamefont {Linke}}, \bibinfo {author} {\bibfnamefont {B.}~\bibnamefont {Yoshida}}, \bibinfo {author} {\bibfnamefont {N.~Y.}\ \bibnamefont {Yao}},\ and\ \bibinfo {author} {\bibfnamefont {C.}~\bibnamefont {Monroe}},\ }\bibfield  {title} {\bibinfo {title} {Verified quantum information scrambling},\ }\href {https://doi.org/10.1038/s41586-019-0952-6} {\bibfield  {journal} {\bibinfo  {journal} {Nature}\ }\textbf {\bibinfo {volume} {567}},\ \bibinfo {pages} {61} (\bibinfo {year} {2019})},\ \Eprint {https://arxiv.org/abs/1806.02807} {arXiv:1806.02807} \BibitemShut {NoStop}%
\bibitem [{\citenamefont {Green}\ \emph {et~al.}(2022)\citenamefont {Green}, \citenamefont {Elben}, \citenamefont {Alderete}, \citenamefont {Joshi}, \citenamefont {Nguyen}, \citenamefont {Zache}, \citenamefont {Zhu}, \citenamefont {Sundar},\ and\ \citenamefont {Linke}}]{green2022experimental}%
  \BibitemOpen
  \bibfield  {author} {\bibinfo {author} {\bibfnamefont {A.~M.}\ \bibnamefont {Green}}, \bibinfo {author} {\bibfnamefont {A.}~\bibnamefont {Elben}}, \bibinfo {author} {\bibfnamefont {C.~H.}\ \bibnamefont {Alderete}}, \bibinfo {author} {\bibfnamefont {L.~K.}\ \bibnamefont {Joshi}}, \bibinfo {author} {\bibfnamefont {N.~H.}\ \bibnamefont {Nguyen}}, \bibinfo {author} {\bibfnamefont {T.~V.}\ \bibnamefont {Zache}}, \bibinfo {author} {\bibfnamefont {Y.}~\bibnamefont {Zhu}}, \bibinfo {author} {\bibfnamefont {B.}~\bibnamefont {Sundar}},\ and\ \bibinfo {author} {\bibfnamefont {N.~M.}\ \bibnamefont {Linke}},\ }\bibfield  {title} {\bibinfo {title} {Experimental measurement of out-of-time-ordered correlators at finite temperature},\ }\href@noop {} {\bibfield  {journal} {\bibinfo  {journal} {Physical Review Letters}\ }\textbf {\bibinfo {volume} {128}},\ \bibinfo {pages} {140601} (\bibinfo {year} {2022})}\BibitemShut {NoStop}%
\bibitem [{\citenamefont {Simon}\ \emph {et~al.}(2011)\citenamefont {Simon}, \citenamefont {Bakr}, \citenamefont {Ma}, \citenamefont {Tai}, \citenamefont {Preiss},\ and\ \citenamefont {Greiner}}]{simon2011quantum}%
  \BibitemOpen
  \bibfield  {author} {\bibinfo {author} {\bibfnamefont {J.}~\bibnamefont {Simon}}, \bibinfo {author} {\bibfnamefont {W.~S.}\ \bibnamefont {Bakr}}, \bibinfo {author} {\bibfnamefont {R.}~\bibnamefont {Ma}}, \bibinfo {author} {\bibfnamefont {M.~E.}\ \bibnamefont {Tai}}, \bibinfo {author} {\bibfnamefont {P.~M.}\ \bibnamefont {Preiss}},\ and\ \bibinfo {author} {\bibfnamefont {M.}~\bibnamefont {Greiner}},\ }\bibfield  {title} {\bibinfo {title} {Quantum simulation of antiferromagnetic spin chains in an optical lattice},\ }\href@noop {} {\bibfield  {journal} {\bibinfo  {journal} {Nature}\ }\textbf {\bibinfo {volume} {472}},\ \bibinfo {pages} {307} (\bibinfo {year} {2011})}\BibitemShut {NoStop}%
\bibitem [{\citenamefont {Martinez}\ \emph {et~al.}(2016)\citenamefont {Martinez}, \citenamefont {Muschik}, \citenamefont {Schindler}, \citenamefont {Nigg}, \citenamefont {Erhard}, \citenamefont {Heyl}, \citenamefont {Hauke}, \citenamefont {Dalmonte}, \citenamefont {Monz}, \citenamefont {Zoller},\ and\ \citenamefont {Blatt}}]{martinezRealtimeDynamicsLattice2016}%
  \BibitemOpen
  \bibfield  {author} {\bibinfo {author} {\bibfnamefont {E.~A.}\ \bibnamefont {Martinez}}, \bibinfo {author} {\bibfnamefont {C.~A.}\ \bibnamefont {Muschik}}, \bibinfo {author} {\bibfnamefont {P.}~\bibnamefont {Schindler}}, \bibinfo {author} {\bibfnamefont {D.}~\bibnamefont {Nigg}}, \bibinfo {author} {\bibfnamefont {A.}~\bibnamefont {Erhard}}, \bibinfo {author} {\bibfnamefont {M.}~\bibnamefont {Heyl}}, \bibinfo {author} {\bibfnamefont {P.}~\bibnamefont {Hauke}}, \bibinfo {author} {\bibfnamefont {M.}~\bibnamefont {Dalmonte}}, \bibinfo {author} {\bibfnamefont {T.}~\bibnamefont {Monz}}, \bibinfo {author} {\bibfnamefont {P.}~\bibnamefont {Zoller}},\ and\ \bibinfo {author} {\bibfnamefont {R.}~\bibnamefont {Blatt}},\ }\bibfield  {title} {\bibinfo {title} {Real-time dynamics of lattice gauge theories with a few-qubit quantum computer},\ }\href {https://doi.org/10.1038/nature18318} {\bibfield  {journal} {\bibinfo  {journal} {Nature}\ }\textbf {\bibinfo {volume} {534}},\ \bibinfo {pages} {516} (\bibinfo {year}
  {2016})},\ \Eprint {https://arxiv.org/abs/1605.04570} {arXiv:1605.04570} \BibitemShut {NoStop}%
\bibitem [{\citenamefont {Monroe}\ \emph {et~al.}(2021)\citenamefont {Monroe}, \citenamefont {Campbell}, \citenamefont {Duan}, \citenamefont {Gong}, \citenamefont {Gorshkov}, \citenamefont {Hess}, \citenamefont {Islam}, \citenamefont {Kim}, \citenamefont {Linke}, \citenamefont {Pagano}, \citenamefont {Richerme}, \citenamefont {Senko},\ and\ \citenamefont {Yao}}]{monroeProgrammableQuantumSimulations2021}%
  \BibitemOpen
  \bibfield  {author} {\bibinfo {author} {\bibfnamefont {C.}~\bibnamefont {Monroe}}, \bibinfo {author} {\bibfnamefont {W.~C.}\ \bibnamefont {Campbell}}, \bibinfo {author} {\bibfnamefont {L.-M.}\ \bibnamefont {Duan}}, \bibinfo {author} {\bibfnamefont {Z.-X.}\ \bibnamefont {Gong}}, \bibinfo {author} {\bibfnamefont {A.~V.}\ \bibnamefont {Gorshkov}}, \bibinfo {author} {\bibfnamefont {P.}~\bibnamefont {Hess}}, \bibinfo {author} {\bibfnamefont {R.}~\bibnamefont {Islam}}, \bibinfo {author} {\bibfnamefont {K.}~\bibnamefont {Kim}}, \bibinfo {author} {\bibfnamefont {N.}~\bibnamefont {Linke}}, \bibinfo {author} {\bibfnamefont {G.}~\bibnamefont {Pagano}}, \bibinfo {author} {\bibfnamefont {P.}~\bibnamefont {Richerme}}, \bibinfo {author} {\bibfnamefont {C.}~\bibnamefont {Senko}},\ and\ \bibinfo {author} {\bibfnamefont {N.~Y.}\ \bibnamefont {Yao}},\ }\bibfield  {title} {\bibinfo {title} {Programmable {{Quantum Simulations}} of {{Spin Systems}} with {{Trapped Ions}}},\ }\href {https://doi.org/10.1103/RevModPhys.93.025001}
  {\bibfield  {journal} {\bibinfo  {journal} {Rev. Mod. Phys.}\ }\textbf {\bibinfo {volume} {93}},\ \bibinfo {pages} {025001} (\bibinfo {year} {2021})},\ \Eprint {https://arxiv.org/abs/1912.07845} {arXiv:1912.07845} \BibitemShut {NoStop}%
\bibitem [{\citenamefont {Zhao}\ \emph {et~al.}(2022)\citenamefont {Zhao}, \citenamefont {Zhou}, \citenamefont {Shaw}, \citenamefont {Li},\ and\ \citenamefont {Childs}}]{zhao2022hamiltonian}%
  \BibitemOpen
  \bibfield  {author} {\bibinfo {author} {\bibfnamefont {Q.}~\bibnamefont {Zhao}}, \bibinfo {author} {\bibfnamefont {Y.}~\bibnamefont {Zhou}}, \bibinfo {author} {\bibfnamefont {A.~F.}\ \bibnamefont {Shaw}}, \bibinfo {author} {\bibfnamefont {T.}~\bibnamefont {Li}},\ and\ \bibinfo {author} {\bibfnamefont {A.~M.}\ \bibnamefont {Childs}},\ }\bibfield  {title} {\bibinfo {title} {Hamiltonian simulation with random inputs},\ }\href@noop {} {\bibfield  {journal} {\bibinfo  {journal} {Physical Review Letters}\ }\textbf {\bibinfo {volume} {129}},\ \bibinfo {pages} {270502} (\bibinfo {year} {2022})}\BibitemShut {NoStop}%
\bibitem [{\citenamefont {Chen}\ and\ \citenamefont {Brand{\~a}o}(2024)}]{chen2024average}%
  \BibitemOpen
  \bibfield  {author} {\bibinfo {author} {\bibfnamefont {C.-F.}\ \bibnamefont {Chen}}\ and\ \bibinfo {author} {\bibfnamefont {F.~G.}\ \bibnamefont {Brand{\~a}o}},\ }\bibfield  {title} {\bibinfo {title} {Average-case speedup for product formulas},\ }\href@noop {} {\bibfield  {journal} {\bibinfo  {journal} {Communications in Mathematical Physics}\ }\textbf {\bibinfo {volume} {405}},\ \bibinfo {pages} {32} (\bibinfo {year} {2024})}\BibitemShut {NoStop}%
\bibitem [{\citenamefont {Heyl}\ \emph {et~al.}(2013)\citenamefont {Heyl}, \citenamefont {Polkovnikov},\ and\ \citenamefont {Kehrein}}]{heylDynamicalQuantumPhase2013}%
  \BibitemOpen
  \bibfield  {author} {\bibinfo {author} {\bibfnamefont {M.}~\bibnamefont {Heyl}}, \bibinfo {author} {\bibfnamefont {A.}~\bibnamefont {Polkovnikov}},\ and\ \bibinfo {author} {\bibfnamefont {S.}~\bibnamefont {Kehrein}},\ }\bibfield  {title} {\bibinfo {title} {Dynamical {{Quantum Phase Transitions}} in the {{Transverse Field Ising Model}}},\ }\href {https://doi.org/10.1103/PhysRevLett.110.135704} {\bibfield  {journal} {\bibinfo  {journal} {Phys. Rev. Lett.}\ }\textbf {\bibinfo {volume} {110}},\ \bibinfo {pages} {135704} (\bibinfo {year} {2013})},\ \Eprint {https://arxiv.org/abs/1206.2505} {arXiv:1206.2505} \BibitemShut {NoStop}%
\bibitem [{\citenamefont {Heyl}(2018)}]{heylDynamicalQuantumPhase2018}%
  \BibitemOpen
  \bibfield  {author} {\bibinfo {author} {\bibfnamefont {M.}~\bibnamefont {Heyl}},\ }\bibfield  {title} {\bibinfo {title} {Dynamical quantum phase transitions: A review},\ }\href {https://doi.org/10.1088/1361-6633/aaaf9a} {\bibfield  {journal} {\bibinfo  {journal} {Rep. Prog. Phys.}\ }\textbf {\bibinfo {volume} {81}},\ \bibinfo {pages} {054001} (\bibinfo {year} {2018})},\ \Eprint {https://arxiv.org/abs/1709.07461} {arXiv:1709.07461} \BibitemShut {NoStop}%
\bibitem [{\citenamefont {De~Nicola}\ \emph {et~al.}(2021)\citenamefont {De~Nicola}, \citenamefont {Michailidis},\ and\ \citenamefont {Serbyn}}]{denicolaEntanglementViewDynamical2021}%
  \BibitemOpen
  \bibfield  {author} {\bibinfo {author} {\bibfnamefont {S.}~\bibnamefont {De~Nicola}}, \bibinfo {author} {\bibfnamefont {A.~A.}\ \bibnamefont {Michailidis}},\ and\ \bibinfo {author} {\bibfnamefont {M.}~\bibnamefont {Serbyn}},\ }\bibfield  {title} {\bibinfo {title} {Entanglement view of dynamical quantum phase transitions},\ }\href {https://doi.org/10.1103/PhysRevLett.126.040602} {\bibfield  {journal} {\bibinfo  {journal} {Phys. Rev. Lett.}\ }\textbf {\bibinfo {volume} {126}},\ \bibinfo {pages} {040602} (\bibinfo {year} {2021})},\ \Eprint {https://arxiv.org/abs/2008.04894} {arXiv:2008.04894} \BibitemShut {NoStop}%
\bibitem [{\citenamefont {Halimeh}\ \emph {et~al.}(2021)\citenamefont {Halimeh}, \citenamefont {Trapin}, \citenamefont {Van~Damme},\ and\ \citenamefont {Heyl}}]{halimehLocalMeasuresDynamical2021}%
  \BibitemOpen
  \bibfield  {author} {\bibinfo {author} {\bibfnamefont {J.~C.}\ \bibnamefont {Halimeh}}, \bibinfo {author} {\bibfnamefont {D.}~\bibnamefont {Trapin}}, \bibinfo {author} {\bibfnamefont {M.}~\bibnamefont {Van~Damme}},\ and\ \bibinfo {author} {\bibfnamefont {M.}~\bibnamefont {Heyl}},\ }\bibfield  {title} {\bibinfo {title} {Local measures of dynamical quantum phase transitions},\ }\href {https://doi.org/10.1103/PhysRevB.104.075130} {\bibfield  {journal} {\bibinfo  {journal} {Phys. Rev. B}\ }\textbf {\bibinfo {volume} {104}},\ \bibinfo {pages} {075130} (\bibinfo {year} {2021})},\ \Eprint {https://arxiv.org/abs/2010.07307} {arXiv:2010.07307} \BibitemShut {NoStop}%
\bibitem [{\citenamefont {Friis}\ \emph {et~al.}(2018)\citenamefont {Friis}, \citenamefont {Marty}, \citenamefont {Maier}, \citenamefont {Hempel}, \citenamefont {Holz\"apfel}, \citenamefont {Jurcevic}, \citenamefont {Plenio}, \citenamefont {Huber}, \citenamefont {Roos}, \citenamefont {Blatt},\ and\ \citenamefont {Lanyon}}]{PhysRevX.8.021012}%
  \BibitemOpen
  \bibfield  {author} {\bibinfo {author} {\bibfnamefont {N.}~\bibnamefont {Friis}}, \bibinfo {author} {\bibfnamefont {O.}~\bibnamefont {Marty}}, \bibinfo {author} {\bibfnamefont {C.}~\bibnamefont {Maier}}, \bibinfo {author} {\bibfnamefont {C.}~\bibnamefont {Hempel}}, \bibinfo {author} {\bibfnamefont {M.}~\bibnamefont {Holz\"apfel}}, \bibinfo {author} {\bibfnamefont {P.}~\bibnamefont {Jurcevic}}, \bibinfo {author} {\bibfnamefont {M.~B.}\ \bibnamefont {Plenio}}, \bibinfo {author} {\bibfnamefont {M.}~\bibnamefont {Huber}}, \bibinfo {author} {\bibfnamefont {C.}~\bibnamefont {Roos}}, \bibinfo {author} {\bibfnamefont {R.}~\bibnamefont {Blatt}},\ and\ \bibinfo {author} {\bibfnamefont {B.}~\bibnamefont {Lanyon}},\ }\bibfield  {title} {\bibinfo {title} {Observation of entangled states of a fully controlled 20-qubit system},\ }\href {https://doi.org/10.1103/PhysRevX.8.021012} {\bibfield  {journal} {\bibinfo  {journal} {Phys. Rev. X}\ }\textbf {\bibinfo {volume} {8}},\ \bibinfo {pages} {021012} (\bibinfo {year}
  {2018})}\BibitemShut {NoStop}%
\bibitem [{\citenamefont {Zhou}\ \emph {et~al.}(2022)\citenamefont {Zhou}, \citenamefont {Xiao}, \citenamefont {Li}, \citenamefont {Zhao}, \citenamefont {Yuan}, \citenamefont {Ma},\ and\ \citenamefont {Pan}}]{zhou2022scheme}%
  \BibitemOpen
  \bibfield  {author} {\bibinfo {author} {\bibfnamefont {Y.}~\bibnamefont {Zhou}}, \bibinfo {author} {\bibfnamefont {B.}~\bibnamefont {Xiao}}, \bibinfo {author} {\bibfnamefont {M.-D.}\ \bibnamefont {Li}}, \bibinfo {author} {\bibfnamefont {Q.}~\bibnamefont {Zhao}}, \bibinfo {author} {\bibfnamefont {Z.-S.}\ \bibnamefont {Yuan}}, \bibinfo {author} {\bibfnamefont {X.}~\bibnamefont {Ma}},\ and\ \bibinfo {author} {\bibfnamefont {J.-W.}\ \bibnamefont {Pan}},\ }\bibfield  {title} {\bibinfo {title} {A scheme to create and verify scalable entanglement in optical lattice},\ }\href@noop {} {\bibfield  {journal} {\bibinfo  {journal} {npj Quantum Information}\ }\textbf {\bibinfo {volume} {8}},\ \bibinfo {pages} {99} (\bibinfo {year} {2022})}\BibitemShut {NoStop}%
\bibitem [{\citenamefont {Suzuki}(1991)}]{suzuki1991general}%
  \BibitemOpen
  \bibfield  {author} {\bibinfo {author} {\bibfnamefont {M.}~\bibnamefont {Suzuki}},\ }\bibfield  {title} {\bibinfo {title} {{General theory of fractal path integrals with applications to many‐body theories and statistical physics}},\ }\href {https://doi.org/10.1063/1.529425} {\bibfield  {journal} {\bibinfo  {journal} {Journal of Mathematical Physics}\ }\textbf {\bibinfo {volume} {32}},\ \bibinfo {pages} {400} (\bibinfo {year} {1991})}\BibitemShut {NoStop}%
\bibitem [{\citenamefont {Wild}\ and\ \citenamefont {Alhambra}(2023)}]{wildClassicalSimulationShortTime2023}%
  \BibitemOpen
  \bibfield  {author} {\bibinfo {author} {\bibfnamefont {D.~S.}\ \bibnamefont {Wild}}\ and\ \bibinfo {author} {\bibfnamefont {{\'A}.~M.}\ \bibnamefont {Alhambra}},\ }\bibfield  {title} {\bibinfo {title} {Classical {{Simulation}} of {{Short-Time Quantum Dynamics}}},\ }\href {https://doi.org/10.1103/PRXQuantum.4.020340} {\bibfield  {journal} {\bibinfo  {journal} {PRX Quantum}\ }\textbf {\bibinfo {volume} {4}},\ \bibinfo {pages} {020340} (\bibinfo {year} {2023})},\ \Eprint {https://arxiv.org/abs/2210.11490} {arXiv:2210.11490} \BibitemShut {NoStop}%
\bibitem [{\citenamefont {Bravyi}\ \emph {et~al.}(2023)\citenamefont {Bravyi}, \citenamefont {Gosset},\ and\ \citenamefont {Liu}}]{bravyiClassicalSimulationPeaked2023}%
  \BibitemOpen
  \bibfield  {author} {\bibinfo {author} {\bibfnamefont {S.}~\bibnamefont {Bravyi}}, \bibinfo {author} {\bibfnamefont {D.}~\bibnamefont {Gosset}},\ and\ \bibinfo {author} {\bibfnamefont {Y.}~\bibnamefont {Liu}},\ }\href {http://arxiv.org/abs/2309.08405} {\bibinfo {title} {Classical simulation of peaked shallow quantum circuits}} (\bibinfo {year} {2023}),\ \Eprint {https://arxiv.org/abs/2309.08405} {arXiv:2309.08405} \BibitemShut {NoStop}%
\bibitem [{\citenamefont {Bravyi}\ \emph {et~al.}(2021)\citenamefont {Bravyi}, \citenamefont {Gosset},\ and\ \citenamefont {Movassagh}}]{bravyiClassicalAlgorithmsQuantum2021}%
  \BibitemOpen
  \bibfield  {author} {\bibinfo {author} {\bibfnamefont {S.}~\bibnamefont {Bravyi}}, \bibinfo {author} {\bibfnamefont {D.}~\bibnamefont {Gosset}},\ and\ \bibinfo {author} {\bibfnamefont {R.}~\bibnamefont {Movassagh}},\ }\bibfield  {title} {\bibinfo {title} {Classical algorithms for quantum mean values},\ }\href {https://doi.org/10.1038/s41567-020-01109-8} {\bibfield  {journal} {\bibinfo  {journal} {Nat. Phys.}\ }\textbf {\bibinfo {volume} {17}},\ \bibinfo {pages} {337} (\bibinfo {year} {2021})},\ \Eprint {https://arxiv.org/abs/1909.11485} {arXiv:1909.11485} \BibitemShut {NoStop}%
\bibitem [{\citenamefont {Zhang}\ \emph {et~al.}(2017)\citenamefont {Zhang}, \citenamefont {Pagano}, \citenamefont {Hess}, \citenamefont {Kyprianidis}, \citenamefont {Becker}, \citenamefont {Kaplan}, \citenamefont {Gorshkov}, \citenamefont {Gong},\ and\ \citenamefont {Monroe}}]{zhang2017observation}%
  \BibitemOpen
  \bibfield  {author} {\bibinfo {author} {\bibfnamefont {J.}~\bibnamefont {Zhang}}, \bibinfo {author} {\bibfnamefont {G.}~\bibnamefont {Pagano}}, \bibinfo {author} {\bibfnamefont {P.~W.}\ \bibnamefont {Hess}}, \bibinfo {author} {\bibfnamefont {A.}~\bibnamefont {Kyprianidis}}, \bibinfo {author} {\bibfnamefont {P.}~\bibnamefont {Becker}}, \bibinfo {author} {\bibfnamefont {H.}~\bibnamefont {Kaplan}}, \bibinfo {author} {\bibfnamefont {A.~V.}\ \bibnamefont {Gorshkov}}, \bibinfo {author} {\bibfnamefont {Z.-X.}\ \bibnamefont {Gong}},\ and\ \bibinfo {author} {\bibfnamefont {C.}~\bibnamefont {Monroe}},\ }\bibfield  {title} {\bibinfo {title} {Observation of a many-body dynamical phase transition with a 53-qubit quantum simulator},\ }\href@noop {} {\bibfield  {journal} {\bibinfo  {journal} {Nature}\ }\textbf {\bibinfo {volume} {551}},\ \bibinfo {pages} {601} (\bibinfo {year} {2017})}\BibitemShut {NoStop}%
\bibitem [{\citenamefont {Jurcevic}\ \emph {et~al.}(2017)\citenamefont {Jurcevic}, \citenamefont {Shen}, \citenamefont {Hauke}, \citenamefont {Maier}, \citenamefont {Brydges}, \citenamefont {Hempel}, \citenamefont {Lanyon}, \citenamefont {Heyl}, \citenamefont {Blatt},\ and\ \citenamefont {Roos}}]{jurcevicDirectObservationDynamical2017}%
  \BibitemOpen
  \bibfield  {author} {\bibinfo {author} {\bibfnamefont {P.}~\bibnamefont {Jurcevic}}, \bibinfo {author} {\bibfnamefont {H.}~\bibnamefont {Shen}}, \bibinfo {author} {\bibfnamefont {P.}~\bibnamefont {Hauke}}, \bibinfo {author} {\bibfnamefont {C.}~\bibnamefont {Maier}}, \bibinfo {author} {\bibfnamefont {T.}~\bibnamefont {Brydges}}, \bibinfo {author} {\bibfnamefont {C.}~\bibnamefont {Hempel}}, \bibinfo {author} {\bibfnamefont {B.~P.}\ \bibnamefont {Lanyon}}, \bibinfo {author} {\bibfnamefont {M.}~\bibnamefont {Heyl}}, \bibinfo {author} {\bibfnamefont {R.}~\bibnamefont {Blatt}},\ and\ \bibinfo {author} {\bibfnamefont {C.~F.}\ \bibnamefont {Roos}},\ }\bibfield  {title} {\bibinfo {title} {Direct observation of dynamical quantum phase transitions in an interacting many-body system},\ }\href {https://doi.org/10.1103/PhysRevLett.119.080501} {\bibfield  {journal} {\bibinfo  {journal} {Phys. Rev. Lett.}\ }\textbf {\bibinfo {volume} {119}},\ \bibinfo {pages} {080501} (\bibinfo {year} {2017})},\ \Eprint
  {https://arxiv.org/abs/1612.06902} {arXiv:1612.06902} \BibitemShut {NoStop}%
\bibitem [{\citenamefont {Huggins}\ \emph {et~al.}(2022)\citenamefont {Huggins}, \citenamefont {O'Gorman}, \citenamefont {Rubin}, \citenamefont {Reichman}, \citenamefont {Babbush},\ and\ \citenamefont {Lee}}]{hugginsUnbiasingFermionicQuantum2022}%
  \BibitemOpen
  \bibfield  {author} {\bibinfo {author} {\bibfnamefont {W.~J.}\ \bibnamefont {Huggins}}, \bibinfo {author} {\bibfnamefont {B.~A.}\ \bibnamefont {O'Gorman}}, \bibinfo {author} {\bibfnamefont {N.~C.}\ \bibnamefont {Rubin}}, \bibinfo {author} {\bibfnamefont {D.~R.}\ \bibnamefont {Reichman}}, \bibinfo {author} {\bibfnamefont {R.}~\bibnamefont {Babbush}},\ and\ \bibinfo {author} {\bibfnamefont {J.}~\bibnamefont {Lee}},\ }\bibfield  {title} {\bibinfo {title} {Unbiasing {{Fermionic Quantum Monte Carlo}} with a {{Quantum Computer}}},\ }\href {https://doi.org/10.1038/s41586-021-04351-z} {\bibfield  {journal} {\bibinfo  {journal} {Nature}\ }\textbf {\bibinfo {volume} {603}},\ \bibinfo {pages} {416} (\bibinfo {year} {2022})},\ \Eprint {https://arxiv.org/abs/2106.16235} {arXiv:2106.16235} \BibitemShut {NoStop}%
\bibitem [{\citenamefont {Wang}\ \emph {et~al.}(2024)\citenamefont {Wang}, \citenamefont {Liu}, \citenamefont {Chen}, \citenamefont {Chen}, \citenamefont {Zhao}, \citenamefont {Ying}, \citenamefont {Shang}, \citenamefont {Wang}, \citenamefont {Huo}, \citenamefont {Peng}, \citenamefont {Zhu}, \citenamefont {Lu},\ and\ \citenamefont {Pan}}]{wangRealizationFractionalQuantum2024}%
  \BibitemOpen
  \bibfield  {author} {\bibinfo {author} {\bibfnamefont {C.}~\bibnamefont {Wang}}, \bibinfo {author} {\bibfnamefont {F.-M.}\ \bibnamefont {Liu}}, \bibinfo {author} {\bibfnamefont {M.-C.}\ \bibnamefont {Chen}}, \bibinfo {author} {\bibfnamefont {H.}~\bibnamefont {Chen}}, \bibinfo {author} {\bibfnamefont {X.-H.}\ \bibnamefont {Zhao}}, \bibinfo {author} {\bibfnamefont {C.}~\bibnamefont {Ying}}, \bibinfo {author} {\bibfnamefont {Z.-X.}\ \bibnamefont {Shang}}, \bibinfo {author} {\bibfnamefont {J.-W.}\ \bibnamefont {Wang}}, \bibinfo {author} {\bibfnamefont {Y.-H.}\ \bibnamefont {Huo}}, \bibinfo {author} {\bibfnamefont {C.-Z.}\ \bibnamefont {Peng}}, \bibinfo {author} {\bibfnamefont {X.}~\bibnamefont {Zhu}}, \bibinfo {author} {\bibfnamefont {C.-Y.}\ \bibnamefont {Lu}},\ and\ \bibinfo {author} {\bibfnamefont {J.-W.}\ \bibnamefont {Pan}},\ }\bibfield  {title} {\bibinfo {title} {Realization of fractional quantum {{Hall}} state with interacting photons},\ }\href {https://doi.org/10.1126/science.ado3912} {\bibfield
  {journal} {\bibinfo  {journal} {Science}\ }\textbf {\bibinfo {volume} {384}},\ \bibinfo {pages} {579} (\bibinfo {year} {2024})},\ \Eprint {https://arxiv.org/abs/2401.17022} {arXiv:2401.17022} \BibitemShut {NoStop}%
\bibitem [{\citenamefont {Mele}(2024)}]{mele2024introduction}%
  \BibitemOpen
  \bibfield  {author} {\bibinfo {author} {\bibfnamefont {A.~A.}\ \bibnamefont {Mele}},\ }\bibfield  {title} {\bibinfo {title} {Introduction to haar measure tools in quantum information: A beginner's tutorial},\ }\href@noop {} {\bibfield  {journal} {\bibinfo  {journal} {Quantum}\ }\textbf {\bibinfo {volume} {8}},\ \bibinfo {pages} {1340} (\bibinfo {year} {2024})}\BibitemShut {NoStop}%
\bibitem [{\citenamefont {Hubbard}(1963)}]{hubbardElectronCorrelationsNarrow1963}%
  \BibitemOpen
  \bibfield  {author} {\bibinfo {author} {\bibfnamefont {J.}~\bibnamefont {Hubbard}},\ }\bibfield  {title} {\bibinfo {title} {Electron correlations in narrow energy bands},\ }\bibfield  {journal} {\bibinfo  {journal} {Proc. R. Soc. Lond. Ser. Math. Phys. Sci.}\ }\href {https://doi.org/10.1098/rspa.1963.0204} {10.1098/rspa.1963.0204} (\bibinfo {year} {1963})\BibitemShut {NoStop}%
\bibitem [{\citenamefont {Hensgens}\ \emph {et~al.}(2017)\citenamefont {Hensgens}, \citenamefont {Fujita}, \citenamefont {Janssen}, \citenamefont {Li}, \citenamefont {Van~Diepen}, \citenamefont {Reichl}, \citenamefont {Wegscheider}, \citenamefont {Das~Sarma},\ and\ \citenamefont {Vandersypen}}]{hensgensQuantumSimulationFermi2017}%
  \BibitemOpen
  \bibfield  {author} {\bibinfo {author} {\bibfnamefont {T.}~\bibnamefont {Hensgens}}, \bibinfo {author} {\bibfnamefont {T.}~\bibnamefont {Fujita}}, \bibinfo {author} {\bibfnamefont {L.}~\bibnamefont {Janssen}}, \bibinfo {author} {\bibfnamefont {X.}~\bibnamefont {Li}}, \bibinfo {author} {\bibfnamefont {C.~J.}\ \bibnamefont {Van~Diepen}}, \bibinfo {author} {\bibfnamefont {C.}~\bibnamefont {Reichl}}, \bibinfo {author} {\bibfnamefont {W.}~\bibnamefont {Wegscheider}}, \bibinfo {author} {\bibfnamefont {S.}~\bibnamefont {Das~Sarma}},\ and\ \bibinfo {author} {\bibfnamefont {L.~M.~K.}\ \bibnamefont {Vandersypen}},\ }\bibfield  {title} {\bibinfo {title} {Quantum simulation of a {{Fermi}}--{{Hubbard}} model using a semiconductor quantum dot array},\ }\href {https://doi.org/10.1038/nature23022} {\bibfield  {journal} {\bibinfo  {journal} {Nature}\ }\textbf {\bibinfo {volume} {548}},\ \bibinfo {pages} {70} (\bibinfo {year} {2017})}\BibitemShut {NoStop}%
\bibitem [{\citenamefont {Shao}\ \emph {et~al.}(2024)\citenamefont {Shao}, \citenamefont {Wang}, \citenamefont {Zhu}, \citenamefont {Zhu}, \citenamefont {Sun}, \citenamefont {Chen}, \citenamefont {Zhang}, \citenamefont {Fan}, \citenamefont {Deng}, \citenamefont {Yao}, \citenamefont {Chen},\ and\ \citenamefont {Pan}}]{shaoAntiferromagneticPhaseTransition2024}%
  \BibitemOpen
  \bibfield  {author} {\bibinfo {author} {\bibfnamefont {H.-J.}\ \bibnamefont {Shao}}, \bibinfo {author} {\bibfnamefont {Y.-X.}\ \bibnamefont {Wang}}, \bibinfo {author} {\bibfnamefont {D.-Z.}\ \bibnamefont {Zhu}}, \bibinfo {author} {\bibfnamefont {Y.-S.}\ \bibnamefont {Zhu}}, \bibinfo {author} {\bibfnamefont {H.-N.}\ \bibnamefont {Sun}}, \bibinfo {author} {\bibfnamefont {S.-Y.}\ \bibnamefont {Chen}}, \bibinfo {author} {\bibfnamefont {C.}~\bibnamefont {Zhang}}, \bibinfo {author} {\bibfnamefont {Z.-J.}\ \bibnamefont {Fan}}, \bibinfo {author} {\bibfnamefont {Y.}~\bibnamefont {Deng}}, \bibinfo {author} {\bibfnamefont {X.-C.}\ \bibnamefont {Yao}}, \bibinfo {author} {\bibfnamefont {Y.-A.}\ \bibnamefont {Chen}},\ and\ \bibinfo {author} {\bibfnamefont {J.-W.}\ \bibnamefont {Pan}},\ }\bibfield  {title} {\bibinfo {title} {Antiferromagnetic phase transition in a {{3D}} fermionic {{Hubbard}} model},\ }\href {https://doi.org/10.1038/s41586-024-07689-2} {\bibfield  {journal} {\bibinfo  {journal} {Nature}\ ,\ \bibinfo
  {pages} {1}} (\bibinfo {year} {2024})},\ \Eprint {https://arxiv.org/abs/2402.14605} {arXiv:2402.14605} \BibitemShut {NoStop}%
\bibitem [{\citenamefont {McClean}\ \emph {et~al.}(2020)\citenamefont {McClean}, \citenamefont {Rubin}, \citenamefont {Sung}, \citenamefont {Kivlichan}, \citenamefont {{Bonet-Monroig}}, \citenamefont {Cao}, \citenamefont {Dai}, \citenamefont {Fried}, \citenamefont {Gidney}, \citenamefont {Gimby}, \citenamefont {Gokhale}, \citenamefont {H{\"a}ner}, \citenamefont {Hardikar}, \citenamefont {Havl{\'i}{\v c}ek}, \citenamefont {Higgott}, \citenamefont {Huang}, \citenamefont {Izaac}, \citenamefont {Jiang}, \citenamefont {Liu}, \citenamefont {McArdle}, \citenamefont {Neeley}, \citenamefont {O'Brien}, \citenamefont {O'Gorman}, \citenamefont {Ozfidan}, \citenamefont {Radin}, \citenamefont {Romero}, \citenamefont {Sawaya}, \citenamefont {Senjean}, \citenamefont {Setia}, \citenamefont {Sim}, \citenamefont {Steiger}, \citenamefont {Steudtner}, \citenamefont {Sun}, \citenamefont {Sun}, \citenamefont {Wang}, \citenamefont {Zhang},\ and\ \citenamefont {Babbush}}]{mccleanOpenFermionElectronicStructure2020}%
  \BibitemOpen
  \bibfield  {author} {\bibinfo {author} {\bibfnamefont {J.~R.}\ \bibnamefont {McClean}}, \bibinfo {author} {\bibfnamefont {N.~C.}\ \bibnamefont {Rubin}}, \bibinfo {author} {\bibfnamefont {K.~J.}\ \bibnamefont {Sung}}, \bibinfo {author} {\bibfnamefont {I.~D.}\ \bibnamefont {Kivlichan}}, \bibinfo {author} {\bibfnamefont {X.}~\bibnamefont {{Bonet-Monroig}}}, \bibinfo {author} {\bibfnamefont {Y.}~\bibnamefont {Cao}}, \bibinfo {author} {\bibfnamefont {C.}~\bibnamefont {Dai}}, \bibinfo {author} {\bibfnamefont {E.~S.}\ \bibnamefont {Fried}}, \bibinfo {author} {\bibfnamefont {C.}~\bibnamefont {Gidney}}, \bibinfo {author} {\bibfnamefont {B.}~\bibnamefont {Gimby}}, \bibinfo {author} {\bibfnamefont {P.}~\bibnamefont {Gokhale}}, \bibinfo {author} {\bibfnamefont {T.}~\bibnamefont {H{\"a}ner}}, \bibinfo {author} {\bibfnamefont {T.}~\bibnamefont {Hardikar}}, \bibinfo {author} {\bibfnamefont {V.}~\bibnamefont {Havl{\'i}{\v c}ek}}, \bibinfo {author} {\bibfnamefont {O.}~\bibnamefont {Higgott}}, \bibinfo {author}
  {\bibfnamefont {C.}~\bibnamefont {Huang}}, \bibinfo {author} {\bibfnamefont {J.}~\bibnamefont {Izaac}}, \bibinfo {author} {\bibfnamefont {Z.}~\bibnamefont {Jiang}}, \bibinfo {author} {\bibfnamefont {X.}~\bibnamefont {Liu}}, \bibinfo {author} {\bibfnamefont {S.}~\bibnamefont {McArdle}}, \bibinfo {author} {\bibfnamefont {M.}~\bibnamefont {Neeley}}, \bibinfo {author} {\bibfnamefont {T.}~\bibnamefont {O'Brien}}, \bibinfo {author} {\bibfnamefont {B.}~\bibnamefont {O'Gorman}}, \bibinfo {author} {\bibfnamefont {I.}~\bibnamefont {Ozfidan}}, \bibinfo {author} {\bibfnamefont {M.~D.}\ \bibnamefont {Radin}}, \bibinfo {author} {\bibfnamefont {J.}~\bibnamefont {Romero}}, \bibinfo {author} {\bibfnamefont {N.~P.~D.}\ \bibnamefont {Sawaya}}, \bibinfo {author} {\bibfnamefont {B.}~\bibnamefont {Senjean}}, \bibinfo {author} {\bibfnamefont {K.}~\bibnamefont {Setia}}, \bibinfo {author} {\bibfnamefont {S.}~\bibnamefont {Sim}}, \bibinfo {author} {\bibfnamefont {D.~S.}\ \bibnamefont {Steiger}}, \bibinfo {author} {\bibfnamefont
  {M.}~\bibnamefont {Steudtner}}, \bibinfo {author} {\bibfnamefont {Q.}~\bibnamefont {Sun}}, \bibinfo {author} {\bibfnamefont {W.}~\bibnamefont {Sun}}, \bibinfo {author} {\bibfnamefont {D.}~\bibnamefont {Wang}}, \bibinfo {author} {\bibfnamefont {F.}~\bibnamefont {Zhang}},\ and\ \bibinfo {author} {\bibfnamefont {R.}~\bibnamefont {Babbush}},\ }\bibfield  {title} {\bibinfo {title} {{{OpenFermion}}: The electronic structure package for quantum computers},\ }\href {https://doi.org/10.1088/2058-9565/ab8ebc} {\bibfield  {journal} {\bibinfo  {journal} {Quantum Sci. Technol.}\ }\textbf {\bibinfo {volume} {5}},\ \bibinfo {pages} {034014} (\bibinfo {year} {2020})},\ \Eprint {https://arxiv.org/abs/1710.07629} {arXiv:1710.07629} \BibitemShut {NoStop}%
\bibitem [{\citenamefont {Suzuki}(1993)}]{suzuki1993general}%
  \BibitemOpen
  \bibfield  {author} {\bibinfo {author} {\bibfnamefont {M.}~\bibnamefont {Suzuki}},\ }\bibfield  {title} {\bibinfo {title} {General decomposition theory of ordered exponentials},\ }\href@noop {} {\bibfield  {journal} {\bibinfo  {journal} {Proceedings of the Japan Academy, Series B}\ }\textbf {\bibinfo {volume} {69}},\ \bibinfo {pages} {161} (\bibinfo {year} {1993})}\BibitemShut {NoStop}%
\bibitem [{\citenamefont {Hastings}(2010)}]{hastings2010locality}%
  \BibitemOpen
  \bibfield  {author} {\bibinfo {author} {\bibfnamefont {M.~B.}\ \bibnamefont {Hastings}},\ }\bibfield  {title} {\bibinfo {title} {Locality in quantum systems},\ }\href@noop {} {\bibfield  {journal} {\bibinfo  {journal} {Quantum Theory from Small to Large Scales}\ }\textbf {\bibinfo {volume} {95}},\ \bibinfo {pages} {171} (\bibinfo {year} {2010})}\BibitemShut {NoStop}%
\end{thebibliography}%

\setcounter{theorem}{0}
\setcounter{lemma}{0}
\setcounter{proposition}{0}
\setcounter{definition}{0}
\setcounter{corollary}{0}
\setcounter{algocf}{0}
\clearpage
\onecolumngrid
\appendix

\section{Preliminaries}
\subsection{Properties of Hamiltonians}\label{sec:Hamil}
In the case of dynamic evolutions of closed $n$-qubit systems, Hamiltonian operators guided the evolution processes according to Schr\"{o}dinger's equation.
Due to the Hermiticity of Hamiltonian operators, we can always decompose them on $n$-fold Pauli operators with real coefficients.
Namely, for a Hamiltonian $H$
\begin{gather}
    H=\sum_{\alpha\in {\sf P}^n}s_\alpha P_\alpha,\ \forall\alpha \in {\sf P}^n\ s_\alpha\in\mathbb{R},
\end{gather}
where ${\sf P}^n$ represents the quotient Pauli group on phases.
We further introduce some constraints for the Hamiltonian.
The first is the conventional locality of Hamiltonian.
\begin{definition}\label{def:constriant-local}
    Consider an $n$-qubit Hamiltonian operator $H=\sum_{\alpha\in{\sf P}^n}s_\alpha P_\alpha$.
    We call it an $\ell$-local Hamiltonian when $s_\alpha\neq0$ only if the weight of the Pauli operator $P_\alpha$ is no larger than $\ell$.
\end{definition}
In order to use the locality constraint of the Hamiltonian to limit the effective interaction range, we also need the following bound of the interaction strength for each qubit.
\begin{definition}\label{def:constriant-qubitwise}
    We define the interaction per qubit of this Hamiltonian $H$ by $b$: 
    \begin{gather}
        b\coloneqq\max_{j\in[n]}\sum_{\alpha: j\in S(P_\alpha)}\abs{s_\alpha}.
    \end{gather}
\end{definition}
The norm we use to quantify scales of the Hamiltonian is the 1-norm based on the Pauli decomposition. 
\begin{definition}
    For a Hermitian operator $H=\sum_{\alpha\in{\sf P}^n}s_\alpha P_\alpha$, we define the 1-norm to be
    \begin{gather}
        \|H\|_1\coloneqq\sum_{\alpha\in{\sf P}^n}|s_\alpha|.
    \end{gather}
\end{definition}
The problem of digital quantum simulation is to construct a quantum circuit from digital quantum gates to approximate an evolution $\mathrm{e}^{-\ii Ht}$ with specified $H$ and time $t$.
Our analysis, particularly in the context of short-time analysis, centers on the evolution of specific families of observables.
Therefore, we can extend to approximate the evolved observable $O$ in Heisenberg's picture, $\mathrm{e}^{\ii Ht}O\mathrm{e}^{-\ii Ht}$.

\subsection{Product Formula}
The product formula is a commonly used technique for approximating matrix exponentials, particularly in the context of quantum simulation where it is favored for its simplicity.
This method involves decomposing the target matrix into sub-matrices and then implementing exponentials of these sub-matrices in sequence. 
In the case of Heisenberg's picture, a general product formula for the decomposition $H=\sum_{\gamma=1}^\Gamma H_\gamma$ can be outlined as
\begin{gather}\label{eq:ap-product}
    \mathscr{S}(t)\coloneqq\prod_{\upsilon=1}^{\Upsilon}\prod_{\gamma=1}^\Gamma\mathrm{e}^{\ii ta_{(\upsilon,\gamma)}H_{\pi_{\upsilon}(\gamma)}},
\end{gather}
where $\Upsilon$ is the number of stages involved, $a_{(\upsilon,\gamma)}$ denotes the coefficient of each exponential, and $\{\pi_\upsilon\}$ represents a family of permutations over different sub-Hamiltonians.
To be consistent with the convention, we use the notation $\prod$ to denote the product of several operators from right to left, \emph{i.e.},
\begin{gather*}
    \prod_{\gamma=1}^\Gamma U_\gamma=U_\Gamma\cdots U_2U_1.
\end{gather*}
Based on the general structure of the product formula, we would like to introduce a specific method that is widely used in Hamiltonian simulation, which is known as \emph{Suzuki-Trotter} formula~\cite{suzuki1991general}.
The Suzuki-Trotter formula is a family of recursively defined formulas, and we start the elaboration from the first-order case.
The first-order Suzuki-Trotter simply implements the exponentials of sub-Hamiltonians sequentially in a certain order $\pi$ with only one stage as
\begin{gather}
    \mathscr{S}_1(t)\coloneqq\prod_{\gamma=1}^\Gamma\mathrm{e}^{\ii H_{\pi(\gamma)} t}.
\end{gather}
The second-order Suzuki-Trotter formula is in a back-and-forth manner
\begin{gather}
    \mathscr{S}_2(t)\coloneqq\prod_{\gamma=\Gamma}^1\mathrm{e}^{\ii H_{\pi(\gamma)} t/2}\cdot\prod_{\gamma=1}^\Gamma\mathrm{e}^{\ii H_{\pi(\gamma)} t/2}.
\end{gather}
For an even integer $p$, the further definition of $p$th-order formulas comes from a combination of lower-order formulas,
\begin{gather}\label{eq:suzuki}
    \mathscr{S}_{p}(t)\coloneqq\mathscr{S}_{p-2}^2(u_pt)\mathscr{S}_{p-2}((1-4u_p)t)\mathscr{S}_{p-2}^2(u_pt),
\end{gather}
where $u_p=1/(4-4^{1/(p-1)})$.
It is clear to see that a $p$th-order formula requires $2\times 5^{p/2-1}$ stages and the coefficients are recursively calculated.
Based on Suzuki's analysis~\cite{suzuki1993general}, for short time $t\ll1$, the $p$th-order formula approximates the ideal Hamiltonian exponential $\mathrm{e}^{-\ii Ht}$ with error proportional to $t^{p+1}$.

As per the iterative definition of the Suzuki-Trotter formula, the coefficients $a_{(\upsilon,\gamma)}$ and stage number $\Upsilon$ in Eq.~\eqref{eq:ap-product} are entirely determined by $p$. 
The remaining parameters to be determined are the decomposition and the symmetric permutation, collectively referred to as the configuration. 
In subsequent analysis, our primary focus lies on identifying advantageous configurations to achieve improved scalings for short-time simulation.

\subsection{Largest Possible Support Propagation}
\begin{definition}\label{def:oplus}
Given a set $S$ of qubits and a unitary $U$ that acts non-trivially on the support set $S(U)$, we define the operation $\uplus$ to be
\begin{gather}
    U\uplus S\coloneqq \begin{cases}
    S\cup S(U)& S\cap S(U)\neq\emptyset\\
    S& \text{otherwise}
    \end{cases}.
\end{gather}
\end{definition}
\begin{remark}
    \rm The operation $\uplus$ estimates the upper bound of the possible support of $U OU^\dagger$, which depicts the expansion of the support $S$ of $O$ under unitary operations.
    Therefore, the worst-case is $S\cup S(U)$ where support is propagated to the largest set under the unitary operation.
    It is worth noting that with more details of the unitary, the upper bound might shrink.
    Namely, we have 
    \begin{gather*}
        U_2\uplus(U_1\uplus S)\neq (U_2\cdot U_1)\uplus S.
    \end{gather*}
    This inequality also hints that implementing unitaries in different sequences causes different support analyses, which we will make use of later.
\end{remark}
\begin{observation}\label{ob:monoton}
    Consider an arbitrary unitary $U$ and two sets $S_1\subseteq S_2$. The worst-case supports satisfy
    \begin{gather}
        U\uplus S_1\subseteq U\uplus S_2.
    \end{gather}
    Similarly, consider two unitary $U_1,U_2$ that $S(U_1)\subseteq S(U_2)$ and an arbitrary set $S$. The worst-case supports satisfy
    \begin{gather}
        U_1\uplus S\subseteq U_2\uplus S.
    \end{gather}
\end{observation}

The worst-case support propagation estimated by the operation $\uplus$ identifies those relevant unitaries out from irrelevant ones.
For example, suppose we want to sequentially implement unitary operations as $U_{\Gamma}\cdots U_2U_1OU^\dagger_1U^\dagger_2\cdots U^\dagger_{\Gamma}$.
If $S(U_\Gamma)\cap \biguplus_{\gamma=1}^{\Gamma-1}U_\gamma\uplus S(O)=\emptyset$, the unitary $U_\Gamma$ causes no effects and can be ignored during implementation.
Consequently, we can sift those effective unitary operations and vastly reduce the number of gates needed to execute a sequence of unitary evolution.

In this paper, we mainly consider worst-case supports and the possible reduction of product formula algorithms, which recruit sequential unitary operations.
Moreover, while different configurations cause little to no distinctions in error scaling, they result in varying support propagations.
In this sense, a careful design of the product formula can suppress the propagation to the slowest case, which makes the most of the reduction from irrelevant unitaries. 
In the following two sections, we introduce the design method of the optimally slow configurations in short-time simulation of local observables.

\section{Short-Time Simulation for Local Observables}\label{sec:append-short}
According to Heisenberg's picture, the measurement of evolved state $\Tr(\mathrm{e}^{-\ii Ht}\rho\mathrm{e}^{\ii Ht}O)$ equals measuring $\rho$ on the evolved operator $\mathrm{e}^{\ii Ht}O\mathrm{e}^{-\ii Ht}$.
In the remainder of this paper, we mainly illustrate our analyses from Heisenberg's picture.

To explore the advantage of the short-time simulation, we need to suppress the expansion of the observable in our simulation circuit.
To this end, we introduce the following configuration based on the interaction structure of Hamiltonian regarding the specific $O$.
This can help to organize a proper product formula for the short-time simulation.
\begin{definition}\label{def:step_edge}
Suppose we have an $n$-qubit Hamiltonian operator $H=\sum_{\alpha\in{\sf P}^n}s_\alpha P_\alpha$. 
Given a set $S$ of qubits, we can define the intrinsic Hamiltonian on $S$ to be $H_0^S\coloneqq \sum_{\alpha:S(P_\alpha)\subset S}s_\alpha P_\alpha$.
For simplicity, we use $E_0^S$ to denote $S$.
For a positive integer $k$, we can further define sub-Hamiltonians $H_k^S$ and edge sets $E_k^S$, respectively,
\begin{align}\label{eq:edgeset}
            H_k^S&\coloneqq \sum_{\substack{\alpha:P_\alpha\notin H_{k-1}^S\\S(P_\alpha)\cap E_{k-1}^S\neq\emptyset}}s_\alpha P_\alpha,\ \ \ 
    E_k^S\coloneqq S(H_k^S)-E_{k-1}^S,
\end{align}
where $P_\alpha\notin H_{k-1}^S$ means $P_\alpha$ is not a summand in the linear combination of $H_{k-1}^S$.
This induction stops when we cannot find any more terms of the Hamiltonian, and we assume a total of $\Gamma_0+1$ edge sets.
\end{definition}
We will show shortly that edge sets reveal the fundamental structures of interactions since they imply the tight lower bounds of the expanding support of $O(t)$.
Note that we can always decompose the system and Hamiltonian in this form regardless of the specific settings.

To complete the illustration of the configuration for product formulae, we adopt the even-odd permutation to be
\begin{align}\label{eq:evenodd}
    \pi^{\text{eo}}_\upsilon(0,1,2,3,4,5,\cdots)=
    \begin{cases}
        0,2,4,\cdots, 1,3,5,\cdots & \text{$\upsilon$ is odd}\\
         1,3,5,\cdots,0,2,4,\cdots & \text{$\upsilon$ is even}
    \end{cases}.
\end{align}
This even-odd sequence is symmetric, which implies the interactive decomposition and even-odd permutation is a valid configuration for the Suzuki-Trotter formula.
The following lemma indicates that the corresponding formula can suppress the propagation of support optimally.
\begin{lemma}\label{lm:ap-optimal}
Consider an operator $O$ with support $S$, an $n$-qubit Hamiltonian operator $H$, and a $\Upsilon$-stage product formula as in Eq.~\eqref{eq:ap-product} with an arbitrary decomposition and permutation. 
The largest possible support expands as 
\begin{gather}\label{eq:ap-worst-case-prod}
\left(\biguplus_{\upsilon=1}^{\Upsilon}\biguplus_{\gamma=1}^{\Gamma}\mathrm{e}^{\ii ta_{(\upsilon,\gamma)}H_{\pi_\upsilon(\gamma)}}\uplus S\right)\supseteq \bigcup_{k=0}^\Upsilon E_k^S,
\end{gather}
with equality holds with the decomposition in Eq.~\eqref{eq:edgeset} and permutation in Eq.~\eqref{eq:evenodd}.
\end{lemma}
\begin{proof}
We would like to first verify Eq.~\eqref{eq:ap-worst-case-prod} and then check the mentioned configuration. Consider the case that $H$ has interactions with qubits in $S$.
Otherwise, the statement is trivially true.
We start the proof from $\Upsilon=1$.
    Consider an arbitrary qubit $Q_j\in E^S_1$.
    According to the definition of $E_1^S$, there exists a Pauli constituent $P_\beta$ in $H$ such that $Q_j\in S(P_\beta)$ and $S(P_\beta)\cap S\neq\emptyset$.
    Without loss of generality, we assume that $H_1$ contains a non-zero decomposition coordinate on $P_\beta$ and $\pi_1(k_0)=1$.
    Since $\biguplus_{\gamma=1}^{k_0-1}\mathrm{e}^{\ii ta_{(1,\gamma)}H_{\pi_1(\gamma)}}\uplus S\supseteq S$, it is easy to check by Observation~\ref{ob:monoton} that
    \begin{gather*}
        \biguplus_{\gamma=1}^{k_0}\mathrm{e}^{ita_{(1,\gamma)}H_{\pi_1(\gamma)}}\uplus S\supseteq S\cup\{Q_j\}.
    \end{gather*}
    By enumerating all qubits in $E_1^S$, we proved the following for $\Upsilon=1$ $$\biguplus_{\gamma=1}^{\Gamma}\mathrm{e}^{\ii ta_{(1,\gamma)}H_{\pi_1(\gamma)}}\uplus S\supseteq S\cup E_1^S.$$

    For $\Upsilon>1$, we assume the statement to be true up to $\Upsilon-1$.
    The $\Upsilon$ case can be derived as
    \begin{align*}
        \left(\biguplus_{\upsilon=1}^{\Upsilon}\biguplus_{\gamma=1}^{\Gamma}\mathrm{e}^{\ii ta_{(\upsilon,\gamma)}H_{\pi_\upsilon(\gamma)}}\uplus S\right)=&\biguplus_{\gamma=1}^{\Gamma}\mathrm{e}^{\ii ta_{(\upsilon,\gamma)}H_{\pi_\upsilon(\gamma)}}\uplus\biguplus_{\upsilon=1}^{\Upsilon-1}\biguplus_{\gamma=1}^{\Gamma}\mathrm{e}^{\ii ta_{(\upsilon,\gamma)}H_{\pi_\upsilon(\gamma)}}\uplus S\\
        \supseteq& \biguplus_{\gamma=1}^{\Gamma}\mathrm{e}^{\ii ta_{(\upsilon,\gamma)}H_{\pi_\upsilon(\gamma)}}\uplus \left(\bigcup_{k=1}^{\Upsilon-1} E_k^S\cup S\right).
    \end{align*}
Based on Definition~\ref{def:step_edge} and a similar argument as for the case $\Upsilon=1$, we can complete this inductive proof of Eq.~\eqref{eq:ap-worst-case-prod}.

Let us track the propagation of support $S$ under the mentioned configuration by $\uplus$.
Note from Eq.~\eqref{eq:edgeset} that $S(H_0^S)=E_0^S$ and that $S(H_k^S)=E_{k-1}^S\cup E_k^S$ for all $k\in\mathbb{Z}^+$.
Therefore, every sub-Hamiltonian only interacts with its nearest neighbors.
According to the even-odd permutation, we have
\begin{align}\label{eq:even-odd ws}
    \biguplus_{\upsilon=1}^{\Upsilon}\biguplus_{\gamma=1}^{\Gamma_0}\mathrm{e}^{\ii ta_{(\upsilon,\gamma)}H^S_{\pi^{\text{eo}}_\upsilon(\gamma)}}\uplus S=&\biguplus_{\upsilon=2}^{\Upsilon}\biguplus_{\gamma=0}^{\Gamma_0}\mathrm{e}^{\ii ta_{(\upsilon,\gamma)}H^S_{\pi^{\text{eo}}_\upsilon(\gamma)}}\uplus(S(H_1^S)\cup S(H_0^S))\notag\\
    =&\biguplus_{\upsilon=3}^{\Upsilon}\biguplus_{\gamma=0}^{\Gamma_0}\mathrm{e}^{\ii ta_{(\upsilon,\gamma)}H^S_{\pi^{\text{eo}}_\upsilon(\gamma)}}\uplus( S(H^S_2)\cup S(H^S_1)\cup S(H_0^S))\notag\\
    =&\bigcup_{k=1}^\Upsilon E_k^S\cup E_0^S.
\end{align}
The first equation comes from the fact that only $H_0$ and $H_1$ terms interact with $O$.
The second is due to the intersection between $S(H_1)$ and $S(H_2)$.
Repeat this deduction, and we can discover that the even-odd permutation and interactive decomposition turn out to achieve the slowest expansion rate in Eq.~\eqref{eq:ap-worst-case-prod}.
\end{proof}

The efficiency of the product formula for simulation can be significantly enhanced by excluding unitaries outside the support, termed irrelevant unitaries, in accordance with the growth of the support.
This refinement does not impact the simulation outcomes because the irrelevant unitaries in the product formula must commute with the evolved $O(t)$, and their unitary conjugations will nullify them.
Since the proof of Lemma~\ref{lm:optimal} is independent of the stage number, the proposed configuration consistently maintains the optimally slow expansion of the support during the simulation.
Particularly, no further reduction can be made while ensuring the simulation results remain intact.

In Algorithm~\ref{alg:ap-rpf}, we elucidate the design of the reduced product formula.
The coefficients $a_{(\upsilon,\gamma)}$ and the number of stages $\Upsilon$ are determined by its order based on the standard Suzuki-Trotter in Eq.~\eqref{eq:suzuki}.
This algorithm also accounts for the effects of multiple Trotter steps.
To reconcile the even-odd permutation between successive steps, the algorithm employs $b+\upsilon$ to indicate the effective parity of the permutation.

\begin{algorithm}[t]
\caption{Reduced Product Formula}\label{alg:ap-rpf}
    \SetKwInOut{Input}{input}\SetKwInOut{Output}{output}
    \Input{Observable $O$ with support $S$; Hamiltonian $H$; Evolving time $t$; Step number $r$; Order $p$}
    \Output{Unitary $U$}
    Decompose $H$ into $H_0^S,\cdots,H_{\Gamma_0}^S$ according to Eq.~\eqref{eq:edgeset}\;
    Adopt Suzuki-Trotter coefficients $a_{(\upsilon,\gamma)}$ from~\cite{suzuki1991general}\;
    $\Upsilon\leftarrow 2\cdot5^{p/2-1}$,
    $\tau\leftarrow t/r$,
    $U\leftarrow I$\;
    \For{$j=1,\cdots,r$}{
        \For{$\upsilon=1,\cdots,\Upsilon$}{
            $U\leftarrow\prod_{\gamma=0}^{ \upsilon+(j-1)\Upsilon}\mathrm{e}^{\ii a_{(\upsilon,\gamma)}\tau H_{\pi^{\text{eo}}_{\upsilon}(\gamma)}^S}\cdot U$\;
            
        }
    }
\end{algorithm}
We provide a more concrete depiction of this concept in Figure~2 of the main text, illustrating the implementation of the second-order reduced product formula. 
As the operator evolves, its support expands gradually. 
The highlighted unitaries precisely correspond to those indicated in Algorithm~\ref{alg:ap-rpf}.

Implicitly, this algorithm necessitates that $\Gamma_0$ is sufficiently large to obviate the need for considering boundary cases. 
We establish the error bound based on this ideal scenario subsequently. 
However, even in cases where this condition is not fulfilled, the algorithm can halt the expanding range of unitaries once it has enumerated all terms in the Hamiltonian.
In such circumstances, the error bound reduces to the standard worst-case bound outlined in~\cite{childs2021theory}.

\begin{theorem}[Single-Observable Error]\label{thm:ap-single}
Consider a local observable $O$ with support $S$.
Suppose the $n$-qubit $H$ is $\ell$-local with a constant $\ell$ and 
has bounded interaction per qubit. 
With the light-cone width $\textstyle w_r\coloneqq \sum\nolimits_{k=0}^{r\Upsilon+1}\|H_k^S\|_1$, the simulation error of $O$ by an $r$-step $p$th-order $U$ from Alg.~\ref{alg:ap-rpf} is bounded by $$\|\mathrm{e}^{\ii Ht}O\mathrm{e}^{-\ii Ht}-U OU^\dagger\|=\order{\frac{\|O\|w_rt^{p+1}}{r^{p}}}.$$ 
\end{theorem}

\begin{remark}
    \rm
    It is evident that the estimated error is lower than the operator norm bound in~\cite{childs2021theory} given $\sum_{k=0}^rh_k<\|H\|_1$.
    This discrepancy arises from the fact that the ``width" of the light cone originating from $S$ cannot fully encompass the entire system.
    The corresponding choices of $r$ for fixed error should be determined based on this bound and the geometric properties of the system, \emph{i.e}, the scaling of $\{h_k\}$.

\end{remark}

\begin{proof}
We first consider a family of $r$ different Hamiltonian divisions, where $r$ is the number of time steps we adopted in Algorithm~\ref{alg:ap-rpf}.
\begin{align}\label{eq:virtual decomp}
\begin{split}
    \text{Step 1:}\ H=&H_0^S+H_1^S+H^S_2+\cdots+H^S_\Upsilon+H_{\text{other},1}=\sum_{\gamma=0}^{\Upsilon+1}H_\gamma^{(1)},\\
    \text{Step 2:}\ H=&H_0^S+H^S_1+H^S_2+\cdots+H^S_{2\Upsilon}+H_{\text{other},2}=\sum_{\gamma=0}^{2\Upsilon+1}H_\gamma^{(2)},\\
    &\vdots\\
    \text{Step $r$:}\ H=&H_0^S+H^S_1+H^S_2+\cdots+H^S_{r\Upsilon}+H_{\text{other},r}=\sum_{\gamma=0}^{r\Upsilon+1}H_\gamma^{(r)}.
\end{split} 
\end{align}
Note that all decompositions in Eq.~\eqref{eq:virtual decomp} are equal to that in~\eqref{eq:edgeset} except for those tail terms.
We will simplify the superscript $S$ of sub-Hamiltonians in the remainder of this proof.
Therefore, $H_\gamma=H_\gamma^{(j)}$ for all $\gamma\in[1,j\Upsilon]$ and $H_{\text{other},j}=H_{j\Upsilon+1}^{(j)}$.
Based on these divisions, we introduce a virtual product formula consisting of $r$ successive Trotter formula steps.
$$\mathscr{S}_{\text{V}}(t) O\mathscr{S}_{\text{V}}^\dagger(t)\coloneqq\mathscr{S}_r(\tau)\cdots\mathscr{S}_1(\tau)O\mathscr{S}_1(\tau)^\dagger \cdots\mathscr{S}_r(\tau)^\dagger,$$
where $\tau=t/r$ is the length of each time step.
For each time step $\tau$.  the formula adaptively adopts different $p$th-order Suzuki Trotter formulae,
\begin{align}
    \forall j\in[r],\ \mathscr{S}_j(\tau)\coloneqq\prod_{\upsilon=1}^\Upsilon\prod_{\gamma=0}^{j\Upsilon+1}\mathrm{e}^{\ii\tau a_{(\upsilon,\gamma)}H^{(j)}_{\pi^{\text{eo}}_{(j-1)\Upsilon+\upsilon}(\gamma)}}.
\end{align} 
As hinted by its name, the virtual product formula is never truly implemented.
We introduce this product formula as an intermediate to employ previous nested-commutator error bound for Suzuki-Trotter formulas.
Based on the analysis in~\cite{childs2021theory}, the error of approximating $\mathrm{e}^{\ii H\tau}$ by each virtual formula is
\begin{align}\label{eq:rpf_trotter}
    \epsilon_j=\|\mathrm{e}^{\ii H\tau}-\mathscr{S}_j(\tau)\|=\order{\sum_{\gamma_1,\cdots,\gamma_{p+1}=0}^{j\Upsilon+1}\|[H^{(j)}_{\gamma_{p+1}},\cdots[H^{(j)}_{\gamma_{3}},[H^{(j)}_{\gamma_2},H^{(j)}_{\gamma_1}]]\cdots]\|\tau^{p+1}}.
\end{align}

Recall that the tail term consists of all other sub-Hamiltonians $H^{(j)}_{j\Upsilon+1}\coloneqq\sum_{k=j\Upsilon+1}^{\Gamma_0}H_k$, thus, the summation seems inevitably incur some dependence with the system size $n$ in the above bound.
An idea to remove this dependence is to limit the effective part of the tail Hamiltonian $H^{(j)}_{j\Upsilon+1}$ located at the innermost layer of the commutator.
When $H^{(j)}_{j\Upsilon+1}$ is involved, the other must be $H^{(j)}_{j\Upsilon}$ to make sure the commutator nontrivial.
According to Definition~\ref{def:step_edge}, we can substitute both tail Hamiltonians in the innermost layer for $H_{j\Upsilon+1}$ and keep the commutator the same.
\begin{gather}
    \epsilon_j=\order{\sum_{\gamma_1,\cdots,\gamma_{p+1}=0}^{j\Upsilon+1}\|[H^{(j)}_{\gamma_{p+1}},\cdots[H^{(j)}_{\gamma_{3}},[H_{\gamma_2},H_{\gamma_1}]]\cdots]\|\tau^{p+1}}.
\end{gather}

We shall consider the Pauli decomposition of every sub-Hamiltonian and the limitations thereof to quantify the norm.
The analysis starts from the innermost layer.
\begin{align}
    \sum_{\gamma_2=0}^{j\Upsilon+1}\|[H_{\gamma_2},H_{\gamma_1}]\|\leq&\sum_{\gamma_2=0}^{j\Upsilon+1}\sum_{\beta:P_\beta\in H_{\gamma_2}}\sum_{\alpha:P_\alpha\in H_{\gamma_1}}\| [s_\beta P_\beta,s_\alpha P_\alpha]\|\notag
    \\
    =&\sum_{\alpha:P_\alpha\in H_{\gamma_1}}\sum_{\gamma_2=0}^{j\Upsilon+1}\sum_{\beta:P_\beta\in H_{\gamma_2}}\mathbbm{1}[S(P_\beta)\cap S(P_\alpha)]\| [s_\beta P_\beta,s_\alpha P_\alpha]\|\\
    \leq&\sum_{\alpha:P_\alpha\in H_{\gamma_1}}\sum_{\gamma_2=0}^{j\Upsilon+1}\sum_{\beta:P_\beta\in H_{\gamma_2}}\mathbbm{1}[S(P_\beta)\cap S(P_\alpha)]2\|s_\beta P_\beta\|\|s_\alpha P_\alpha\|
    \leq\sum_{\alpha:P_\alpha\in H_{\gamma_1}}2\ell b\|s_\alpha P_\alpha\|,
\end{align}
where the equation comes from the fact that only terms that have overlaps can have a nonzero commutator, and $\ell$ denotes the possible position of intersection in $P_\alpha$.
Note that the underlying term of this commutator is supported on at most $\ell+(\ell-1)$ qubits.
We have a similar analysis for the first two layers:
\begin{align}
    \sum_{\gamma_2,\gamma_3=0}^{j\Upsilon+1}\|[H_{\gamma_3}^{(j)},[H_{\gamma_2},H_{\gamma_1}]]
    \leq&\sum_{\gamma_3=0}^{j\Upsilon+1}\sum_{\gamma:P_\gamma\in H_{\gamma_3}^{(j)}}\sum_{\gamma_2=0}^{j\Upsilon+1}\sum_{\beta:P_\beta\in H_{\gamma_2}}\sum_{\alpha:P_\alpha\in H_{\gamma_1}}\|[s_\gamma P_\gamma,[s_\beta P_\beta,s_\alpha P_\alpha]]\|\notag\\
    \leq& \sum_{\gamma_2=0}^{j\Upsilon+1}\sum_{\beta:P_\beta\in H_{\gamma_2}}\sum_{\alpha:P_\alpha\in H_{\gamma_1}}2(2l-1)b\|[s_\beta P_\beta,s_\alpha P_\alpha]\|\\
    \leq& \sum_{\alpha:P_\alpha\in H_{\gamma_1}}2^2b^2(2\ell-1)\ell\|s_\alpha P_\alpha\|.
\end{align}
Recursively, as each time the support of underlying terms increases by $\ell-1$, the overall nested commutator is bounded by
\begin{align}\label{eq:nc}
    \sum_{\gamma_1=0}^{j\Upsilon+1}\sum_{\gamma_2,\cdots,\gamma_{p+1}=0}^{j\Upsilon+1}\|[H^{(j)}_{\gamma_{p+1}},\cdots[H^{(j)}_{\gamma_{3}},[H_{\gamma_2},H_{\gamma_1}]]\cdots]\|\leq&(2b)^{p}\prod_{i=1}^{p}(i\ell-i+1)\sum_{\gamma_1=0}^{j\Upsilon+1}\sum_{\alpha:P_\alpha\in H_{\gamma_1}}\|s_\alpha P_\alpha\|\notag\\
    =&\order{\sum_{k=0}^{j\Upsilon+1}\sum_{\alpha:P_\alpha\in H_{k}}|s_\alpha|}
    =\order{\sum_{k=0}^{j\Upsilon+1}h_k},
\end{align}
where the order $p$, locality $\ell$ and single-site norm $b$ are chosen as constant numbers, and $h_k\coloneqq \|H_k\|_1$.

Furthermore, we will show that the unitary generated by Algorithm~\ref{alg:ap-rpf} has the identical effects as $\mathscr{S}_{\text{V}}(t)$.
According to Lemma~\ref{lm:ap-optimal}, consider the circumstance just before implementing the stage $\upsilon_0$ in step $j_0$ of $\mathscr{S}_{\text{V}}(t)$, the real-time support of $O$ is
\begin{align}
    S(O)_{j_0,\upsilon_0}\coloneqq&\biguplus_{\upsilon=1}^{\upsilon_0-1}\biguplus_{\gamma=0}^{j_0\Upsilon+1}\mathrm{e}^{\ii \tau a_{(\upsilon,\gamma)}H^{(j_0)}_{\pi^{\text{eo}}_{(j_0-1)\Upsilon+\upsilon}(\gamma)}}\uplus\biguplus_{j=1}^{j_0-1}\biguplus_{\upsilon=1}^{\Upsilon}\biguplus_{\gamma=0}^{j\Upsilon+1}\mathrm{e}^{\ii\tau a_{(\upsilon,\gamma)}H^{(j)}_{\pi^{\text{eo}}_{(j-1)\Upsilon+\upsilon}(\gamma)}}\uplus E_0^S\notag\\
    =&\biguplus_{\upsilon=1}^{\upsilon_0-1}\biguplus_{\gamma=0}^{j_0\Upsilon+1}\mathrm{e}^{\ii\tau a_{(\upsilon,\gamma)}H^{(j_0)}_{\pi^{\text{eo}}_{(j_0-1)\Upsilon+\upsilon}(\gamma)}}\uplus \bigcup_{k=0}^{(j_0-1)\Upsilon} E^S_{k}\notag\\
    =&\bigcup_{k=0}^{(j_0-1)\Upsilon+\upsilon_0-1}E^S_k.
\end{align}
Since the propagated support never reaches the range of tail sub-Hamiltonians, we can directly apply the same analysis as Eq.~\eqref{eq:even-odd ws} to get the above equations.
Therefore, $\{H^{(j)}_{\upsilon_0+1},H^{(j)}_{\upsilon_0+2},\cdots\}$ commutes with the evolved observable.
During stage $\upsilon_0$ of the step $j$, $H^{(j)}_{\upsilon_0+1}$ is always implemented before $H^{(j)}_{\upsilon_0}$.
Hence $H^{(j)}_{\upsilon_0+1}$ is still irrelevant for $O(t)$ during stage $\upsilon_0$.
Removing these irrelevant exponentials, we can transform $\mathscr{S}_{\text{V}}(t)$ to $U$.
The overall simulation error by Algorithm~\ref{alg:ap-rpf} is
\begin{align}
        \|\mathrm{e}^{\ii Ht}O\mathrm{e}^{-\ii Ht}-U OU^\dagger\|\leq&\|\mathrm{e}^{\ii Ht}O\mathrm{e}^{-\ii Ht}-\mathscr{S}_{\text{V}}(t) O\mathscr{S}_{\text{V}}(t)^\dagger\|+\|\mathscr{S}_{\text{V}}(t) O\mathscr{S}_{\text{V}}(t)^\dagger-U OU^\dagger\|\notag\\
        \leq&\sum_{j=1}^r\left\|\mathrm{e}^{\ii H\tau}\prod_{h=1}^{j-1}\mathscr{S}_{h}(\tau)O\prod_{h=1}^{j-1}\mathscr{S}_{j-h}(\tau)^\dagger \mathrm{e}^{-\ii H\tau}-\prod_{h=1}^{j}\mathscr{S}_{h}(\tau) O\prod_{h=1}^{j}\mathscr{S}_{j+1-h}(\tau)^\dagger\right\|\notag\\
        &+\|\mathscr{S}_{\text{V}}(t) O\mathscr{S}_{\text{V}}(t)^\dagger-U OU^\dagger\|\notag\\
        \leq& 2\|O\|\sum_{j=1}^r\epsilon_j+0\notag\\
        =&\order{\|O\|\frac{\sum_{k=0}^{r\Upsilon+1}h_kt^{p+1}}{r^{p}}}
\end{align}
\end{proof}

We also acknowledge the potential presence of multiple terms within the sub-Hamiltonians $\{H_k^S\}$, making it challenging to construct perfect matrix exponentials of sub-Hamiltonians as required in Algorithm.~\ref{alg:ap-rpf}.
To address this issue, we propose further decomposing each sub-Hamiltonian into Pauli terms and embedding an additional first-order Suzuki-Trotter formula to approximate matrix exponentials.
This simplified implementation does not introduce additional errors since the proof recruits the same analysis as if the embedded formula were from the first-order Trotter.
In calculating the norm in Eq.~\eqref{eq:nc}, we have already relaxed the case of individual sub-Hamiltonians to a series of Pauli terms therein. This relaxed treatment already encompasses errors incurred by further embedding first-order formulae for the matrix exponentials.

\section{Short-Time Simulation for Global Observables}\label{sec:append-multi}
In the realm of quantum physics, measuring a global property of a system holds significant importance for extracting global information.
Nevertheless, the support of this global observable might already saturate the whole space, rendering the preceding support analysis inapplicable.
In this section, we consider a special family of global observables that consist of summations of local observables, where the support of each individual summand remains local during brief temporal simulation.
Therefore, we can expect the short-time advantage from the support analysis to remain by viewing each summand separately.

In this section, instead of adopting the above-mentioned configuration which is only optimal for a specific observable, we employ a another construction of the product formula based on the intrinsic geometry of the Hamiltonian.
We show in the following that this configuration can universally suppress the expansion of an arbitrary local observable.
Initially, for an $n$-qubit Hamiltonian $H$, we adopt the elementary Pauli decomposition, $H=\sum_{\alpha\in{\sf P}^n}s_\alpha P_\alpha$ and subsequently regroup components therein according to their supports.
This regrouping is expressed as:
\begin{gather}\label{eq:support-de}
    H=\sum_{i=1}^LH_{S_i},\ \text{where }\forall i\in[L],\ H_{S_i}\coloneqq\sum_{P_\alpha:S(P_\alpha)\subseteq S_i} s_\alpha P_\alpha\neq0 \text{ with at least one Pauli having support $S_i$}.
\end{gather}
To make this regrouping nontrivial, it is further stipulated that $S_i\not\subseteq S_j$ for all $i,j$.
This regrouping yields a valid decomposition of $H$.
We then construct the interaction hypergraph.
\begin{definition}
    Regarding the Pauli decomposition $H=\sum_{\alpha\in{\sf P}^n}s_\alpha P_\alpha$, the sets $\{S_i\}_{i=1}^L$ are defined such that every nonzero Pauli term belongs to some $S_i$, with no $S_i$ belonging to any other $S_j$.
Additionally, at least one Pauli term has support equal to every $S_i$.
The interaction hypergraph $G$ consists of the set of $n$ qubits and $\{S_i\}_{i=1}^L$ as its vertex and hyperedge sets, respectively. 
\end{definition}
Next, we address the edge coloring problem of $G$, where edges covering same qubits are assigned different colors.
Suppose we have a coloring of interaction graph $G$ of $H$ using $\chi$ colors.
We represent the coloring as a map $\varphi:\{S_i\}_{i=1}^L\rightarrow [\chi]$, where $\varphi(S_i)$ denotes the color assigned to $S_i$.
We can rearrange the regrouped sub-Hamiltonians so that those with the same colors are listed together.
We elaborate on the product formula based on this configuration in Algorithm~\ref{alg:ap-mpf}.
This product formula achieves a  $\chi$-approximated optimal expansion of the support for an arbitrary local observable, as asserted in Lemma~\ref{lm:ap-global_support}.

\begin{algorithm}[t]
    \caption{Chromatic Product Formula}\label{alg:ap-mpf}
    \SetKwInOut{Input}{input}\SetKwInOut{Output}{output}
    \Input{Hamiltonian $H$; Evolving time $t$; Step number $r$; Order $p$; Coloring $\varphi$ and number of colors $\chi$}
    \Output{Unitary $U$ organized from product formula}
    Regroup $H$ into $H_{S_1},\cdots,H_{S_L}$\;
   Adopt Suzuki-Trotter coefficients $a_{(\upsilon,\gamma)}$ from~\cite{suzuki1991general}\;
    $\Upsilon\leftarrow2\cdot5^{p/2-1}$, $\tau\leftarrow t/r$, $U\leftarrow I$\;
    \For{$j=1,\cdots,r$}{
        \For{$\upsilon=1,\cdots,\Upsilon$}{
            \eIf{$\upsilon$ is odd}{
                $U\leftarrow\prod_{c=1}^{\chi}\prod_{\gamma=1\to L:\varphi(S_\gamma)=c}\mathrm{e}^{\ii a_{(\upsilon,\gamma)}\tau H_{S_\gamma}}\cdot U$
            }{
                $U\leftarrow\prod_{c=\chi}^{1}\prod_{\gamma=L\to 1:\varphi(S_\gamma)=c}\mathrm{e}^{\ii a_{(\upsilon,\gamma)}\tau H_{S_\gamma}}\cdot U$
            }
        }
    }
\end{algorithm}

\begin{lemma}[Chromatic Optimality]\label{lm:ap-global_support}
    For an arbitrary set of observables $\{O_m\}$ and a Hamiltonian $H$, an $r$-step product formula $U$ from Alg.~\ref{alg:ap-mpf} simultaneously 
 enlarges the support of every $UO_mU^\dagger$ to at most $\bigcup_{k=0}^{(\chi-1) r\Upsilon+1}E^{S(O_m)}_k$. Given that $\chi$ is a constant, the support expansion of all observables is optimally slow.
\end{lemma}
\begin{proof}
We start from the first stage of step one.
    Recalling the notation $\uplus$ in Definition~\ref{def:oplus}, we first consider the effect of sub-Hamiltonian with the first color.
    Suppose $\biguplus_{\gamma=1\to L:\varphi(S_\gamma)=1}\mathrm{e}^{\ii a_{(1,\gamma)}\tau H_{S_\gamma}}\uplus S(O)$ contains some qubits in $E_2^{S(O)}$.
    There must be a sub-hamiltonian $H_{S_j}$ with $\varphi(S_j)=1$ such that $H_{S_j}$ firstly introduces the some qubit $Q_i\in E_2^{S(O)}$ to $\biguplus_{\gamma=1\to j-1:\varphi(S_\gamma)=1}\mathrm{e}^{\ii a_{(1,\gamma)}\tau H_{S_\gamma}}\uplus S(O)$.
    
    According to the definition of the interaction graph and coloring, every two sub-Hamiltonians with the same color are mutually disjoint.
    Since $S_j$ intersects with $E_2^{S(O)}$ and there exists a Pauli in the elementary decomposition of $H$ with support $S_j$, it is easy to verify that $S_j$ is disjoint with $E_0^{S(O)}$.
    Therefore, this sub-Hamiltonian cause no effects on the support expansion,
    \begin{gather}
        \biguplus_{\gamma=1\to j:\varphi(S_\gamma)=1}\mathrm{e}^{\ii a_{(1,\gamma)}\tau H_{S_\gamma}}\uplus S(O)=\biguplus_{\gamma=1\to j-1:\varphi(S_\gamma)=1}\mathrm{e}^{\ii a_{(1,\gamma)}\tau H_{S_\gamma}}\uplus S(O),
    \end{gather}
    which makes a contradiction.
    Thus, we have     $\biguplus_{\gamma=1\to L:\varphi(S_\gamma)=1}\mathrm{e}^{\ii a_{(1,\gamma)}\tau H_{S_\gamma}}\uplus S(O)\subseteq E_0^{S(O)}\cup E_1^{S(O)}$.

    It is noteworthy that the above claim does not depend on either the colors or the steps. 
    Hence, we can iteratively use the above analysis to prove similar statements for the subsequent colors, stages, and steps.
    
   Since Algorithm~\ref{alg:ap-mpf} is a standard $p$th-order Suzuki formula, there exist some repetitions of exponentials between two neighboring stages.
   For example, the first two stages of the formula is
   \begin{gather}
       \left(\prod_{c=\chi}^{1}\prod_{\gamma=L\to 1:\varphi(S_\gamma)=c}\mathrm{e}^{\ii a_{(\upsilon,\gamma)}\tau H_{S_\gamma}}\right)\cdot\left(\prod_{c=1}^{\chi}\prod_{\gamma=1\to L:\varphi(S_\gamma)=c}\mathrm{e}^{\ii a_{(\upsilon,\gamma)}\tau H_{S_\gamma}}\right).
   \end{gather}
The last color on the right and the first color on the left are the same.
These two collections of exponentials with the same color only enlarge the support by one layer.
This is also true for the repetition between two neighboring steps.
Therefore, there are altogether $r\Upsilon-1$ repetitions in the $r$-step $\Upsilon$-stage implementation of Algorithm~\ref{alg:ap-mpf}, which can enlarge the support by at most $(\chi-1)r\Upsilon+1$ layers.
    This completes the proof.
\end{proof}

Up to this point, we have introduced a specific configuration and the resulting product formula, along with proving its suppression of expanding supports of local observables under this configuration. 
Moving forward, we provide a new error bound for the proposed algorithm. 
By considering each summand independently, we can bound the error by summing over the widths of all light cones.

\begin{theorem}[Summation-Observable Error]\label{thm:ap-multiple}
Consider the summation $O=\sum_{m=1}^MO_m$ where every $O_m$ is a local observable with support $S(O_m)$.
Suppose the $n$-qubit $H$ is $\ell$-local with a constant $\ell$ and 
has bounded interaction per qubit. 
    With light-cone widths $w_{r,m}\coloneqq \sum\nolimits_{k=0}^{(\chi-1) r\Upsilon+3}\|H_k^{S(O_m)}\|_1$, the simulation error by an $r$-step $p$th-order $U$ from Alg.~\ref{alg:ap-mpf} is bounded by $$\|\mathrm{e}^{\ii Ht}O\mathrm{e}^{-\ii Ht}-U OU^\dagger\|=\order{\frac{\sum_{m=1}^M\|O_m\|w_{r,m}t^{p+1}}{r^{p}}}.$$
\end{theorem}
\begin{remark}
    \rm Compared to Algorithm~\ref{alg:ap-rpf}, during the simulation of multiple local observables, it is hard to identify the irrelevant parts since the overall observable $O$ might span the whole system.
    In this sense, we have to recruit the standard Suzuki-Trotter formula.
    Nevertheless, by decomposing $O$ into local pieces, we can still find a similar analysis of the error bounds which is tighter than the worst-case analysis with the operator-norm distance in~\cite{childs2021theory} given $rh<\|H\|_1$ and $\sum_m\|O_m\|=\order{\|O\|}$.
\end{remark}
\begin{proof}
    The proof is similar to that of Theorem~\ref{thm:ap-single}, except that we have to consider the virtual product formulas for multiple observables.
    We pick an arbitrary observable $O_m$ as an example and consider its corresponding virtual formula.
    To avoid cluttering notation, we ignore the superscript $S(O_m)$ of $H^{S(O_m)}_k$ and $E^{S(O_m)}_k$ in this proof.
    The following $r$ different divisions of Hamiltonian are suitable for $r$ Trotter steps correspondingly.
    \begin{align*}
    \text{Step 1:}\ H=&\sum_{\substack{\gamma:\varphi(S_\gamma)=1\\S_\gamma\cap\bigcup_{k=0}^{(\chi-1)\Upsilon+1}E_k\neq\emptyset}}H_{S_\gamma}+\cdots+\sum_{\substack{\gamma:\varphi(S_\gamma)=\chi-1\\S_\gamma\cap\bigcup_{k=0}^{(\chi-1)\Upsilon+1}E_k\neq\emptyset}}H_{S_\gamma}+\sum_{\substack{\gamma:\varphi(S_\gamma)=\chi\\S_\gamma\cap\bigcup_{k=0}^{(\chi-1)\Upsilon+1}E_k\neq\emptyset}}H_{S_\gamma}+\sum_{k=(\chi-1)\Upsilon+3}^{\Gamma_0}H_k=\sum_{c=1}^{\chi+1}H_{1,c},\\
    \text{Step 2:}\ H=&\sum_{\substack{\gamma:\varphi(S_\gamma)=1\\S_\gamma\cap\bigcup_{k=0}^{2(\chi-1)\Upsilon+1}E_k\neq\emptyset}}H_{S_\gamma}+\cdots+\sum_{\substack{\gamma:\varphi(S_\gamma)=\chi-1\\S_\gamma\cap\bigcup_{k=0}^{2(\chi-1)\Upsilon+1}E_k\neq\emptyset}}H_{S_\gamma}+\sum_{\substack{\gamma:\varphi(S_\gamma)=\chi\\S_\gamma\cap\bigcup_{k=0}^{2(\chi-1)\Upsilon+1}E_k\neq\emptyset}}H_{S_\gamma}+\sum_{k=2(\chi-1)\Upsilon+3}^{\Gamma_0}H_k=\sum_{c=1}^{\chi+1}H_{2,c},\\
    &\vdots\\
    \text{Step $r$:}\ H=&\sum_{\substack{\gamma:\varphi(S_\gamma)=1\\S_\gamma\cap\bigcup_{k=0}^{r(\chi-1)\Upsilon+1}E_k\neq\emptyset}}H_{S_\gamma}+\cdots+\sum_{\substack{\gamma:\varphi(S_\gamma)=\chi-1\\S_\gamma\cap\bigcup_{k=0}^{r(\chi-1)\Upsilon+1}E_k\neq\emptyset}}H_{S_\gamma}+\sum_{\substack{\gamma:\varphi(S_\gamma)=\chi\\S_\gamma\cap\bigcup_{k=0}^{r(\chi-1)\Upsilon+1}E_k\neq\emptyset}}H_{S_\gamma}+\sum_{k=r(\chi-1)\Upsilon+3}^{\Gamma_0}H_k=\sum_{c=1}^{\chi+1}H_{r ,c}.\\
\end{align*}
Intuitively, for step $j$, we pick up all terms in $\sum_{k=0}^{j(\chi-1)\Upsilon+2}H_k$ and group them by colors, and we discard all other terms into the tail sub-Hamiltonian.
For simplicity, we denote the sub-Hamiltonians for step $j$ by $\{H_{j,c}\}_{c=1}^{\chi+1}$, where $H_{j,\chi+1}$ denotes the tail term.
Hence we construct the adaptive product formula for $O_m$ as
$$\mathscr{S}_{\text{V},m}(t) O_m\mathscr{S}_{\text{V},m}(t)^\dagger\coloneqq\mathscr{S}_{r,m}(\tau)\cdots\mathscr{S}_{1,m}(\tau) O_m\mathscr{S}_{1,m}(\tau)^\dagger\cdots\mathscr{S}_{r,m}(\tau)^\dagger,$$
where $\tau=\frac{t}{r}$ is the unit evolving time for each Trotter step.
Inside each step, we adopt the trivial back-and-forth permutation $\pi_\upsilon^{\text{bf}}$ to implement the standard Suzuki Trotter.
\begin{gather}
    \forall\,j\in[r],\ \mathscr{S}_{j,m}\coloneqq\prod_{\upsilon=1}^\Upsilon\prod_{c=1}^{\chi+1}\mathrm{e}^{\ii \tau a_{(\upsilon,c)}H_{j,\pi_{\upsilon}^{\text{bf}}(c)}}.
\end{gather}
The distances between this adaptive formula and the ideal evolution are bounded by~\cite{childs2021theory},
\begin{gather}\label{eq:errorbound2}
    \epsilon_{j,m}\coloneqq\|\mathrm{e}^{\ii H\tau}-\mathscr{S}_{j,m}\|=\order{\sum_{c_1,\cdots,c_{p+1}=1}^{\chi+1}\|[H_{j,c_{p+1}},\cdots[H_{j,c_2},H_{j,c_1}]\cdots]\|\tau^{p+1}}.
\end{gather}

Similarly, we can show that substituting $H_{j,\chi+1}$ by $H_{j\chi\Upsilon+2}$ in the innermost commutator causes no changes.
When one of $\{c_1,c_2\}$ equals $\chi+1$, the other must be in $[\chi]$ to keep the commutator nonzero.
It is clear that $H_{j(\chi-1)\Upsilon+4}$ and terms thereafter in the tail have no overlaps with the former $\{H_{k}\}_{k=0}^{j(\chi-1)\Upsilon+2}$.
Therefore, it causes no effects to replace the $H_{j,\chi+1}$ in the innermost commutator by $H_{j(\chi-1)\Upsilon+3}$.

Fixing $c_1$, we can employ the same argument as in the proof of Theorem~\ref{thm:ap-single}.
Since the locality of effective Pauli terms and the qubit-wise coefficients are bounded by constants, we can get the following bound,
\begin{align}\label{eq:nc2}
    \sum_{c_1=1}^{\chi+1}\sum_{c_2,\cdots,c_{p+1}=1}^{\chi+1}\|[H_{j,c_{p+1}},\cdots[H_{j,c_2},H_{j,c_1}]\cdots]\|\leq& (2b)^{p}\prod_{i=1}^{p}(i\ell-i+1)\sum_{c_1=1}^{\chi+1}\sum_{\alpha:P_\alpha\in H_{j,c_1}}\|s_\alpha P_\alpha\|\notag\\
    =&(2b)^{p}\prod_{i=1}^{p}(i\ell-i+1)\sum_{k=0}^{j(\chi-1)\Upsilon+3}\sum_{\alpha:P_\alpha\in H_k}\|s_\alpha P_\alpha\|\notag\\
    \leq&\order{\sum_{k=0}^{j(\chi-1)\Upsilon+3}h_{m,k}}.
\end{align}
Therefore, we bound the error $\epsilon_{j,m}$ by $\order{\sum_{k=0}^{j(\chi-1)\Upsilon+3}h_{m,k}\tau^{p+1}}$ as $\ell$, $b$, and $p$ are constant numbers.

To quantify the simulation errors, there remain gaps between the virtual formulas and the proposed $U$ from Algorithm~\ref{alg:ap-mpf}.
According to Lemma~\ref{lm:ap-global_support}, $O_m$ evolved under the chromatic Suzuki-Trotter $U$ after the $\upsilon$th stage of the $j$th step has support no larger than $\bigcup_{k=0}^{((j-1)\Upsilon+\upsilon)(\chi-1)+1}E_k$.
Thus, all sub-Hamiltonians outside $E_{(\chi-1)j\Upsilon+1}$ are irrelevant to $O_m(t)$ in step $j$'s simulation as they commute with the observable.
Removing these irrelevant unitaries from $U$ remains effectively the same simulation.

For the virtual formula, note that we move all sub-Hamiltonians outside $E_{(\chi-1)j\Upsilon+1}$ to the tail term and leave others the same, which will only slow the expansion of supports.
Therefore, after the $\upsilon$th stage of $j$th step, the support of $O_m(t)$ is smaller than $\bigcup_{k=0}^{((j-1)\Upsilon+\upsilon)(\chi-1)+1}E_k$, and the tail term always commutes with the evolved $O_m(t)$.
The effective unitary is thus the same as that of $U$.

\begin{align}
    \|\mathrm{e}^{\ii Ht}O_m\mathrm{e}^{-\ii Ht}-U O_mU^\dagger\|\leq&\|\mathrm{e}^{\ii Ht}O_m\mathrm{e}^{-\ii Ht}-\mathscr{S}_{\text{V},m}(t) O_m\mathscr{S}_{\text{V},m}(t)^\dagger\|+\|\mathscr{S}_{\text{V},m}(t) O_m\mathscr{S}_{\text{V},m}(t)^\dagger-U O_mU^\dagger\|\notag\\
    \leq&\sum_{j=1}^r\left\|\mathrm{e}^{\ii H\tau}\prod_{i=1}^{j-1}\mathscr{S}_{i,m}(\tau) O_m\prod_{i=1}^{j-1}\mathscr{S}_{j-i,m}(\tau)^\dagger\mathrm{e}^{-\ii H\tau}-\prod_{i=1}^{j}\mathscr{S}_{i,m}(\tau) O_m\prod_{i=1}^{j}\mathscr{S}_{j+1-i,m}(\tau)^\dagger\right\|\notag\\
        &+\|\mathscr{S}_{\text{V},m}(t) O_m\mathscr{S}_{\text{V},m}(t)^\dagger-U O_mU^\dagger\|\notag\\
        \leq& \sum_{j=1}^r2\|O_m\|\epsilon_{j,m}+0\notag\\
        =& \order{\|O_m\|\frac{\sum_{k=0}^{r(\chi-1)\Upsilon+3}h_{m,k}t^{p+1}}{r^{p}}}.
\end{align}
Therefore, this proof can be applied to any other $O_{m'}$.
We can get
\begin{gather}
    \|\mathrm{e}^{\ii Ht}O\mathrm{e}^{-\ii Ht}-U OU^\dagger\|\leq\sum_{m=1}^M\|\mathrm{e}^{\ii Ht}O_m\mathrm{e}^{-\ii Ht}-U O_mU^\dagger\|\leq\order{\sum_{m=1}^M\|O_m\|\frac{\sum_{k=0}^{r(\chi-1)\Upsilon+3}h_{m,k}t^{p+1}}{r^{p}}}
\end{gather}
\end{proof}
Similar to Algorithm~\ref{alg:ap-rpf}, the sub-Hamiltonians ${H_{S_j}}$ in Algorithm~\ref{alg:ap-mpf} may contain several distinct Pauli terms, posing challenges for faithful implementation of matrix exponentials. To address this, we embed additional first-order formulae based on the Pauli decomposition of these sub-Hamiltonians. Importantly, this embedding does not introduce any additional errors, ensuring the validity of Theorem~\ref{thm:ap-multiple}, as Eq.~\eqref{eq:nc2} has already accommodated this relaxation.

\section{Arbitrary-Time Simulation with Random Inputs}\label{Sec:Append-random}

In our preceding analysis, we primarily focus on the advantages of observable knowledge in short-time simulation.
Additionally, it is an intriguing question to study the advantage of simulating for an arbitrary time with observable knowledge.
Although a thorough analysis of this topic is in general challenging, we shift our analysis to the typical average-case simulation performance with random input states.
We discover that the observable knowledge also enhances the average simulation accuracy over a comparably long time for random input states.

To quantify the distance between arbitrary pairs of unitaries $U$ and $U_0$ for a given observable $O$ and an input ensemble $\mu=\{p_i,\psi_i\}$, we employ the average distance over measurement results as a metric:
\begin{equation}
\begin{aligned}
    D(U_0,U)_{O,\mu}\coloneqq \sum_i p_i |\bra{\psi_i}U^{\dagger}O U\ket{\psi_i}-\bra{\psi_i}U_0^{\dagger}O U_0\ket{\psi_i}|. 
\end{aligned}
\end{equation}
This choice is motivated by the fact that experimental implementations of quantum circuits always end with measurements.
As a comparison, the initial random-input distance with an ensemble $\mu$ is defined in~\cite{zhao2022hamiltonian} as
\begin{gather}
D(U_0,U)_\mu\coloneqq\max_{\|O\|=1}\sum_i p_i |\bra{\psi_i}U^{\dagger}O U\ket{\psi_i}-\bra{\psi_i}U_0^{\dagger}O U_0\ket{\psi_i}|.
\end{gather}

In the context of quantum simulation, the ideal unitary is always the temporal evolution while the erroneous one is some simulation method.
Particularly, since all simulation methods consists of $r$ short steps, we assign $\mathscr{U}_0=\mathrm{e}^{\ii Ht/r}$ and $\mathscr{U}=\mathscr{S}(t/r)$.
Therefore, $U_0=\mathscr{U}_0^r$ and $U=\mathscr{U}^r$.
To unveil the typical simulation errors, we need to utilize the Haar-random states as our inputs.
Nevertheless, it is more accessible to utilize its discrete version known as the complex projective $t$-design $\mu_t$, from which sampled states can mimic the Haar-random states for all polynomials with order $t$~\cite{mele2024introduction}.
Specifically, to show the average distance of random inputs based on measurements in the following theorem, we only recruit the 2-design ensemble, which means that the requirement of advantages is way easier to fulfill than the Haar-random ensemble.

\begin{theorem}[2-design Random-Input]\label{thm:2design}
    For a 2-design ensemble $\mu_2$ of quantum states and an observable $O$, we can bound the average distance of simulation:
    \begin{gather}
        D(\mathscr{U}_0^r,\mathscr{U}^r)_{O,\mu_2}\leq\sqrt{2}r\frac{\|O\|_2\cdot\|\mathscr{M}\|_2}{\sqrt{d(d+1)}},
    \end{gather}
    where $\mathscr{M}\coloneqq\mathscr{U}_0^\dag\mathscr{U}-I$ represents the multiplicative error for each small step, $d=2^n$ is the dimension of the $n$-qubit Hilbert space, and $\|A\|_p=[\Tr(|A|^p)]^{1/p}$ denotes the Schatten $p$-norm.
    The variance of errors can also be bounded by
\begin{gather*}
     \text{Var}(\mathscr{U}_0^r,\mathscr{U}^r)_{O,\mu_2
     }\le\frac{2r^2\|O\|^2_2\cdot\|\mathscr{M}\|^2_2}{d(d+1)}.
 \end{gather*}
\end{theorem}
\begin{proof}
    \begin{align}
        D(\mathscr{U}_0^r,\mathscr{U}^r)_{O,\mu_2}&=\int_{\psi} |\bra{\psi}\mathscr{U}_0^rO (\mathscr{U}_0^\dag)^r\ket{\psi}-\bra{\psi}\mathscr{U}^rO (\mathscr{U}^\dag)^r\ket{\psi}| d\mu_2(\psi)\notag\\
        &\leq\int_{\psi}\Big|\Tr \Big(\big(((\mathscr{U}_0^\dag)^r\ket{\psi}\bra{\psi}\mathscr{U}_0^r-(\mathscr{U}_0^\dag)^{r-1}\mathscr{U}^\dag\ket{\psi}\bra{\psi}\mathscr{U}\mathscr{U}_0^{r-1})+\notag\\
        &\ \ \cdots+(\mathscr{U}_0^\dag(\mathscr{U}^\dag)^{r-1}\ket{\psi}\bra{\psi}\mathscr{U}^{r-1}\mathscr{U}_0-(\mathscr{U}^\dag)^{r}\ket{\psi}\bra{\psi}\mathscr{U}^{r})\big)O\Big)\Big|d\mu_2(\psi)\notag\\
        &\leq\int_{\psi}\abs{\Tr (((\mathscr{U}_0^\dag)^r\ket{\psi}\bra{\psi}\mathscr{U}_0^r-(\mathscr{U}_0^\dag)^{r-1}\mathscr{U}^\dag\ket{\psi}\bra{\psi}\mathscr{U}\mathscr{U}_0^{r-1})O)}d\mu_2(\psi)+\notag\\
        &\ \ \cdots+\int_{\psi}\abs{\Tr((\mathscr{U}_0^\dag(\mathscr{U}^\dag)^{r-1}\ket{\psi}\bra{\psi}\mathscr{U}^{r-1}\mathscr{U}_0-(\mathscr{U}^\dag)^{r}\ket{\psi}\bra{\psi}\mathscr{U}^{r})O)}d\mu_2(\psi)\notag\\
        &\leq \int_{\psi_0}\Tr (\abs{(\mathscr{U}_0^\dag\ket{\psi_0}\bra{\psi_0}\mathscr{U}_0-\mathscr{U}^\dag\ket{\psi_0}\bra{\psi_0}\mathscr{U})\mathscr{U}_0^{r-1}O(\mathscr{U}_0^\dag)^{r-1}})d\mu_2(\psi_0)+\notag\\
        &\ \ \cdots+\int_{\psi_{r-1}}\Tr (\abs{(\mathscr{U}_0^\dag\ket{\psi_{r-1}}\bra{\psi_{r-1}}\mathscr{U}_0-\mathscr{U}^\dag\ket{\psi_{r-1}}\bra{\psi_{r-1}}\mathscr{U})\mathscr{U}_0^{r-1}O(\mathscr{U}_0^\dag)^{r-1}})d\mu_2(\psi_{r-1})\notag\\
        &\leq \sum_{i=0}^{r-1}\int_{\psi_{i}}\|\mathscr{U}_0^\dag\ket{\psi_{i}}\bra{\psi_{i}}\mathscr{U}_0-\mathscr{U}^\dag\ket{\psi_{i}}\bra{\psi_{i}}\mathscr{U}\|_2\cdot\|O\|_2d\mu_2(\psi_i)\notag\\
        &\leq r\|O\|_2\sqrt{\frac{2d^2-2\Tr(I+\mathscr{M})\Tr(I+\mathscr{M}^\dag)}{d(d+1)}}\notag\\
        &= r\|O\|_2\sqrt{\frac{2\Tr(\mathscr{M}\mathscr{M}^\dag)-2\Tr(\mathscr{M})\Tr(\mathscr{M}^\dag)}{d(d+1)}}\notag\\
        &\leq \sqrt{2}r\frac{\|O\|_2\cdot\|\mathscr{M}\|_2}{\sqrt{d(d+1)}}.
    \end{align}
    The first and second inequalities in the derivation come from the triangle inequality.
    In the third inequality, we denote each state $(\mathscr{U}^\dag)^{i}\ket{\psi}\bra{\psi}\mathscr{U}^{i}$ by $\ket{\psi_i}\bra{\psi_i}$.
    In the fourth inequality, we recruit the H\"{o}lder's inequality for the Schatten 1-norm.
    In the fifth inequality, we first employ the Cauchy-Schwartz inequality for the integral.
    Due to the unitary invariance of $t$-design ensemble, each $\{\psi_i\}$ is still a 2-design ensemble.
    Therefore, we can use the 2-design property:
    \begin{gather}
        \int_\psi\ket{\psi}\bra{\psi}^{\otimes 2}d\mu_2(\psi)=\frac{I+S}{d(d+1)}.
    \end{gather}
    Due to the unitarity of $I+\mathscr{M}$, we get $\mathscr{M}+\mathscr{M}^\dag=-\mathscr{M}\mathscr{M}^\dag$, leading to the second equality.

    For the variance, we similarly calculate the second-order moment first.
    \begin{align}
        V(\mathscr{U}_0^r,\mathscr{U}^r)_{O,\mu_2}&\coloneqq\int_{\psi} (|\bra{\psi}\mathscr{U}_0^rO (\mathscr{U}_0^\dag)^r\ket{\psi}-\bra{\psi}\mathscr{U}^rO (\mathscr{U}^\dag)^r\ket{\psi}|)^2d\mu_2(\psi)\notag\\
        &\leq\int_{\psi}\Big|\Tr \Big(\big(((\mathscr{U}_0^\dag)^r\ket{\psi}\bra{\psi}\mathscr{U}_0^r-(\mathscr{U}_0^\dag)^{r-1}\mathscr{U}^\dag\ket{\psi}\bra{\psi}\mathscr{U}\mathscr{U}_0^{r-1})+\notag\\
        &\ \ \cdots+(\mathscr{U}_0^\dag(\mathscr{U}^\dag)^{r-1}\ket{\psi}\bra{\psi}\mathscr{U}^{r-1}\mathscr{U}_0-(\mathscr{U}^\dag)^{r}\ket{\psi}\bra{\psi}\mathscr{U}^{r})\big)O\Big)\Big|^2d\mu_2(\psi)\notag\\
        &\leq\int_{\psi}\Big(\abs{\Tr (((\mathscr{U}_0^\dag)^r\ket{\psi}\bra{\psi}\mathscr{U}_0^r-(\mathscr{U}_0^\dag)^{r-1}\mathscr{U}^\dag\ket{\psi}\bra{\psi}\mathscr{U}\mathscr{U}_0^{r-1})O)}+\notag\\
        &\ \ \cdots+\abs{\Tr((\mathscr{U}_0^\dag(\mathscr{U}^\dag)^{r-1}\ket{\psi}\bra{\psi}\mathscr{U}^{r-1}\mathscr{U}_0-(\mathscr{U}^\dag)^{r}\ket{\psi}\bra{\psi}\mathscr{U}^{r})O)}\Big)^2d\mu_2(\psi)\notag\\
        &\leq\int_{\psi}\Big(\Tr (\abs{((\mathscr{U}_0^\dag)^r\ket{\psi}\bra{\psi}\mathscr{U}_0^r-(\mathscr{U}_0^\dag)^{r-1}\mathscr{U}^\dag\ket{\psi}\bra{\psi}\mathscr{U}\mathscr{U}_0^{r-1})O})+\notag\\
        &\ \ \cdots+\Tr(\abs{(\mathscr{U}_0^\dag(\mathscr{U}^\dag)^{r-1}\ket{\psi}\bra{\psi}\mathscr{U}^{r-1}\mathscr{U}_0-(\mathscr{U}^\dag)^{r}\ket{\psi}\bra{\psi}\mathscr{U}^{r})O})\Big)^2d\mu_2(\psi)\notag\\
        &\leq\int_{\psi}r\Big(\Tr (\abs{((\mathscr{U}_0^\dag)^r\ket{\psi}\bra{\psi}\mathscr{U}_0^r-(\mathscr{U}_0^\dag)^{r-1}\mathscr{U}^\dag\ket{\psi}\bra{\psi}\mathscr{U}\mathscr{U}_0^{r-1})O})^2+\notag\\
        &\ \ \cdots+\Tr(\abs{(\mathscr{U}_0^\dag(\mathscr{U}^\dag)^{r-1}\ket{\psi}\bra{\psi}\mathscr{U}^{r-1}\mathscr{U}_0-(\mathscr{U}^\dag)^{r}\ket{\psi}\bra{\psi}\mathscr{U}^{r})O})^2\Big)d\mu_2(\psi)\notag\\
        &\leq \int_{\psi_0}r\Tr (\abs{(\mathscr{U}_0^\dag\ket{\psi_0}\bra{\psi_0}\mathscr{U}_0-\mathscr{U}^\dag\ket{\psi_0}\bra{\psi_0}\mathscr{U})\mathscr{U}_0^{r-1}O(\mathscr{U}_0^\dag)^{r-1}})^2d\mu_2(\psi_0)+\notag\\
        &\ \ \cdots+\int_{\psi_{r-1}}r\Tr (\abs{(\mathscr{U}_0^\dag\ket{\psi_{r-1}}\bra{\psi_{r-1}}\mathscr{U}_0-\mathscr{U}^\dag\ket{\psi_{r-1}}\bra{\psi_{r-1}}\mathscr{U})\mathscr{U}_0^{r-1}O(\mathscr{U}_0^\dag)^{r-1}})^2d\mu_2(\psi_{r-1})\notag\\
        &\leq\sum_{i=0}^{r-1}\int_{\psi_{i}}r\|\mathscr{U}_0^\dag\ket{\psi_{i}}\bra{\psi_{i}}\mathscr{U}_0-\mathscr{U}^\dag\ket{\psi_{i}}\bra{\psi_{i}}\mathscr{U}\|^2_2\cdot\|O\|^2_2d\mu_2(\psi_i)\notag\\
        &\leq r\|O\|_2\sqrt{\frac{2d^2-2\Tr(I+\mathscr{M})\Tr(I+\mathscr{M}^\dag)}{d(d+1)}}\notag\\
        &= r^2\|O\|^2_2\frac{2\Tr(\mathscr{M}\mathscr{M}^\dag)-2\Tr(\mathscr{M})\Tr(\mathscr{M}^\dag)}{d(d+1)}\notag\\
        &\leq \frac{2r^2\|O\|^2_2\cdot\|\mathscr{M}\|^2_2}{d(d+1)}.
    \end{align}
    Here, we find all the derivation steps are similar except that we recruit the mean inequality chain in the fourth inequality.
    Therefore, we can bound the value of the variance.
    \begin{align}
        \text{Var}(\mathscr{U}_0^r,\mathscr{U}^r)_{O,\mu_2}=V(\mathscr{U}_0^r,\mathscr{U}^r)_{O,\mu_2}-D(\mathscr{U}_0^r,\mathscr{U}^r)_{O,\mu_2}^2\leq \frac{2r^2\|O\|^2_2\cdot\|\mathscr{M}\|^2_2}{d(d+1)}.
    \end{align}
\end{proof}
Focusing on the standard $r$-step $p$th-order Suzuki product formula, we can further bound the error by the following corollary.
\begin{corollary}[Product-Formula Average Error]\label{co:productrandom2}
    A standard $r$-step $p$th-order Suzuki product formula $\mathscr{S}_p$ has an average simulation error $\order{T_2\|O\|_2t^{p+1}d^{-1/2}r^{-p})}$  with
\begin{gather*}
    T_2\coloneqq\sum_{\gamma_1,\dots,\gamma_{p+1}=1}^\Gamma\frac{1}{\sqrt{d}}\left\|[H_{\gamma_{p+1}},[H_{\gamma_p},\dots,[H_{\gamma_2},H_{\gamma_{1}}]]] \right\|_2,
\end{gather*}
and $H=\sum_{\gamma=1}^\Gamma H_\gamma$ is the decomposition.
For $p=2$ case, we have a triangle-bound as
\begin{gather*}
    D(\mathscr{U}_0^r,\mathscr{U}^r)_{O,\mu_2}\leq\frac{\sqrt{2}t^3\|O\|_2}{12r^2d}\left(\sum_{\gamma_1=1}^\Gamma\left\|\left[\sum_{\gamma_2=\gamma_1+1}^\Gamma H_{\gamma_2},\left[\sum_{\gamma_3=\gamma_1+1}^\Gamma H_{\gamma_3},H_{\gamma_1}\right]\right]\right\|_2+\frac{1}{2}\sum_{\gamma_1=1}^\Gamma\left\|\left[H_{\gamma_1},\left[H_{\gamma_1},\sum_{\gamma_2=\gamma_1+1}^\Gamma H_{\gamma_2}\right]\right]\right\|_2\right).
\end{gather*}
\end{corollary}
\begin{proof}
    For the $p$th-order bound, we need to analyze the $\|\mathscr{M}\|_2$ for the Suzuki product formula.
    From Theorem 8 in~\cite{zhao2022hamiltonian}, we can bound 
    \begin{gather}
        \frac{\|\mathscr{M}\|_2}{\sqrt{d}}\leq\order{\frac{T_2t^{p+1}}{r^p}}.
    \end{gather}
    Therefore, we can prove the first statement according to Theorem~\ref{thm:2design}.

    In the case of a standard second-order Suzuki-Trotter formula, we can further prove the triangle bound which offers a tighter estimation of the error.
We first consider the case where we only have two sub-Hamiltonians $H=A+B$, and $\mathscr{U}=\mathscr{S}_2=\mathrm{e}^{\ii A\tau/2}\mathrm{e}^{\ii B\tau}\mathrm{e}^{\ii A\tau/2}$.
Therefore, the additive error is
\begin{align}
    \mathscr{A}=&\int_0^\tau \d\tau_1\int_0^{\tau_1}\d\tau_2\int_0^{\tau_2} \d\tau_3 e^{\ii(\tau-\tau_1)H}e^{\ii\tau_1A/2}\notag\\
    &\cdot \left( e^{\ii\tau_3B}\left[\ii B,\left[\ii B,\ii\frac{A}{2}\right]\right]e^{-\ii\tau_3B}+ e^{-\ii\tau_3A/2}\left[-\ii\frac{A}{2},\left[-\ii\frac{A}{2},-\ii B\right]\right]e^{\ii\tau_3A/2}\right) e^{\ii\tau_1B} e^{\ii\tau_1A/2}.
\end{align}
Noticing that all the adjunct matrix exponentials here are unitary matrices, we can calculate and bound the 2-norm of $\mathscr{M}$ by Cauchy-Schwarz inequality as
\begin{align}
    \Tr(\mathscr{M}\mathscr{M}^\dag)=\Tr(\mathscr{A}\mathscr{A}^\dag)\leq& \frac{\tau^{6}}{6^2}\left(\sqrt{\Tr(|[\ii B,[\ii B,\ii\frac{A}{2}]]|^2)}+\sqrt{\Tr(|[-\ii\frac{A}{2},[-\ii\frac{A}{2},-\ii B]]|^2)}\right)^2\notag\\
    =&\frac{\tau^{6}}{6^2}\left(\frac{1}{2}\sqrt{\Tr(|[B,[B,A]]|^2)}+\frac{1}{4}\sqrt{\Tr(|[A,[A,B]]|^2)}\right)^2.
\end{align}
Consequently, we have 
\begin{gather}
 \|\mathscr{M}\|_2\leq\frac{\tau^3}{12}\|[B,[B,A]]\|_2+ \frac{\tau^3}{24} \|[A,[A,B]]\|_2.
\end{gather}

As for a more general case where $H=\sum_{\gamma=1}^\Gamma H_\gamma$, we first define  
\begin{equation}
\mathscr{U}[\gamma_0]\coloneqq\mathrm{e}^{\ii H_{\gamma_0}\tau/2} \mathrm{e}^{\ii\sum_{\gamma=\gamma_0+1}^\Gamma H_{\gamma}\tau}\mathrm{e}^{\ii H_{\gamma_0}\tau/2},\quad \mathscr{U}_0[\gamma_0]\coloneqq\mathrm{e}^{\ii\sum_{\gamma=\gamma_0}^\Gamma H_{\gamma}\tau}.
\end{equation}
We then use the triangle inequality to decompose the additive error as
\begin{align}
    \|\mathscr{M}\|_2=\|\mathscr{A}\|_2=\|\mathscr{U}_0-\mathscr{U}\|_2\leq\sum_{\gamma=1}^\Gamma\|\mathscr{U}_0[\gamma]-\mathscr{U}[\gamma]\|_2.
\end{align}
Note that all $\{\mathscr{U}[\gamma]\}$ are second-order product formulas with two sub-Hamiltonians.
Therefore, we can bound the 2-norm according to the preceding calculation:
\begin{align}
    \|\mathscr{M}\|_2\leq\frac{\tau^3}{12}\sum_{\gamma_1=1}^\Gamma\left\|\left[\sum_{\gamma_2=\gamma_1+1}^\Gamma H_{\gamma_2},\left[\sum_{\gamma_3=\gamma_1+1}^\Gamma H_{\gamma_3},H_{\gamma_1}\right]\right]\right\|_2+\frac{\tau^3}{24}\sum_{\gamma_1=1}^\Gamma\left\|\left[H_{\gamma_1},\left[H_{\gamma_1},\sum_{\gamma_2=\gamma_1+1}^\Gamma H_{\gamma_2}\right]\right]\right\|_2.
\end{align}
The second statement is therefore proved according to Theorem~\ref{thm:2design}.
\end{proof}

Even when it might be challenging to sample states from a 2-design ensemble, we can get a weaker error analysis that still demonstrates error reductions compared to the previous analysis without observable knowledge.
In the following, we elaborate on this 1-design error bound.

\begin{theorem}[1-Design Error]
For a 1-design ensemble $\mu_1$ of quantum states and an observable $O$, we bound the average distance of simulation:
\begin{gather*}\label{eq:4norm}
     D(\mathscr{U}_0^r,\mathscr{U}^r)_{O,\mu_1}\le\frac{2r\|\mathscr{A}\|_4\|O\|_4}{\sqrt{d}},
     \end{gather*}
where $\mathscr{A}\coloneqq \mathscr{U}_0-\mathscr{U}$ is the additive error for each small step, $d=2^n$ is the dimension of the $n$-qubit Hilbert space, and $\|A\|_p=[\Tr(|A|^p)]^{1/p}$ denotes the Schatten $p$-norm. 
The variance of errors can also be bounded by
\begin{gather*}\label{eq:var}
     \text{Var}(\mathscr{U}_0^r,\mathscr{U}^r)_{O,\mu_1}\le\frac{4r^2\|\mathscr{A}\|^2_4\|O\|^2_4}{d}.
 \end{gather*}
\label{thm:ap-random}
\end{theorem}
\noindent It is clear to see that $\mathscr{A}$ is analog to $\mathscr{M}$ up to a unitary.
Therefore, all Schatten norms on these two errors are the same.
This theorem claims a slightly weaker bound on the average error due to degradation of the randomness.
Nevertheless, we will see that the asymptotic scaling of this bound is nearly the same as Theorem~\ref{thm:2design}.
\begin{proof}

Following the definition of the average distance, we have
    \begin{align}
D(\mathscr{U}_0^r,\mathscr{U}^r)_{O,\mu_1}&=\int_{\psi} |\bra{\psi}\mathscr{U}_0^rO (\mathscr{U}_0^\dag)^r\ket{\psi}-\bra{\psi}\mathscr{U}^rO (\mathscr{U}^\dag)^r\ket{\psi}| d\mu_1(\psi)\notag\\
&\le  \int_{\psi} |\bra{\psi}\mathscr{U}_0^rO (\mathscr{U}_0^\dag)^r\ket{\psi}-\bra{\psi}\mathscr{U}_0^rO (\mathscr{U}^\dag)^r\ket{\psi}| d\mu_1(\psi)+
\int_{\psi} |\bra{\psi}\mathscr{U}_0^rO (\mathscr{U}^\dag)^r\ket{\psi}-\bra{\psi}\mathscr{U}^rO (\mathscr{U}^\dag)^r\ket{\psi} | d\mu_1(\psi)\notag\\
&\le \sum_{i=0}^{r-1} \int_{\psi} |\bra{\psi}\mathscr{U}_0^rO(\mathscr{U}_0^\dag)^{r-i}(\mathscr{U}^\dag)^i \ket{\psi}-\bra{\psi}\mathscr{U}_0^rO(\mathscr{U}_0^\dag)^{r-i-1}(\mathscr{U}^\dag)^{i+1} \ket{\psi}| d\mu_1(\psi)\notag\\
&\ \ \ +\sum_{i=0}^{r-1}\int_{\psi} |\bra{\psi}\mathscr{U}_0^{i+1}\mathscr{U}^{r-i-1}O(\mathscr{U}^\dag)^r \ket{\psi}-\bra{\psi}\mathscr{U}_0^{i}\mathscr{U}^{r-i}O(\mathscr{U}^\dag)^r \ket{\psi}| d\mu_1(\psi) \notag\\
&=  \sum_{i=0}^{r-1} \int_{\psi}  |\bra{\psi}\mathscr{U}_0^rO(\mathscr{U}_0^\dag)^{r-i-1}\mathscr{A}^\dag(\mathscr{U}^\dag)^i \ket{\psi}| +  |\bra{\psi}\mathscr{U}_0^i\mathscr{A}\mathscr{U}^{r-i-1} O(\mathscr{U}^\dag)^{r}\ket{\psi}|  d\mu_1(\psi) \notag\\
&\le \sum_{i=0}^{r-1}\int_{\psi} \|(\mathscr{U}_0^\dag)^r\ket{\psi}\|_2 \|O(\mathscr{U}_0^\dag)^{r-i-1}\mathscr{A}^\dag(\mathscr{U}^\dag)^i \ket{\psi}\|_2 + \|(\mathscr{U}^\dag)^{r}\ket{\psi}\|_2 \|O(\mathscr{U}^\dag)^{r-i-1}\mathscr{A}^\dag (\mathscr{U}_0^\dag)^{i} \ket{\psi}\|_2
 d\mu_1(\psi)\notag\\
&= \sum_{i=0}^{r-1} \int_{\psi} \sqrt{\bra{\psi}\mathscr{U}^i\mathscr{A} \mathscr{U}_0^{r-i-1}O^{\dagger}  O (\mathscr{U}_0^\dag)^{r-i-1}\mathscr{A}^\dag(\mathscr{U}^\dag)^i \ket{\psi}}\notag\\
&\ \ \ \ \ +\sqrt{\bra{\psi}\mathscr{U}_0^i\mathscr{A} \mathscr{U}^{r-i-1}O^{\dagger}  O (\mathscr{U}^\dag)^{r-i-1}\mathscr{A}^\dag(\mathscr{U}_0^\dag)^i \ket{\psi}}  d\mu_1(\psi)\notag \\
&\le \sum_{i=0}^{r-1} \sqrt{\int_{\psi} \bra{\psi}\mathscr{U}^i\mathscr{A} \mathscr{U}_0^{r-i-1}O^{\dagger}  O (\mathscr{U}_0^\dag)^{r-i-1}\mathscr{A}^\dag(\mathscr{U}^\dag)^i \ket{\psi} d\mu_1(\psi)}\notag\\
&\ \ \ \ \ +\sqrt{\int_{\psi} \bra{\psi}\mathscr{U}_0^i\mathscr{A} \mathscr{U}^{r-i-1}O^{\dagger}  O (\mathscr{U}^\dag)^{r-i-1}\mathscr{A}^\dag(\mathscr{U}_0^\dag)^i \ket{\psi} d\mu_1(\psi)}\notag\\
&=\sum_{i=0}^{r-1} \sqrt{\Tr(\mathscr{A}^{\dagger}\mathscr{A} \mathscr{U}_0^{r-i-1}O^{\dagger}  O (\mathscr{U}_0^\dag)^{r-i-1} )/d}+\sqrt{\Tr(\mathscr{A}^\dag \mathscr{A} \mathscr{U}^{r-i-1}O^{\dagger}  O (\mathscr{U}^\dag)^{r-i-1})/d}\notag\\
&\le 2r \qty(\Tr(\mathscr{A} \mathscr{A}^\dagger \mathscr{A} \mathscr{A}^\dagger )/d)^{\frac{1}{4}}
\qty(\Tr(O O^{\dagger} O O^{\dagger} )/d)^{\frac{1}{4}},
\end{align}
where $\mu_1(\psi)$ denotes the 1-design ensemble of states.
The fifth line is due to Cauchy-Schwarz inequality $$|\braket{a|b}|\le \sqrt{\braket{a|a}} \sqrt{\braket{b|b}}.$$
The seventh line is from the Cauchy-Schwarz inequality for integrals. 
The eighth line recruits the property of the 1-design for the single-copy integral
\begin{gather}
    \int_\psi \ket{\psi}\bra{\psi} d \mu_1(\psi)=\frac{I}{d}.
\end{gather}
The last line is because $|\Tr(AB)|\le \sqrt{\Tr(AA^{\dagger})}\sqrt{\Tr(BB^{\dagger})}.$
This completes our proof for the first statement.

As for the variance, we first calculate the second-order moment.
\begin{align}
    V(\mathscr{U}_0^r,\mathscr{U}^r)_{O,\mu_1}&=\int_{\psi} (|\bra{\psi}\mathscr{U}_0^rO (\mathscr{U}_0^\dag)^r\ket{\psi}-\bra{\psi}\mathscr{U}^rO (\mathscr{U}^\dag)^r\ket{\psi}|)^2d\mu_1(\psi)\notag\\
\le & \int_{\psi} (|\bra{\psi}\mathscr{U}_0^rO (\mathscr{U}_0^\dag)^r\ket{\psi}-\bra{\psi}\mathscr{U}_0^rO (\mathscr{U}^\dag)^r\ket{\psi}| + |\bra{\psi}\mathscr{U}_0^rO (\mathscr{U}^\dag)^r\ket{\psi}-\bra{\psi}\mathscr{U}^rO (\mathscr{U}^\dag)^r\ket{\psi} |)^2 d\mu_1(\psi)\notag\\
\le& \int_{\psi}\Big(\sum_{i=0}^{r-1}  |\bra{\psi}\mathscr{U}_0^rO(\mathscr{U}_0^\dag)^{r-i}(\mathscr{U}^\dag)^i \ket{\psi}-\bra{\psi}\mathscr{U}_0^rO(\mathscr{U}_0^\dag)^{r-i-1}(\mathscr{U}^\dag)^{i+1} \ket{\psi}| \notag\\
&\ \ \ +\sum_{i=0}^{r-1}|\bra{\psi}\mathscr{U}_0^{i+1}\mathscr{U}^{r-i-1}O(\mathscr{U}^\dag)^r \ket{\psi}-\bra{\psi}\mathscr{U}_0^{i}\mathscr{U}^{r-i}O(\mathscr{U}^\dag)^r \ket{\psi}|\Big)^2 d\mu_1(\psi) \notag\\
= & \int_{\psi}\left(\sum_{i=0}^{r-1} |\bra{\psi}\mathscr{U}_0^rO(\mathscr{U}_0^\dag)^{r-i-1}\mathscr{A}^\dag(\mathscr{U}^\dag)^i \ket{\psi}| +  |\bra{\psi}\mathscr{U}_0^i\mathscr{A}\mathscr{U}^{r-i-1} O(\mathscr{U}^\dag)^{r}\ket{\psi}| \right)^2 d\mu_1(\psi) \notag\\
\le& \int_{\psi} \left(\sum_{i=0}^{r-1}\|(\mathscr{U}_0^\dag)^r\ket{\psi}\|_2 \|O(\mathscr{U}_0^\dag)^{r-i-1}\mathscr{A}^\dag(\mathscr{U}^\dag)^i \ket{\psi}\|_2 + \|(\mathscr{U}^\dag)^{r}\ket{\psi}\|_2 \|O(\mathscr{U}^\dag)^{r-i-1}\mathscr{A}^\dag (\mathscr{U}_0^\dag)^{i} \ket{\psi}\|_2\right)^2
 d\mu_1(\psi)\notag\\
\le& \int_\psi 2r\cdot\left(\sum_{i=0}^{r-1} \|O(\mathscr{U}_0^\dag)^{r-i-1}\mathscr{A}^\dag(\mathscr{U}^\dag)^i \ket{\psi}\|_2^2 +\|O(\mathscr{U}^\dag)^{r-i-1}\mathscr{A}^\dag (\mathscr{U}_0^\dag)^{i} \ket{\psi}\|_2^2\right)d\mu_1(\psi)\notag\\
=& 2r\cdot \sum_{i=0}^{r-1} \int_{\psi} \bra{\psi}\mathscr{U}_0^i\mathscr{A} \mathscr{U}^{r-i-1}O^{\dagger}  O (\mathscr{U}^\dag)^{r-i-1}\mathscr{A}^\dag(\mathscr{U}_0^\dag)^i \ket{\psi}+\bra{\psi}\mathscr{U}_0^i\mathscr{A} \mathscr{U}^{r-i-1}O^{\dagger}  O (\mathscr{U}^\dag)^{r-i-1}\mathscr{A}^\dag(\mathscr{U}_0^\dag)^i \ket{\psi}  d\mu_1(\psi)\notag \\
=&2r\cdot \sum_{i=0}^{r-1} \Tr(\mathscr{A}^{\dagger}\mathscr{A} \mathscr{U}_0^{r-i-1}O^{\dagger}  O (\mathscr{U}_0^\dag)^{r-i-1} )/d+\Tr(\mathscr{A}^\dag \mathscr{A} \mathscr{U}^{r-i-1}O^{\dagger}  O (\mathscr{U}^\dag)^{r-i-1})/d\notag\\
\le& 4r^2 \qty(\Tr(\mathscr{A} \mathscr{A}^\dagger \mathscr{A} \mathscr{A}^\dagger )/d)^{\frac{1}{2}}
\qty(\Tr(O O^{\dagger} O O^{\dagger} )/d)^{\frac{1}{2}}.
\end{align}
This proof is very similar to that of the expected average distance.
The fifth line is due to Cauchy-Schwarz inequality.
The sixth line comes from the mean inequality chain. 
The eighth line recruits the 1-design property of the Haar integral.
The last line is because $|\Tr(AB)|\le \sqrt{\Tr(AA^{\dagger})}\sqrt{\Tr(BB^{\dagger})}.$

We then consider the variance of the distances:
\begin{gather}
    \text{Var}(\mathscr{U}_0^r,\mathscr{U}^r)_{O,\mu_1}=V(\mathscr{U}_0^r,\mathscr{U}^r)_{O,\mu_1}-D(\mathscr{U}_0^r,\mathscr{U}^r)_{O,\mu_1}^2\leq 4r^2 \qty(\Tr(\mathscr{A} \mathscr{A}^\dagger \mathscr{A} \mathscr{A}^\dagger )/d)^{\frac{1}{2}}
\qty(\Tr(O O^{\dagger} O O^{\dagger} )/d)^{\frac{1}{2}}.
\end{gather}

\end{proof}

\begin{corollary}[Product-Formula Average Error]\label{co:productrandom1}
    A standard $p$th-order Suzuki product formula $\mathscr{S}_p$ has an average distance $\order{T_4\|O\|_4t^{p+1}d^{-1/2}r^{-p})}$  with
\begin{gather*}
    T_4\coloneqq\sum_{\gamma_1,\dots,\gamma_{p+1}=1}^\Gamma \left\|[H_{\gamma_{p+1}},[H_{\gamma_p},\dots,[H_{\gamma_2},H_{\gamma_{1}}]]] \right\|_4,
\end{gather*}
and $H=\sum_{\gamma=1}^\Gamma H_\gamma$ is the decomposition.
For $p=2$ case, we have a triangle-bound as
\begin{gather}
    D(\mathscr{U}_0^r,\mathscr{U}^r)_{O,\mu_1}\leq\frac{t^3\|O\|_4}{6r^2\sqrt{d}}\sum_{\gamma_1=1}^\Gamma\left\|\left[\sum_{\gamma_2=\gamma_1+1}^\Gamma H_{\gamma_2},\left[\sum_{\gamma_3=\gamma_1+1}^\Gamma H_{\gamma_3},H_{\gamma_1}\right]\right]\right\|_4+\frac{t^3\|O\|_4}{12r^2\sqrt{d}}\sum_{\gamma_1=1}^\Gamma\left\|\left[H_{\gamma_1},\left[H_{\gamma_1},\sum_{\gamma_2=\gamma_1+1}^\Gamma H_{\gamma_2}\right]\right]\right\|_4.
\end{gather}
\end{corollary}
\begin{proof}
    For the product-formula scenario, according to the definitions of additive errors $\mathscr{A}$ and multiplicative errors $\mathscr{M} $ of simulation, we have 
     \begin{gather}
         \|\mathscr{A}\|_4=\|\mathscr{M}\|_4.
     \end{gather}
Therefore, we can bound the distance by $\mathscr{M}$ instead.
\begin{gather}\label{eq:Mnorm}
    D(\mathscr{U}_0^r,\mathscr{U}^r)_{O,\mu_1}\leq\frac{2r\|\mathscr{M}\|_4\|O\|_4}{\sqrt{d}}
\end{gather}
We then follow a similar pathway to the proof of Theorem~8 in~\cite{zhao2022hamiltonian}:

We define the order of tuples $(\upsilon,\gamma)$, $\prec$ and $\preceq$, according to the lexicographic ordering.
According to Theorem 3 in Ref.~\cite{childs2021theory}, the multiplicative error $\mathscr{M}(\tau)$
can be expressed as
\begin{equation}
\mathscr{M}(\tau)=\mathrm{e}^{-\ii H\tau} \int_0^\tau \d\tau_1 \mathrm{e}^{\ii(\tau-\tau_1)H} \mathscr{S}_p(\tau_1) \mathscr{N}(\tau_1)=\int_0^\tau \d\tau_1 \mathrm{e}^{-\ii\tau_1H} \mathscr{S}_p(\tau_1) \mathscr{N}(\tau_1),
\end{equation}
where $\mathscr{S}_p$ is the $p$th-order Trotter formula we employed and
\begin{equation}
\begin{aligned}
 \mathscr{N}(\tau_1)= &\sum_{(\upsilon,\gamma)}
 \overrightarrow{\prod}_{(\upsilon',\gamma')\prec(\upsilon,\gamma)}
~ \mathrm{e}^{\tau_1 a_\upsilon H_{\pi_{\upsilon'}(\gamma')}} \left(a_\upsilon H_{\pi_\upsilon(\gamma)}\right)       \overleftarrow{\prod}_{(\upsilon',\gamma')\prec(\upsilon,\gamma)}
~ \mathrm{e}^{-\tau_1 a_\upsilon H_{\pi_{\upsilon'}(\gamma')}}  \\
&-  \overrightarrow{\prod}_{(\upsilon',\gamma')}
~ \mathrm{e}^{\tau_1 a_\upsilon H_{\pi_{\upsilon'}(\gamma')}} H \overleftarrow{\prod}_{(\upsilon',\gamma')}
~ \mathrm{e}^{-\tau_1 a_\upsilon H_{\pi_{\upsilon'}(\gamma')}},
\end{aligned}
\end{equation}
where $\mathscr{N}(\tau_1)=\mathcal O(\tau^p_1)$.
Here we define the vector $\vec{j}_{p+1}=(j_1,j_2,\dots,j_{p+1})$ with $p+1$ entries, $j_1,j_2,\dots,j_{p+1} \in \{(\upsilon,\gamma): \upsilon\in \{1,\dots,\Upsilon\},\gamma\in\{1,\dots,\Gamma\}\}$ and the corresponding nested commutators as
\begin{equation}
 N_{\vec{j}_{p+1}}=[H_{j_{p+1}},[H_{j_p},\dots,[H_{j_2},H_{j_{1}}]]].
\end{equation}
According to Theorem 5 in Ref.~\cite{childs2021theory}, we rewrite
 $\mathscr{N}(\tau_1)$ and $\mathscr{M}(\tau)$ as
\begin{equation}
\begin{aligned}
 \mathscr{N}(\tau_1)&= \sum_{i=1,2} \sum_{\vec{j}_{p+1}\in J_i} \int_0^{\tau_1}\d\tau_2       (\tau_1-\tau_2)^{q(\vec{j}_{p+1})-1}\tau_1^{p-q(\vec{j}_{p+1})} c_{\vec{j}_{p+1}} F_{\vec{j}_{p+1}}^{\dagger} N_{\vec{j}_{p+1}} F_{\vec{j}_{p+1}};\\
 \mathscr{M}(\tau)&=\int_0^\tau \d\tau_1 \int_0^{\tau_1}\d\tau_2
 \sum_{i=1,2} \sum_{\vec{j}_{p+1}\in J_i} (\tau_1-\tau_2)^{q(\vec{j}_{p+1})-1}\tau_1^{p-q(\vec{j}_{p+1})} c_{\vec{j}_{p+1}} R_{\vec{j}_{p+1}} N_{\vec{j}_{p+1}} F_{\vec{j}_{p+1}}.
 \end{aligned}
\end{equation}
Here $J_1$ and $J_2$ correspond to the first and second part in  $\mathscr{N}(\tau_1)$, respectively:
\begin{equation}
    \begin{aligned}
J_1&:=\{(j_1,j_2,\dots,j_{p+1}) : j_1\preceq j_2\preceq\dots \preceq j_{p+1}\},    \\
J_2&:=\{(j_1,j_2,\dots,j_{p+1}) : j_1\preceq j_2\dots \preceq j_{p}, \, j_{p+1}=(1,\gamma_{p+1})\} .
    \end{aligned}
\end{equation}
The $c_{\vec{j}_{p+1}}$ are real coefficients that are functions of $\vec{j}_{p+1}$ and $p$, satisfying $|c_{\vec{j}_{p+1}}|\le 1$. The function
$q(\vec{j}_{p+1})$
is the maximal number $q$ satisfying $j_1=j_2=\dots=j_q$. The $F_{\vec{j}_{p+1}}$ are unitary, constructed as products of terms of the form $\mathrm{e}^{\ii H_\gamma\tau_1}$, and $R_{\vec{j}_{p+1}}:= \mathrm{e}^{\ii(\tau-\tau_1)H} \mathscr{S}_p(\tau_1)F_{\vec{j}_{p+1}}^{\dagger}$.

Consequently, we find
\begin{align}
& \mathscr{M}(\tau)\mathscr{M}^{\dag}(\tau)\mathscr{M}(\tau)\mathscr{M}^{\dag}(\tau)\leq16\times\int_0^\tau \d\tau_{1,1} \d\tau_{2,1}\d\tau_{3,1}\d\tau_{4,1}\int_0^{\tau_{1,1}}\d\tau_{1,2}
\int_0^{\tau_{2,1}}\d\tau_{2,2}\int_0^{\tau_{3,1}}\d\tau_{3,2}\int_0^{\tau_{4,1}}\d\tau_{4,2}
   \notag\\
 &\ \sum_{\substack{\vec{j}_{1,p+1},\vec{j}_{2,p+1}\\\vec{j}_{3,p+1},\vec{j}_{4,p+1}}}
 \prod_{s=1}^4\left((\tau_{s,1}-\tau_{s,2})^{q(\vec{j}_{s,p+1})-1}\tau_{s,1}^{p-q(\vec{j}_{s,p+1})} c_{\vec{j}_{s,p+1}}\right)
 \notag\\
&\times R_{\vec{j}_{1,p+1}} N_{\vec{j}_{1,p+1}} F_{\vec{j}_{1,p+1}}
 F_{\vec{j}_{2,p+1}}^{\dagger} N_{\vec{j}_{2,p+1}}^{\dagger}R_{\vec{j}_{2,p+1}}^{\dagger}R_{\vec{j}_{3,p+1}} N_{\vec{j}_{3,p+1}} F_{\vec{j}_{3,p+1}}
 F_{\vec{j}_{4,p+1}}^{\dagger} N_{\vec{j}_{4,p+1}}^{\dagger}R_{\vec{j}_{4,p+1}}^{\dagger}.
 \end{align}
Iteratively utilizing the Cauchy-Schwarz inequality twice and noting that $R$ and $F$ are unitaries, we can bound the trace as
 \begin{align}\label{Eq:ENF}
    &\left |\Tr( R_{\vec{j}_{1,p+1}} N_{\vec{j}_{1,p+1}} F_{\vec{j}_{1,p+1}}
 F_{\vec{j}_{2,p+1}}^{\dagger} N_{\vec{j}_{2,p+1}}^{\dagger}R_{\vec{j}_{2,p+1}}^{\dagger}R_{\vec{j}_{3,p+1}} N_{\vec{j}_{3,p+1}} F_{\vec{j}_{3,p+1}}
 F_{\vec{j}_{4,p+1}}^{\dagger} N_{\vec{j}_{4,p+1}}^{\dagger}R_{\vec{j}_{4,p+1}}^{\dagger})\right|\notag\\
 &\ \le \|N_{\vec{j}_{1,p+1}}\|_4\|N_{\vec{j}_{2,p+1}}\|_4\|N_{\vec{j}_{3,p+1}}\|_4\|N_{\vec{j}_{4,p+1}}\|_4.
 \end{align}
We therefore have the upper bound
\begin{align}\label{Eq:mainENF}
&\Tr(\mathscr{M}(\tau)\mathscr{M}^{\dag}(\tau)\mathscr{M}(\tau)\mathscr{M}^{\dag}(\tau))\le 4\tau^{4p+4}\Upsilon^{4p+4}
\left[\sum_{\gamma_1,\dots,\gamma_{p+1}=1}^\Gamma\sqrt{\Tr(|[H_{\gamma_{p+1}},[H_{\gamma_p},\dots,[H_{\gamma_2},H_{\gamma_{1}}]]]|^4)}\right]^4
\end{align}
where $\Upsilon$ shows because each $H_\gamma$ could appear in $\Upsilon$ different stages, so there are $\Upsilon^{p+1}$ possibilities in total for each $[H_{\gamma_1},[H_{\gamma_2},\dots,[H_{\gamma_p},H_{\gamma_{p+1}}]]]$.
Our assumption that $p$ is a constant implies $\Upsilon^{4p+4}=\mathcal O(1)$.
Therefore, we can bound the average distance by adopting $\tau=t/r$,
\begin{equation}
 D(\mathscr{U}_0^r,\mathscr{U}^r)_{O,\mu_1}=\order{\frac{T_4\|O\|_4t^{p+1}}{\sqrt[4]{d}r^p}}.
\end{equation}

In the case of a standard second-order Suzuki-Trotter formula, we can further prove the triangle bound which offers a tighter estimation of the error.
We first consider the case where we only have two sub-Hamiltonians $H=A+B$, and $\mathscr{U}=\mathscr{S}_2=\mathrm{e}^{\ii A\tau/2}\mathrm{e}^{\ii B\tau}\mathrm{e}^{\ii A\tau/2}$.
Therefore, the additive error is
\begin{align}
    \mathscr{A}=&\int_0^\tau \d\tau_1\int_0^{\tau_1}\d\tau_2\int_0^{\tau_2} \d\tau_3 e^{\ii(\tau-\tau_1)H}e^{\ii\tau_1A/2}\notag\\
    &\cdot \left( e^{\ii\tau_3B}\left[\ii B,\left[\ii B,\ii\frac{A}{2}\right]\right]e^{-\ii\tau_3B}+ e^{-\ii\tau_3A/2}\left[-\ii\frac{A}{2},\left[-\ii\frac{A}{2},-\ii B\right]\right]e^{\ii\tau_3A/2}\right) e^{\ii\tau_1B} e^{\ii\tau_1A/2}.
\end{align}
Noticing that all the adjunct matrix exponentials here are unitary matrices, we can calculate and bound the 4-norm of $\mathscr{A}$ by Cauchy-Schwarz inequality as
\begin{align}
    \Tr(\mathscr{A}\mathscr{A}^\dag \mathscr{A}\mathscr{A}^\dag)\leq& \frac{\tau^{12}}{6^4}\left(\sqrt[4]{\Tr(|[\ii B,[\ii B,\ii\frac{A}{2}]]|^4)}+\sqrt[4]{\Tr(|[-\ii\frac{A}{2},[-\ii\frac{A}{2},-\ii B]]|^4)}\right)^4\notag\\
    =&\frac{\tau^{12}}{6^4}\left(\frac{1}{2}\sqrt[4]{\Tr(|[B,[B,A]]|^4)}+\frac{1}{4}\sqrt[4]{\Tr(|[A,[A,B]]|^4)}\right)^4.
\end{align}
Consequently, we have 
\begin{gather}
 \|\mathscr{A}\|_4\leq\frac{\tau^3}{12}\|[B,[B,A]]\|_4+ \frac{\tau^3}{24} \|[A,[A,B]]\|_4.
\end{gather}

As for a more general case where $H=\sum_{\gamma=1}^\Gamma H_\gamma$, we first define  
\begin{equation}
\mathscr{U}[\gamma_0]\coloneqq\mathrm{e}^{\ii H_{\gamma_0}\tau/2} \mathrm{e}^{\ii\sum_{\gamma=\gamma_0+1}^\Gamma H_{\gamma}\tau}\mathrm{e}^{\ii H_{\gamma_0}\tau/2},\quad \mathscr{U}_0[\gamma_0]\coloneqq\mathrm{e}^{\ii\sum_{\gamma=\gamma_0}^\Gamma H_{\gamma}\tau}.
\end{equation}
We then use the triangle inequality to decompose the additive error as
\begin{align}
    \|\mathscr{A}\|_4=\|\mathscr{U}_0-\mathscr{U}\|_4\leq\sum_{\gamma=1}^\Gamma\|\mathscr{U}_0[\gamma]-\mathscr{U}[\gamma]\|_4.
\end{align}
Note that all $\{\mathscr{U}[\gamma]\}$ are second-order product formulas with two sub-Hamiltonians.
Therefore, we can bound the 4-norm according to the preceding calculation:
\begin{align}
    \|\mathscr{A}\|_4\leq\frac{\tau^3}{12}\sum_{\gamma_1=1}^\Gamma\left\|\left[\sum_{\gamma_2=\gamma_1+1}^\Gamma H_{\gamma_2},\left[\sum_{\gamma_3=\gamma_1+1}^\Gamma H_{\gamma_3},H_{\gamma_1}\right]\right]\right\|_4+\frac{\tau^3}{24}\sum_{\gamma_1=1}^\Gamma\left\|\left[H_{\gamma_1},\left[H_{\gamma_1},\sum_{\gamma_2=\gamma_1+1}^\Gamma H_{\gamma_2}\right]\right]\right\|_4.
\end{align}
\end{proof}

We further show that the normalized Schatten 4-norms hold advantageous scaling compared to the operator norm for operators that consist of a summation over multiple Pauli operators.
\begin{proposition}\label{prop:summation}
    Consider a summation observable $O=\sum_{m=1}^MO_m$ where summands are Pauli operators with constant norms $\|O_m\|=\Theta(1)$ for all $m$.
    The normalized 4-norm of $O$ satisfies
    \begin{gather*}
        \Omega(\sqrt{M})\leq\frac{\|O\|_4}{\sqrt[4]{d}}\leq \order{\sqrt[4]{M^3}}.
    \end{gather*}
    By further assuming all summands $\{O_m\}$ are geometrically local with constant diameters in their supports.
    The normalized 4-norm can be bounded by
    \begin{gather*}
        \frac{\|O\|_4}{\sqrt[4]{d}}= \order{\sqrt{M}}.
    \end{gather*}
\end{proposition}
\begin{proof}
    Note that all Pauli operators except the identity are traceless and only products of two identical Pauli operators equal to the identity.
    Therefore, the trace of quartic of $O$ is equal to the summation over squares of coefficients in $OO^\dag$.
    We then consider the square, $OO^\dag=\sum_{m,m'=1}^MO_mO_{m'}^\dag$, which contains $M^2$ terms. 
    There must be at least $M$ distinct Pauli operators in this square.
    For this case, the 4-norm achieves the upper bound as each distinct Pauli operators in the square has the norm $\order{M}$.
    Therefore, 
    \begin{align}
        \frac{\|O\|_4}{\sqrt[4]{d}}=\sqrt[4]{\Tr(OO^\dag OO^\dag)/d}=\sqrt[4]{\sum_{m=1}^M\order{M^2}\Tr(I)/d}=\order{\sqrt[4]{M^3}}.
    \end{align}
    The lower bound can be reached by assuming all the resulting Pauli operators in the square are distinct except for the identity term, and we get
    \begin{align}
        \frac{\|O\|_4}{\sqrt[4]{d}}=\sqrt[4]{\Tr(OO^\dag OO^\dag)/d}\geq\sqrt[4]{((M^2-M)\cdot\Omega(1)+\Omega(M^2))\Tr(I)/d}=\Omega(\sqrt{M}).
    \end{align}   

    For the geometrically local operator, we similarly start the analysis from the square $OO^\dag$.
    It is clear to see that an arbitrary Pauli operator $P_\alpha$ (except for the identity) in $OO^\dag$ is also local with a constant size support.
    Suppose there exists a pair $(i,j)$ such that $O_iO_j\propto P_\alpha$.
    At least one of $O_i$ and $O_j$ intersects with $P_\alpha$, and we assume that is $O_i$ without loss of generality.
    Consequently, $S(O_j)$ must either also intersect with $P_\alpha$ or be a subset of $S(O_i)$.
    In either case, there are only $\order{1}$ choices of $O_i$ and $O_j$.
    Therefore, the coefficient associated to this $P_\alpha$ is bounded constantly.
    The resulting normalized 4-norm is
    \begin{gather}
        \frac{\|O\|_4}{\sqrt[4]{d}}\leq\sqrt[4]{M^2\cdot\order{1}\Tr(I)/d}=\order{\sqrt{M}}.
    \end{gather}
\end{proof}

\section{Applications}\label{sec:append-appli}
In this section, we meticulously examine the proposed analyses based on observables in simulating some renowned Hamiltonian models in quantum many-body physics, such as the nearest-neighbor (NN) lattice Hamiltonians and the power-law Hamiltonians. 
We provide detailed algorithm settings for simulating either a single local observable or a summation of observables under both Hamiltonians. 
We demonstrate that our proposals generally offer superior error scalings for both cases compared to the worst-case estimations presented in~\cite{childs2019nearly,childs2021theory}.

\subsection{Nearest-Neighbor Hamiltonians}\label{sec:NNH}
The typical Hamiltonian for an $n$-qubit $D$-dimensional lattice $\Lambda$ with nearest-neighbor (NN) interactions can be written as
\begin{gather}
    H=\sum_{(i,j)\in\Lambda}H_{i,j},\ \text{with }\|H_{i,j}\|_1\leq1\ \forall\,(i,j)\in\Lambda,
\end{gather}
where $(i,j)\in\Lambda$ means that $(i,j)$ is an edge in lattice $\Lambda$.
For an arbitrary local observable $O$ with support $S$ that can be bounded by a constant size convex hull, the edge sets of $\Lambda$ regarding $S$ and $H$ are 
\begin{gather}
    E_0^S=S,\ \ E_{k}^S=\{j\in\Lambda\,|\,j\notin E_{k-1}^S,\ \exists\,(i,j)\in\Lambda\ \text{s.t.}\ i\in E_{k-1}^S\},\ \text{  }k=1,2,\cdots.
\end{gather}
The corresponding interactive decomposition is 
\begin{gather}
    H_0^S=\sum_{\substack{(i,j)\in \Lambda\\i,j\in S}}H_{i,j},\ H_1^S=\sum_{\substack{(i,j)\in\Lambda\\i\in S,\,j\in E_1^S}}H_{i,j},\  H_{k}^S=\sum_{\substack{(i,j)\in\Lambda\\i,j\in E_{k-1}^S}}H_{i,j}+\sum_{\substack{(i,j)\in\Lambda\\i\in E_{k-1}^S,\,j\in E_{k}^S}}H_{i,j}.
\end{gather}
Since every qubit in $\Lambda$ only interacts with its $D$ neighbors, the 1-norm of $H_k^S$  scales with the cardinality of $E_{k-1}^S$.
\begin{gather}
    \|H_k^S\|_1=\order{|E_{k-1}^S|}=\order{k^{D-1}}.
\end{gather}

For the scenario of simulating a single observable, we employ the interactive decomposition and the even-odd permutation of the labels. The resulting $p$th-order Algorithm~\ref{alg:ap-rpf} yields an error given by Theorem~\ref{thm:ap-single}
\begin{gather}\label{eq:NN_local}
\epsilon=\|\mathrm{e}^{\ii Ht}O\mathrm{e}^{-\ii Ht}-UOU^\dag\|=\order{\frac{\|O\|t^{p+1}}{r^{p-D}}}.
\end{gather}
The corresponding step number for a fixed $\epsilon$ is given by $r=\order{\|O\|^{1/(p-D)}t^{(p+1)/(p-D)}\epsilon^{-1/(p-D)}}$.
Algorithm~\ref{alg:ap-rpf} exclusively implements relevant unitaries within the light cone. Consequently, the gate complexity of this simulation amounts to $\order{r^{D}\cdot r}$.

In the case of a summation of multiple local observables $O=\sum_{m=1}^MO_m$, the configuration depends on the regrouping of $H$.
As described in Sec.~\ref{sec:append-multi}, the regrouping maintains the initial representation $\sum_{(i,j)\in\Lambda}H_{i,j}$ since all Pauli terms are two-local and not covered by one another.
Consequently, the graph corresponds to the lattice $\Lambda$.
To color all edges, we partition them into $D$ axes.
Edges along each axis can be assigned with two colors according to their parities.
For example, in a two-dimensional lattice, edges are either vertical or horizontal, allowing us to color them with two colors for each axis, totaling four colors for this lattice.
The edge sets and interactive decomposition for each summand follow the same approach as the preceding analysis.
According to Theorem~\ref{thm:ap-multiple}, $p$th-order implementation of Algorithm~\ref{alg:ap-mpf} based on this configuration can achieve the error bound:
\begin{gather}\label{eq:NN_global}
    \epsilon=\|\mathrm{e}^{\ii Ht}O\mathrm{e}^{-\ii Ht}-UOU^\dag\|=\order{\frac{\sum_{m=1}^M\|O_m\|t^{p+1}}{r^{p-D}}}.
\end{gather}
For a fixed $\epsilon$, this bound implies $r=\order{(\sum_{m=1}^M\|O_m\|)^{1/(p-D)}t^{(p+1)/(p-D)}\epsilon^{-1/(p-D)}}$.
Since we do not delete those irrelevant unitaries and keep a faithful Suzuki-Trotter formula, the gate count is $\order{nr}$.

Particularly, we compare this gate count with that of trivially implementing Algorithm~\ref{alg:ap-rpf} for each summand observable separately for the magnetization $\sum_{j=1}^nZ_j$.
Suppose we allow for an error of $\epsilon_j$ for simulating $Z_j$.
According to the symmetry among summands, we have all $\epsilon_j$ the same.
According to Eq.~\eqref{eq:NN_local}, we have
\begin{gather}
    r_j=\order{\frac{t^{(p+1)/(p-D)}}{\epsilon_j^{1/(p-D)}}}=\order{\frac{n^{1/(p-D)}t^{(p+1)/(p-D)}}{\epsilon^{1/(p-D)}}}, \ \forall j\in[n].
\end{gather}
Therefore, all $r_j$'s scale the same as $r$ from Eq.~\eqref{eq:NN_global}, and we uniformly denote them as $r_0$.
The overall gate count for simulating summands separately is $\order{nr_0^{D+1}}$, and Algorithm~\ref{alg:ap-mpf} incurs much milder overheads.

As for the random-input simulation, we consider the benefits of the product formula in simulating a summation of observables as in Corollary~\ref{co:productrandom2}.
We start from the advantages of the 2-norm of the nested commutator of this NN Hamiltonian.
Here we use the similar decomposition with the previous chromatic decomposition that relies on the edge coloring.
\begin{proposition}\label{prop:T4}
Consider a $D$-dimensional $n$-qubit nearest-neighbor Hamiltonian $H$. 
The nested commutator $T_2$ from Corollary~\ref{co:productrandom2} is scaling as $T_2=\order{\sqrt{n}}$.
\end{proposition}
\begin{proof}
In this $D$-dimensional NN Hamiltonian $H$, there are altogether $\order{n}$ elementary interactions.
According to the edge coloring, we can decompose them into $2D$ sub-Hamiltonians, each of which consists of $\order{n}$ elementary interactions.

Consider a $p$-layer nested commutator,
$[H_{\gamma_{p+1}},[H_{\gamma_p},\dots,[H_{\gamma_2},H_{\gamma_{1}}]]]$.
We first narrow our focus on one nearest-neighbor interaction $H_{i,j}$ in $H_{\gamma_1}$ with $(i,j)\in\Lambda$.
To make the first commutator nonzero, only those interactions overlapping with $\{i,j\}$ in $H_{\gamma_2}$ can contribute to this commutator, and there are only $\order{1}$ choices since $D$ is some constant.
Moreover, the first commutator generates an operator with support on at most three neighboring qubits.
Repeating this analysis for all the nested commutators, we can find that there are only $\order{1}$ choices for each of $\{H_{\gamma_2},\cdots,H_{\gamma_{p+1}}\}$ to make commutator nonzero and that all nonzero commutators are some geometrically local operators of which the support covers $\{i,j\}$.

Enumerating $(i,j)$ in $H_{\gamma_1}$, the nested commutator generates $\order{n}$ geometrically local operators as $p$ and $D$ are constants.
Since all resulting operators must contain their own corresponding $\{i,j\}$'s, the overall Pauli decomposition of this nested commutator consists of $\order{n}$ Pauli terms each of which has coefficient $\order{1}$.
Consequently, the 2-norm is smaller than $\order{\sqrt{n}}$
Summing over all possible $\{\gamma_1,\cdots,\gamma_{p+1}\}$, $T_2$ scales as $\order{\sqrt{n}}$.

\end{proof}

Based on this, we consider the average error of simulating observables consisting of a summation of Pauli operators with even coefficients, $O=\sum_{m=1}^MO_m$.
We can easily calculate that $\|O\|_2/\sqrt{d}=\order{\sqrt{M}\max_m\|O_m\|}$.
According to Corollary~\ref{co:productrandom2} and~\ref{prop:T4}, the average distance of simulating by the $p$th product formula of the NN Hamiltonian is 
\begin{gather}
    D(\mathscr{U}_0^r,\mathscr{U}^r)_{O,\mu_1}=\order{\frac{\max_m\|O_m\|\sqrt{Mn}t^{p+1}}{r^p}}.
\end{gather}
For instance, given the observable is specifically chosen as the correlation function $O=\frac{1}{n-1}\sum_{j=1}^{n-1}P_jP_{j+1}$, the average distance is size-independent as
\begin{gather}
    D(\mathscr{U}_0^r,\mathscr{U}^r)_{O,\mu_1}=\order{\frac{t^{p+1}}{r^p}}.
\end{gather}

\subsection{Power-Law Hamiltonians}\label{sec:powerlaw}
A commonly studied family of Hamiltonians is the power-law model, where interaction intensities decay polynomially with distance.
Mathematically, it is represented as follows:
\begin{gather}
    H=\sum_{i,j\in\Lambda}H_{i,j},\text{  with }\|H_{i,j}\|\leq
    \begin{cases}
        1& i=j\\
        \order{1/\text{d}(i,j)^\alpha}& i\neq j
    \end{cases},
\end{gather}
where $\text{d}(i,j)$ denotes the geometric distance between qubits $i$ and $j$ on the lattice $\Lambda$.
Nevertheless, this Hamiltonian permits all-to-all interactions, leading to unlimited rapid expansion of light cones from arbitrary supports. 
To apply the support analysis for the short-time simulation discussed in previous sections, truncation of this Hamiltonian to short-range interactions is necessary.
To validate truncation, we further impose that $\alpha>2D$.

For the simulation of a single local observable $O$ with support $S$, the lattice is divided into inner and outer parts based on a parameter $d_0>0$.
Specifically,
\begin{gather}
    \Lambda_{in}\coloneqq\{j\in\Lambda\,|\,\text{d}(j,S)\leq(r\Upsilon+1)d_0\},\ \text{and }\Lambda_{out}\coloneqq\Lambda\backslash\Lambda_{in},
\end{gather}
where $r$ and $\Upsilon$ are steps and stages in the product formula, and $\text{d}(j,S)=\min_{i\in S}d(j,i)$.
We explicitly describe the truncated Hamiltonian as follows,
\begin{gather}
    H_{\text{lc}}\coloneqq H-\sum_{\substack{i\in\Lambda_{in},j\in\Lambda\\ \text{d}(i,j)>d_0}}H_{i,j}.
\end{gather}
The truncation error is bounded by accounting for the number of qubits in $\Lambda_{in}$ and all long-range effects, given by
\begin{gather}\label{eq:ap-lc}
    \|\mathrm{e}^{\ii Ht}-\mathrm{e}^{\ii H_{\text{lc}}t}\|\leq\|H-H_{\text{lc}}\|t\leq\order{\frac{(r\Upsilon+1)^Dd_0^Dt}{d_0^{\alpha-D}}}=\order{\frac{r^Dt}{d_0^{\alpha-2D}}}.
\end{gather}
This is based on the premise that the overall long-range interactions on a single qubit for any positive distance $x$ can be bounded as
\begin{gather}
    \sum_{j:\text{d}(i,j)\geq x}\|H_{i,j}\|=\order{\frac{1}{x^{\alpha-D}}}.
\end{gather}
The edge-set division in this scenario is 
\begin{gather}
    E_0^S=S,\ \ E_{r\Upsilon+2}^S=\Lambda_{out},\ \ \text{For }k=1,\cdots,r\Upsilon+1: E_k^S=\{j\in\Lambda\,|\,(k-1)d_0<\text{d}(S,j)\leq kd_0\}.
\end{gather}
The interactive decomposition is determined correspondingly.
In most cases, the 1-norm of each sub-Hamiltonian, $\|H_{i,j}\|_1$, scales similarly to its operator norm.
Therefore, we calculate the scaling by counting the number of qubits in each edge set, and $ h_k=\order{k^{D-1}d_0^D}$.
The overall errors from truncation and simulation by a $p$th-order formula in Algorithm~\ref{alg:ap-rpf} are bounded by
\begin{align}\label{eq:psb}
    \epsilon=\|\mathrm{e}^{\ii Ht}O\mathrm{e}^{-\ii Ht}-U_{\text{lc}}OU_{\text{lc}}^\dag\|&\leq\|\mathrm{e}^{\ii Ht}O\mathrm{e}^{-\ii Ht}-\mathrm{e}^{\ii H_{\text{lc}}t}O\mathrm{e}^{-\ii H_{\text{lc}}t}\|+\|\mathrm{e}^{\ii H_{\text{lc}}t}O\mathrm{e}^{-\ii H_{\text{lc}}t}-U_{\text{lc}}OU_{\text{lc}}^\dag\|\notag\\
    &\leq\order{\frac{r^D\|O\|t}{d_0^{\alpha-2D}}+\frac{\|O\|d_0^Dt^{p+1}}{r^{p-D}}},
\end{align}
where the second inequality comes from Theorem~\ref{thm:ap-single}.
To balance the two terms in Eq.~\eqref{eq:psb}, we find the minimizer $d_0=\order{\left(\frac{r}{t}\right)^{\frac{p}{\alpha-D}}}$, with minimum error given by
\begin{gather}
    \epsilon_{\text{min}}=\order{\|O\|\frac{t^{p+1-\frac{pD}{\alpha-D}}}{r^{p-D-\frac{pD}{\alpha-D}}}}.
\end{gather}

When simulating a summation of multiple local observables, $O=\sum_{m=1}^MO_m$, we adopt a general truncation that subtracts all long-range interactions beyond a certain distant $d_0$ from $H$,
\begin{gather}
    H_{\text{trc}}\coloneqq H-\sum_{\substack{i,j\in\Lambda\\\text{d}(i,j)>d_0}}H_{i,j}.
\end{gather}

To bound the truncation error, we prove a lemma for an arbitrary local observable.
\begin{lemma}
    For an arbitrary local observable $O_m$, the truncation error can be bounded by
    \begin{gather}
        \left\|\mathrm{e}^{\ii Ht}O_m\mathrm{e}^{-\ii Ht}-\mathrm{e}^{\ii H_{\text{trc}}t}O_m\mathrm{e}^{-\ii H_{\text{trc}}t}\right\|=\order{\frac{(t^{D+1}+t\log^D{n})\|O_m\|}{d_0^{\alpha-D}}},
    \end{gather}
    for short-time $t=o(\sqrt[D]{n})$.
\end{lemma}
\begin{proof}
    \begin{align}\label{eq:Lieb}
        \|\mathrm{e}^{\ii Ht}O_m\mathrm{e}^{-\ii Ht}-\mathrm{e}^{\ii H_{\text{trc}}t}O_m\mathrm{e}^{-\ii H_{\text{trc}}t}\|=&\left\|\int_{0}^tds\partial_s\left(\mathrm{e}^{\ii H(t-s)}\mathrm{e}^{\ii H_{\text{trc}}s}O_m\mathrm{e}^{-\ii H_{\text{trc}}s}\mathrm{e}^{-\ii H(t-s)}\right)\right\|\notag\\
        \leq&\int_0^tds\|[H-H_{\text{trc}},\mathrm{e}^{\ii H_{\text{trc}}s}O_m\mathrm{e}^{-\ii H_{\text{trc}}s}]\|.
    \end{align}
    
    Note that $H_{\text{trc}}$ is a two-local short-range Hamiltonian satisfying
    \begin{gather}
        \sum_{j:\text{d}(j,i)\leq d_0}\|H_{i,j}\|\exp{\mu \text{d}(j,i)}\leq C
    \end{gather}
    for all qubit $i$ in the lattice $\Lambda$ with $\mu=1/d_0$ and $C=\order{1}$.
    Therefore, the Lieb-Robinson bound in~\cite{hastings2010locality} is applicable for the truncated interaction $H_{\text{trc}}$.

    To bound the commutator in Eq.~\eqref{eq:Lieb}, we divided the truncated terms into two parts: the internal part of terms within $d_1$-distance to $S(O_m)$, and the outside part of all other terms. 
    According to the locality of $O_m$, we have
    \begin{align}
        \|[H-H_{\text{trc}},\mathrm{e}^{\ii H_{\text{trc}}s}O_m\mathrm{e}^{-\ii H_{\text{trc}}s}]\|=\order{\frac{\|O_m\|d_1^D}{d_0^{\alpha-D}}+\exp{-\mu d_1+Cs+\log{n}}}.
    \end{align}
    Choosing $d_1=\order{d_0t+d_0\log{n}}$, the first term is dominant, and we can bound Eq.~\eqref{eq:Lieb} by
    \begin{gather*}
        \|\mathrm{e}^{\ii Ht}O_m\mathrm{e}^{-\ii Ht}-\mathrm{e}^{\ii H_{\text{trc}}t}O_m\mathrm{e}^{-\ii H_{\text{trc}}t}\|=\order{\frac{(t^{D+1}+t\log^D{n})\|O_m\|}{d_0^{\alpha-2D}}}.
    \end{gather*}
\end{proof}
Consequently, we can bound the overall truncation error from all summands.
\begin{align}
    \|\mathrm{e}^{\ii Ht}O\mathrm{e}^{-\ii Ht}-\mathrm{e}^{\ii H_{\text{trc}}t}O\mathrm{e}^{-\ii H_{\text{trc}}t}\|\leq\sum_{m=1}^M\|\mathrm{e}^{\ii Ht}O_m\mathrm{e}^{-\ii Ht}-\mathrm{e}^{\ii H_{\text{trc}}t}O_m\mathrm{e}^{-\ii H_{\text{trc}}t}\|=\order{\frac{(t^{D+1}+t\log^D{n})\sum_{m=1}^M\|O_m\|}{d_0^{\alpha-2D}}}.
\end{align}

To handle the summation case, we attempt to regroup $H_{\text{trc}}$ as 
in Sec.~\ref{sec:append-multi}, which maintains the two-local representation.
However, the regrouping requires a significant number of colors for edge coloring in its corresponding interaction graph.
To address this challenge, we propose a method to drastically reduce the number of colors while maintaining the mild scaling of sizes of edge sets.

We begin by partitioning $\Lambda$ into disjoint $D$-dimensional cubes $\{C_i\}$ with a side length $d_0$.
Each cube has $3^D-1$ neighbors within a distance less than $d_0$.
Our new regrouping strategy for $H_{\text{trc}}$ ensures that each sub-Hamiltonian supports on $C_i\cup C_j$, where $C_i$ and $C_j$ are neighboring cubes.
With this regrouping, the corresponding hyperedges can be colored by $3^D-1$ colors.

While the cube-based regrouping is inconsistent with the standard approach outlined in Sec.~\ref{sec:append-multi} as no Pauli terms in $H_{\text{trc}}$ locate exactly on supports of sub-Hamiltonians, we address this discrepancy by introducing "illusory" Pauli terms on those supports with vanishing coefficients. 
Consequently, Lemma~\ref{lm:ap-global_support} remains applicable as an upper bound on the expanding support of an arbitrary $O_m$ during simulation by Algorithm~\ref{alg:ap-mpf} with the new regrouping and coloring permutation.
As a result, we continue to consider the first $r\Upsilon(3^D-2)+2$ edge sets for $O_m$, albeit with edge sets differing from those under Pauli regrouping.
For any constant-size support $S$, it can be demonstrated that the cube-based edge set $E_k^S$ covers $\order{k^{D-1}d_0^D}$ qubits, a scale identical to that of the initial edge sets.
Correspondingly, the 1-norm of $H_k^S$ under the cube-based regrouping is linear with the qubit numbers, $h_k=\order{k^{D-1}d_0^D}$.
Therefore, we obtain the total simulation errors of a $p$th-order implementation of Algorithm~\ref{alg:ap-mpf} as follows:
\begin{align}
    \epsilon=\|\mathrm{e}^{\ii Ht}O\mathrm{e}^{-\ii Ht}-U_{\text{trc}}OU_{\text{trc}}^\dag\|&\leq\|\mathrm{e}^{\ii Ht}O\mathrm{e}^{-\ii Ht}-\mathrm{e}^{\ii H_{\text{trc}}t}O\mathrm{e}^{-\ii H_{\text{trc}}t}\|+\|\mathrm{e}^{\ii H_{\text{trc}}t}O\mathrm{e}^{-\ii H_{\text{trc}}t}-U_{\text{trc}}OU_{\text{trc}}^\dag\|\notag\\
    &\leq\order{\frac{(t^{D+1}+t\log^D{n})\sum_{m=1}^M\|O_m\|}{d_0^{\alpha-2D}}+\frac{\sum_{m=1}^M\|O_m\|d_0^Dt^{p+1}}{r^{p-D}}}.
\end{align}
We have the minimizer $d_0=\order{r^{(p-D)/(\alpha-D)}t^{(D-p)/(\alpha-D)}+\text{polylog}(n)}$.
The minimum error is
\begin{gather}
    \epsilon_{\text{min}}=\tilde{\mathcal{O}}\left(\sum_{m=1}^M\|O_m\|\frac{t^{\frac{p(\alpha-2D)+\alpha-D+D^2}{\alpha-D}}}{r^{\frac{(p-D)(\alpha-2D)}{\alpha-D}}}\right).
\end{gather}
The $\tilde{\mathcal{O}}$ omits the poly-log terms since in most cases we have $t\ll\text{polylog}(n)$.

\begin{figure*}[tb]
    \centering
    \subfloat[]{\includegraphics[width=.48\linewidth]{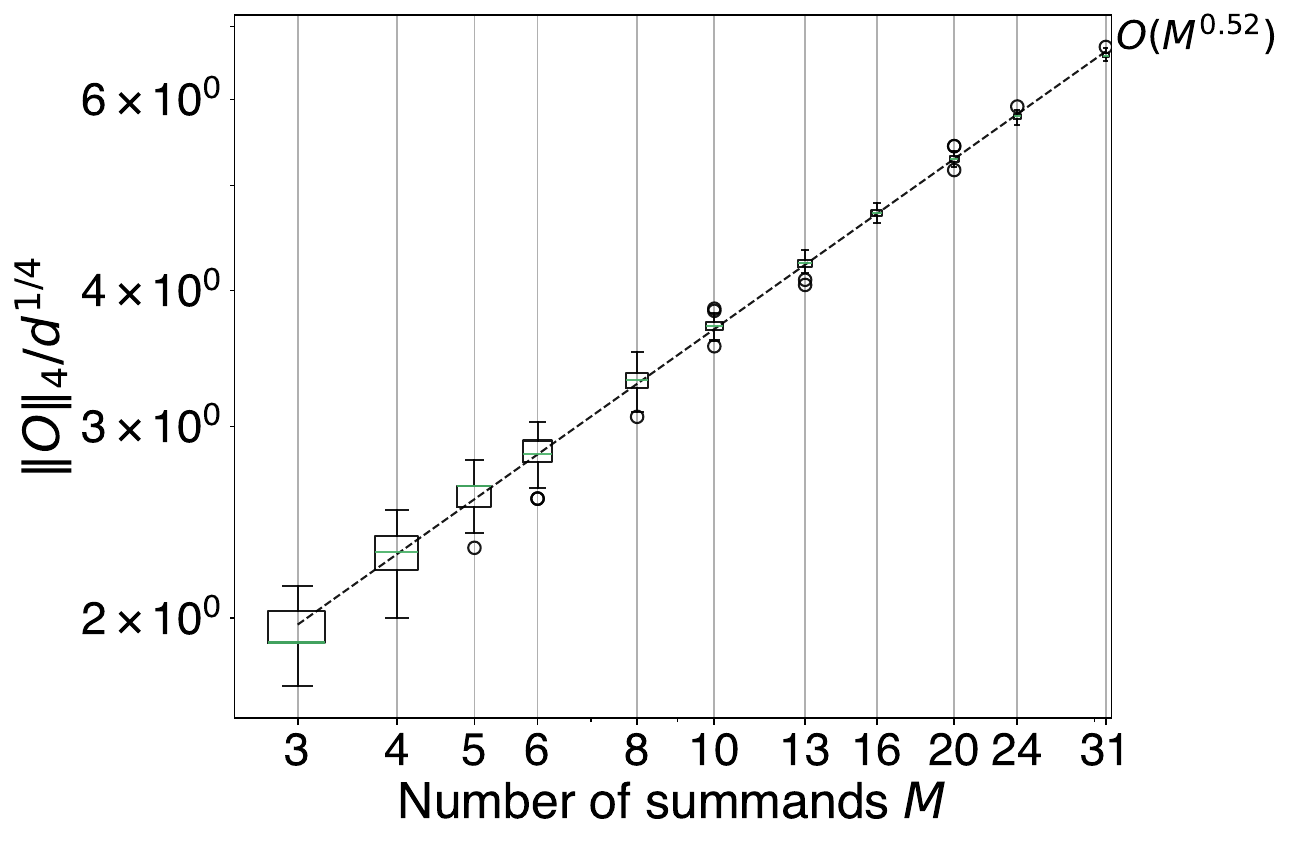}}
    \subfloat[]{\includegraphics[width=.48\linewidth]{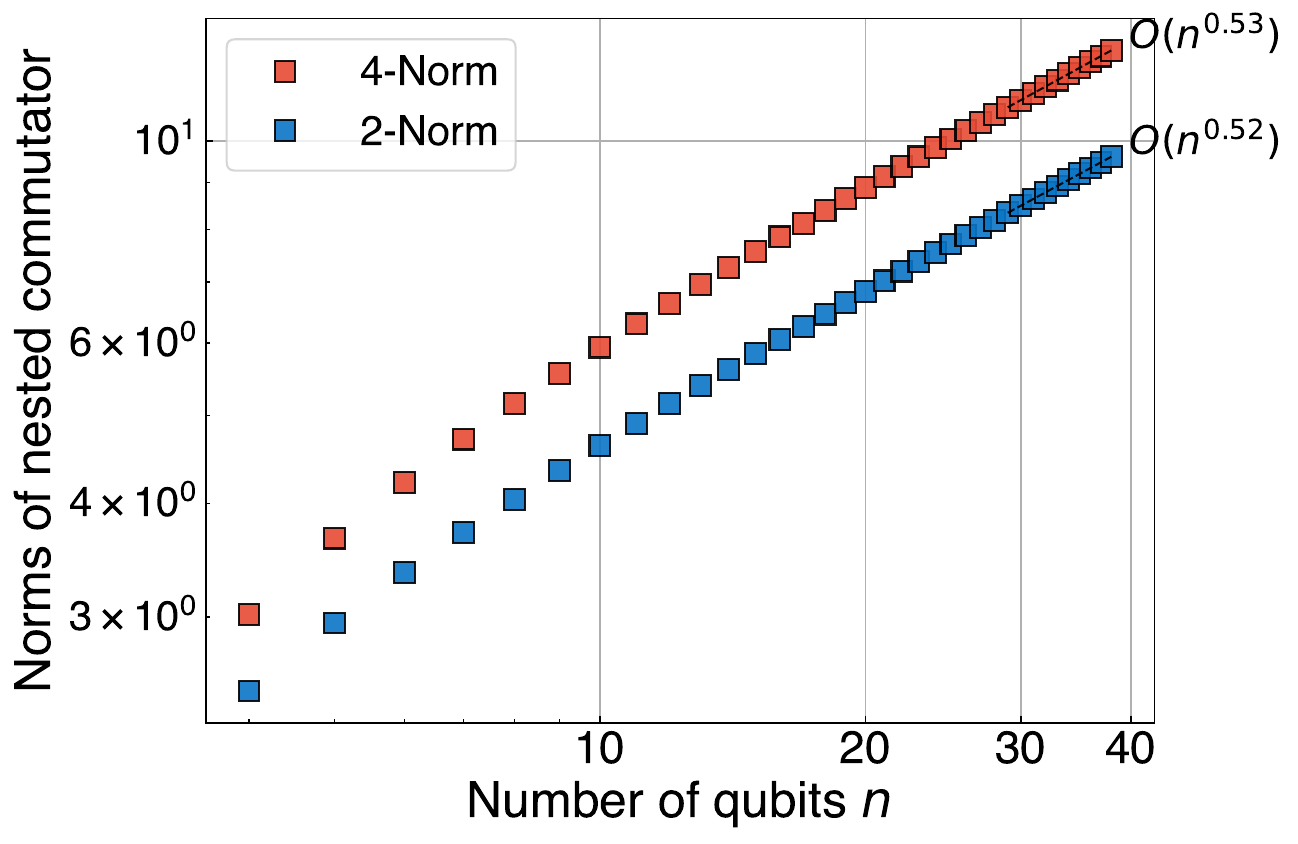}}
    \caption{The numerical evidence for scaling properties of 4-norms. \textbf{(a)} 4-norms of the $50$-qubit summation observables $O=\sum_{m=1}^MO_m$.
    For each number of summands $M$, we sample 100 independent sets of $\{O_m\}$
    of $m$ randomly sampled Pauli operators. \textbf{(b)} Different norms of the nested commutators of power-law Hamiltonian in Eq.~\eqref{eq:powerlaw}. 
    We calculate the 2-norms and 4-norms of the nested commutator comes from the product-formula errors.} 
    \label{fig:4norm}
\end{figure*}

For the random-input simulation, we use the inequalities among different norms for an arbitrary matrix $A$
\begin{gather}
    \frac{\|A\|_2}{\sqrt{d}}\leq\|A\|.
\end{gather}
Therefore, we bound the nested commutator $T_2=\order{n}$ according to the analysis in \cite{childs2021theory}.
It is still noteworthy that empirically this $T_2$ can be quadratically better than this bound, as shown in the next section.
The major advantage then comes from the observable $\|O\|_2/\sqrt{d}$.
According to the same analysis, we can bound the Schattern 2-norm for an evenly distributed Pauli-summation observable by $\order{\max_m\|O_m\|\sqrt{M}}$.

\section{Numerical Details}\label{sec:append-num}
In our numerical results, we often compare results among different error analyses.
Here, we specify the bounds that are used in our numerical evaluations.
We start with the worst-case bound.
\begin{proposition}[Restatement of Prop.~10 in~\cite{childs2021theory}]\label{prop:tightbound}
    Let $H=\sum_{\gamma=1}^\Gamma H_{\gamma}$ be the decomposition.
    The additive error of a second-order Suzuki-Trotter formula to simulate $O$ for time $t$ with step $r$ can be bounded as
    \begin{align}
        \frac{t^3\|O\|}{6r^2} \Bigg(\sum_{\gamma_1=1}^{\Gamma-1}
        \norm{\qty[\sum_{\gamma_3=\gamma_1+1}^\Gamma H_{\gamma_3},\qty[\sum_{\gamma_2=\gamma_1+1}^\Gamma H_{\gamma_2},H_{\gamma_1}]]}+
        \frac{1}{2} \sum_{\gamma_1=1}^{\Gamma-1}
        \norm{\qty[H_{\gamma_1},\qty[H_{\gamma_1}, \sum_{\gamma_2=\gamma_1+1}^\Gamma H_{\gamma_2}]]}\Bigg).
        \label{eq:tight_bound}
    \end{align}
\end{proposition}
\noindent For our short-time bounds, we can also derive the explicit versions based on Proposition~\ref{prop:tightbound} as:
\begin{proposition}[Explicit Version of Thm.~\ref{thm:ap-single}]
    Consider a local observable $O$ with support $S$.
Suppose the $n$-qubit $H$ is $\ell$-local with a constant $\ell$ and 
has bounded interaction per qubit. 
The simulation error of $O$ by an $r$-step second-order $U$ from Alg.~\ref{alg:ap-rpf} is bounded by 
    \begin{align}
        \|\mathrm{e}^{\ii Ht}O\mathrm{e}^{-\ii Ht}-U OU^\dagger\|\leq\frac{t^3\|O\|}{6r^2}\Bigg(\sum_{k_1=1}^{r\Upsilon}
        \norm{\qty[\sum_{k_3=k_1+1}^{r\Upsilon+1} H^S_{k_3},\qty[\sum_{k_2=k_1+1}^{r\Upsilon+1} H^S_{k_2},H^S_{k_1}]]}+
        \frac{1}{2} \sum_{k_1=1}^{r\Upsilon}\norm{\qty[H^S_{k_1},\qty[H^S_{k_1}, \sum_{k_2=k_1+1}^{r\Upsilon+1} H^S_{k_2}]]}\Bigg),
        \label{eq:tight_bound_single}
    \end{align}
    where we adopt the decomposition in Def.~\ref{def:step_edge} and $H^S_{r\Upsilon+1}$ consists all other terms in $H$.
\end{proposition}
\begin{proof}
    The proof is simply from replacing Eq.~\eqref{eq:rpf_trotter} by Proposition~10 in~\cite{childs2021theory}.
\end{proof}

\begin{proposition}[Explicit Version of Thm.~\ref{thm:ap-multiple}]
Consider the summation $O=\sum_{m=1}^MO_m$ where every $O_m$ is a local observable with support $S(O_m)$.
Suppose the $n$-qubit $H$ is $\ell$-local with a constant $\ell$ and 
has bounded interaction per qubit. 
The simulation error of $O$ by an $r$-step second-order $U$ from Alg.~\ref{alg:ap-mpf} is bounded by 
    \begin{align}
        \|\mathrm{e}^{\ii Ht}O\mathrm{e}^{-\ii Ht}-U OU^\dagger\|\leq\sum_{m=1}^M\frac{t^3\|O_m\|}{6r^2}&\Bigg(\sum_{c_1=1}^{\chi}
        \norm{\qty[\sum_{c_3=c_1+1}^{\chi+1} H_{c_3,m},\qty[\sum_{c_2=c_1+1}^{\chi+1} H_{c_2,m},H_{c_1,m}]]}\notag\\
        &+
        \frac{1}{2} \sum_{c_1=1}^{\chi}\norm{\qty[H_{c_1,m},\qty[H_{c_1,m}, \sum_{c_2=c_1+1}^{\chi+1} H_{c_2,m}]]}\Bigg),
        \label{eq:tight_bound_multi}
    \end{align}
    where for each $m$ $\{H_{c,m}\}_c$ is a decomposition with $H_{c,m}=\sum_{\substack{\gamma:\varphi(S_\gamma)=c\\S_\gamma\cap\bigcup_{k=0}^{r(\chi-1)\Upsilon+1}E_k^{S(O_m)}\neq\emptyset}}H_{S_\gamma}$ for $c\in[1,\chi]$ and $H_{\chi+1,m}$ consists all other terms in $H$.
\end{proposition}
\begin{proof}
    The proof is similar.
    Replace the asymptotic Eq.~\eqref{eq:errorbound2} by Proposition~10 in~\cite{childs2021theory}, and notice that all the errors in former steps are smaller than that of the $r$th step.
    We can complete the proof.
\end{proof}
For the case of large $n$, it is hard to calculate the operator norms $\|\cdot\|$ as in these propositions.
To figure this out, we replace all the operator norms with the 1-norms $\|\cdot\|_1$, which calculate the sum of absolute values of Pauli coefficients of the operator.
This calculation is scalable according to the Pauli algebra.

For Figures 3(a) and 3(b) in the main text, we adopt the mixed-field Ising Hamiltonian and transverse-field Ising Hamiltonian, respectively.
The worst-case bounds in these figures are calculated from Proposition~\ref{prop:tightbound} by decomposing the Hamiltonians according to Pauli stabilizer groups of their interactions ($X$-$Y$-$Z$ decomposition).
This is a commonly used decomposition in the standard implementation of the product formulas.
Our bounds employ the edge-set and chromatic decompositions as specified in the requirements of Theorem~\ref{thm:ap-single} and~\ref{thm:ap-multiple}.

For Figure 4 in the main text, we compare the simulation distances between the worst-case bound in Proposition~\ref{prop:tightbound}, the following Proposition~\ref{prop:randomwo} without observable knowledge, and our average-error bound in Corollary~\ref{co:productrandom2}.
\begin{proposition}[Restatement of Thm.~3 in~\cite{zhao2022hamiltonian}]\label{prop:randomwo}
    For the second-order product formula, we get the average simulation distance by 
\begin{gather*}
    D(\mathscr{U}_0^r,\mathscr{U}^r)_{O,\mu_1}\leq\frac{t^3\|O\|}{6r^2\sqrt{d}}\left(\sum_{\gamma_1=1}^\Gamma\left\|\left[\sum_{\gamma_2=\gamma_1+1}^\Gamma H_{\gamma_2},\left[\sum_{\gamma_3=\gamma_1+1}^\Gamma H_{\gamma_3},H_{\gamma_1}\right]\right]\right\|_2+\frac{1}{2}\sum_{\gamma_1=1}^\Gamma\left\|\left[H_{\gamma_1},\left[H_{\gamma_1},\sum_{\gamma_2=\gamma_1+1}^\Gamma H_{\gamma_2}\right]\right]\right\|_2\right),
\end{gather*}
with $H=\sum_{\gamma=1}^\Gamma H_\gamma$ is the decomposition.
\end{proposition}
For Figure~4(a), we choose the $X$-$Y$-$Z$ decomposition of the power-law Hamiltonian for all these three bounds.
For Figure~4(b), since the Hamiltonian is a complicated molecular one consisting of some anti-commutative Pauli components, we first execute the greedy method to group the Pauli decomposition into different commutative sets.
This results in a decomposition of the Hamiltonian, which is employed for implementing both the product formulas and the measurements.

\section{Additional Numerical Results}\label{sec:append-additional}
In Proposition~\ref{prop:summation}, we have shown that the 4-norm of summation observables $O=\sum_{m=1}^MO_m$ with uniform coefficients can be upper and lower bounded by $\order{\sqrt[4]{M^3}}$ and $\Omega(\sqrt{M})$.
Nevertheless, we find that 4-norms in this case are most likely scaling similarly to the lower bound $\Theta(\sqrt{M})$, as shown in Fig.~\ref{fig:4norm}(a).
Numerically, we randomly sample 100 sets of $M$ 50-qubit Pauli operators for each $M$ as the summands and uniformly assign coefficients to be 1.
By calculating the normalized 4-norms of these constructed observables, we can form a box to represent the concentration of these 4-norms for every $M$.
To study the scaling properties of the 4-norm of a randomly chosen $O$, we fit medians of all boxes along the increasing number $M$.
This suggests the potential improvements for practical cases from the upper bound.

We further explore the numerical behavior of the 2-norms and 4-norms of nested commutators as implied from Corollary~\ref{co:productrandom2} and~\ref{co:productrandom1}.
As illustrated in the preceding section, we find the scaling of the 2- and 4-norms of the nested commutators for simulating power-law Hamiltonians challenging to study theoretically. 
Nevertheless, we can get some intuition by viewing a specific case through numerical results.
Here we adopt a typical one-dimensional power-law Hamiltonian $H$ as
\begin{gather}\label{eq:powerlaw}
    H=\sum_{i=1}^n\sum_{j=i+1}^n\frac{J}{(j-i)^\alpha}X_iX_j+h\sum_{i=1}^nZ_i,
\end{gather}
with $J=1$, $h=0.2$, and $\alpha=4$.
We decompose this Hamiltonian according to the Pauli stabilizers as
\begin{gather}
    H_1=\sum_{i=1}^n\sum_{j=i+1}^n\frac{J}{(j-i)^\alpha}X_iX_j,\ H_2=h\sum_{i=1}^nZ_i.
\end{gather}
We then calculate the norms of the nested commutator terms as 
\begin{gather}
\|[H_2,[H_2,H_1]]\|_2+\frac{1}{2}\|[H_1,[H_2,H_1]]\|_2,\\
    \|[H_2,[H_2,H_1]]\|_4+\frac{1}{2}\|[H_1,[H_2,H_1]]\|_4.
\end{gather}
As shown in Fig.~\ref{fig:4norm}(b), the numerical result reveals that the 2-norm and 4-norm of the nested commutator are scaling closely to $\Theta(\sqrt{n})$, which sheds light on the potential empirical improvements from the trivial $\order{n}$ bound.

\end{document}